\documentclass[a4paper,11pt]{extarticle}
\usepackage[left=1.5cm,right=1.5cm,top=1.5cm,bottom=2.5cm]{geometry}

\usepackage{amsmath,mathtools,amssymb,amsthm,xcolor,url,esvect,paralist}
\usepackage{parskip}
\usepackage{tikz}
\usetikzlibrary{arrows,arrows.meta,shapes,calc,decorations}
\usetikzlibrary{decorations.pathmorphing}

\usepackage{mathrsfs}

\usepackage{subcaption}

\usepackage{circledsteps}
\usepackage{bm}

\usepackage{accents}
\DeclareMathAccent{\wtilde}{\mathord}{largesymbols}{"65}

\usepackage{verbatim}

\usepackage{hyperref}

\usepackage{relsize}

\usepackage[most]{tcolorbox}

\swapnumbers
\numberwithin{equation}{section}

\newtheorem{theorem/}[equation]{Theorem}
\newenvironment{theorem}
  {%
   \pushQED{\qed}\begin{theorem/}}
  {\popQED\end{theorem/}}

\newtheorem{lemma/}[equation]{Lemma}
\newenvironment{lemma}
  {%
   \pushQED{\qed}\begin{lemma/}}
  {\popQED\end{lemma/}}

\newtheorem{claim/}[equation]{Claim}
\newenvironment{claim}
  {%
   \pushQED{\qed}\begin{claim/}}
  {\popQED\end{claim/}}

\newtheorem{corollary/}[equation]{Corollary}
\newenvironment{corollary}
  {%
   \pushQED{\qed}\begin{corollary/}}
  {\popQED\end{corollary/}}

\newtheorem{proposition/}[equation]{Proposition}
\newenvironment{proposition}
  {%
   \pushQED{\qed}\begin{proposition/}}
  {\popQED\end{proposition/}}

\theoremstyle{definition}
\newtheorem{definition/}[equation]{Definition}
\newenvironment{definition}
  {%
   \pushQED{\qed}\begin{definition/}}
  {\popQED\end{definition/}}
\newtheorem{example/}[equation]{Example}
\newenvironment{example}
  {%
   \pushQED{\qed}\begin{example/}}
  {\popQED\end{example/}}
\newtheorem{remark/}[equation]{Remark}
\newenvironment{remark}
  {%
   \pushQED{\qed}\begin{remark/}}
  {\popQED\end{remark/}}
\newtheorem*{example*}{Example}
\newtheorem{caveat/}[equation]{Caveat}

\newtheorem*{caveat*}{Caveat}
\newtheorem{conjecture/}[equation]{Conjecture}
\newenvironment{conjecture}
  {%
   \pushQED{\qed}\begin{conjecture/}}
  {\popQED\end{conjecture/}}
\newtheorem*{conjecture*}{Conjecture}
\newtheorem{openproblem/}[equation]{Open Problem}

\newtheorem*{openproblem*}{Open Problem}

\newcommand{\ckt}{\textup{\textsf{ckt}}}
\newcommand{\supp}{\textup{\textsf{supp}}}
\newcommand{\suppdash}{\textup{\textsf{supp-}}}
\newcommand{\poly}{\textup{\textsf{poly}}}

\newcommand{\IC}{\mathbb{C}}
\newcommand{\IR}{\mathbb{R}}

\newcommand{\IF}{\mathbb{F}}
\newcommand{\la}{\lambda}
\newcommand{\IN}{\mathbb{N}}
\newcommand{\aS}{\mathfrak{S}}
\newcommand{\cE}{\End}                  

\newcommand{\GL}{\textup{GL}}
\newcommand{\End}{\textup{End}}

\newcommand{\diag}{\textup{diag}}
\newcommand{\tensor}{{\textstyle\bigotimes}}

\newcommand{\Clique}{\mathrm{Clique}}
\newcommand{\Cliqued}{\mathrm{Cliq}}
\newcommand{\Wedge}{{\mathchoice{\textstyle\bigwedge}{\textstyle\bigwedge}{\bigwedge}{\bigwedge}}}
\newcommand{\sS}{\mathscr S}
\newcommand{\sW}{\mathscr W}
\newcommand{\sT}{\mathscr T}
\newcommand{\id}{\mathrm{id}}
\newcommand{\per}{\mathrm{per}}
\newcommand{\pol}{\mathrm{polar}}
\newcommand{\res}{\mathrm{resti}}
\newcommand{\mon}{\mathrm{mon}}

\newcommand{\HC}{\mathrm{HC}}
\renewcommand{\det}{\mathrm{det}}
\newcommand{\gl}{\mathfrak{gl}}

\newcommand\restr[2]{{
  \left.\kern-\nulldelimiterspace 
  #1 
  \vphantom{\big|} 
  \right|_{#2} 
  }}
\newcommand{\Hom}{\mathrm{Hom}}
\newcommand{\rk}{\mathrm{rk}}

\newcommand{\hGamma}{\textup{\textsf{h}}\ensuremath{\Gamma}}
\newcommand{\hXi}{\textup{\textsf{h}}\ensuremath{\Xi}}
\newcommand{\trXi}{\ensuremath{\textrm{\textup{tr-{\ensuremath{\Xi}}}}}}
\newcommand{\transpose}{\textup{\textsf{T}}}

\newcommand{\w}{\ensuremath{w}}
\newcommand{\asize}{\textup{\textsf{ac}}} 
\newcommand{\csize}{\textup{\textsf{cc}}} 
\newcommand{\fsize}{\textup{\textsf{fc}}} 
\newcommand{\bsize}{\textup{\textsf{bc}}} 

\newcommand{\Sym}{\ensuremath{S}}
\newcommand{\Symt}{\textup{Sym}}

\DeclareMathOperator{\tr}{tr}

\newcommand{\cpc}{\chi} 

\newcommand{\Mat}{\mathrm{Mat}}
\newcommand{\Var}{\mathrm{Var}}

\newcommand{\pow}{\mathrm{pow}}

\newcommand{\sym}{\textup{sym}}
\newcommand{\stab}{\textup{stab}}

\newcommand{\AND}{\mathbin{\texttt{AND}}}
\newcommand{\OR}{\mathbin{\texttt{OR}}}
\newcommand{\NOT}{\mathop{\texttt{NOT}}}

\newcommand{\qu}[1]{\ensuremath{\overset{\textup{\tiny ?}}{#1}}}

\usepackage{xparse}
\ExplSyntaxOn
\NewDocumentCommand{\nsp}{m}
 {
   \clist_map_inline:nn { #1 }
     {
       \mbox{$##1$}
     }
 }
\ExplSyntaxOff

\newcommand{\VFPT}{\mathrm{VFPT}}
\newcommand{\VBFPT}{\mathrm{VBFPT}}
\newcommand{\VBP}{\mathrm{VBP}}
\newcommand{\VNP}{\mathrm{VNP}}
\newcommand{\VW}{\mathrm{VW}}
\newcommand{\W}{\mathrm{W}}

\makeatletter
\newcommand{\multimapinv}{%
  \mathrel{\mathpalette\multimapinv@{}}%
}

\newcommand{\multimapinv@}[2]{%
  \raisebox{0pt}{$#1\reflectbox{$#1\multimap$}$}%
}
\makeatother
\newcommand{\lefthalf}{\ensuremath{\multimap}}
\newcommand{\righthalf}{\ensuremath{\multimapinv}}

\title{Kronecker products and iterated matrix multiplication}

\author{Christian Ikenmeyer}
\date{June 6, 2026}

\begin{document}
\raggedbottom
\maketitle

\begingroup
  \renewcommand\thefootnote{}
  \footnotetext{For the purpose of open access, the author has applied a Creative Commons Attribution (CC-BY) license to any Author Accepted Manuscript version arising from this submission.}%
\endgroup

\begin{abstract}
We observe that the Kronecker product of tensors is the operation that converts the determinant polynomial into Cayley's first hyperdeterminant.
We apply the Kronecker product to iterated matrix multiplication,
which results in the hypercomputant, a $\VNP$-complete and $\VW[1]$-complete polynomial
whose hardness we prove via the equivariance of the Kronecker product.
The construction works over arbitrary commutative semirings
and also for the tensor algebra and the exterior algebra.
For the tensor algebra this gives a version of ``noncommutative $\VNP$'',
and for polynomials over the nonnegative real numbers this gives a version of ``monotone $\VNP$'',
each with the hypercomputant as the complete object.
We take a parameterized complexity viewpoint and
compare the noncommutative setting and the monotone setting.
Using standard techniques we obtain optimal algebraic branching program width lower bounds in both settings,
and these are notably not always the same.
We also prove the polystability of the hypercomputant and that its isotypic components are characterized by their stabilizer.
\end{abstract}

{
\footnotesize
\begin{tabular}{@{}p{2cm}p{0.85\textwidth}@{}}
\textbf{Keywords:} &
Kronecker product, iterated matrix multiplication, symmetric tensors, Valiant's conjecture, parameterized complexity, noncommutative computation, monotone computation, geometric complexity theory
\end{tabular}
}

\setcounter{tocdepth}{1}
\tableofcontents

\section{Introduction}
Valiant's determinant versus permanent conjecture is posed only over fields of characteristic different from 2,
because otherwise the determinant equals the permanent.
Hence, no characteristic free lower bounds method can separate the determinant from the permanent.
This can be avoided by replacing the permanent with the Hamiltonian cycle polynomials,
but in this document we propose a more algebraic approach.
We work over arbitrary commutative semirings for introducing the main objects of interest,
the iterated matrix product (as a matrix of polynomials) and its Kronecker square, the hypercomputant, see \S\ref{sec:multilin}.
Many of our constructions are general and also work for the tensor algebra and the exterior algebra.
For the tensor algebra this gives a version of ``noncommutative $\VNP$'' (\S\ref{subsec:noncomm}),
and for polynomials over the nonnegative real numbers this gives a version of ``monotone $\VNP$'' (\S\ref{subsec:monotone}),
each with the hypercomputant as the complete object.
Since the commutative semiring does not affect any arguments with Boolean circuits,
our monotone $\VNP$
is different from the version studied in \cite{Yeh19}, for which there is no known complete polynomial.
\S\ref{sec:boolean} contains an algebraic $\VNP$-hardness proof of the hypercomputant via the equivariance of the Kronecker product.
We embed the theory into the parameterized complexity framework in \S\ref{sec:PNP},
where we also summarize the relations between the algebraic hardness conjectures and the Boolean hardness conjectures.

We show in \S\ref{sec:geomrepr} that iterated matrix multiplication, if interpreted as a matrix of polynomials,
has a connected reductive stabilizer and is characterized by its stabilizer,
which are desirable properties in geometric complexity theory.
We show that the isotypic components of the hypercomputant
are characterized by their continuous symmetries.
We also prove their polystability.

The first sections are meant to be introductory, and new material starts at \S\ref{sec:multilin}.
Occasionally we included short proofs of classical results to make the paper more self-contained.

\section{Preliminaries}\label{sec:prelim}

Define the set of polynomially bounded sequences in $n$ as $\poly(n)=\bigcup_{d\in\IN} O(n^d)$.
Define the set of bivariately polynomially bounded families as
$\poly(n,d) = \{(r_{n,d})_{\substack{r\in\IN\\d\in\IN}}\mid \exists \phi\mathord{\in}\IR[n,d] \, \forall n,d : r_{n,d}\leq \phi(n,d)\}$.

\subsection{Variables}\label{subsec:variables}
We work over countably infinitely many variables, where the exact names of the variables do not matter, and neither does their order matter. Everything that we encounter will use only a finite number of variables.
Depending on the application, we use different variables names.
Most variables are of the form $x_{i_1,\ldots,i_j}$ for some $j\geq 0$ and some list $(i_1,\ldots,i_j)\in\IN^j$.
The tensor product operation on variables is defined by concatenating index lists, for example $x_{1,2}\boxtimes x_{3} = x_{1,2,3}$.
To improve the readability, sometimes indices are written as superscripts in parentheses, which are to be interpreted as additional indices, i.e., $x_{i,j}^{(k)} = x_{i,j,k}$.
We also use other variable names besides $x$, for example $y$, or more exotic variable names in \S\ref{sec:boolean}.
These are distinct variables from the variables in~$x$.
We use the tensor product for $y$ in the same way as for $x$, e.g., $y_{1,2}\boxtimes y_{3} = y_{1,2,3}$.
We also allow mixed tensor products between variables that have different symbols, but we do not introduce any special notation for those, i.e., $x_{1,2}\boxtimes y_{3}$ is a variable with no other name.
We denote the whole set of variables by~$\Var$.
Since the number of variable symbols we use is finite, it is clear that $\Var$ contains countably infinitely many elements.

\subsection{Semialgebras}
A commutative semiring $R$ is a tuple $(R,+,\cdot,0_R,1_R)$ with $0_R\neq 1_R$ and where $(R,+,0_R)$ is a commutative monoid and $(R,\cdot,1_R)$ is a commutative monoid,
and where distributivity holds, i.e., $\alpha\cdot(\beta+\gamma)=\alpha\cdot \beta + \alpha\cdot \gamma$ and $(\beta+\gamma)\cdot \alpha = \beta\cdot \alpha + \gamma\cdot \alpha$, and it holds that $\alpha\cdot 0_R = 0_R = 0_R \cdot \alpha$.
The main interesting examples in our context are the set of nonnegative real numbers $\IR_{\geq 0}$ and the set of complex numbers~$\IC$.

An $R$-semimodule $(M,+,\cdot,0_M)$ is a tuple such that $(M,+,0_M)$ is a commutative monoid and the scalar multiplication map $\cdot:R\times M\to M$ satisfies $\alpha\cdot (m+n)=\alpha\cdot m+\alpha\cdot n$ and $(\alpha+\beta)\cdot m = \alpha\cdot m + \beta \cdot m$, and $(\alpha\cdot \beta)\cdot m = \alpha\cdot(\beta\cdot m)$, and $1_R \cdot m = m$, and $0_R \cdot m = 0_M = \alpha \cdot 0_M$, for $\alpha,\beta\in R$, $m,n\in M$. Multiplication symbols are often omitted.
This is sometimes called a \emph{left} $R$-semimodule, but in this paper all $R$-seminodules are left $R$-seminodules.

An $R$-semialgebra $A$ is an $R$-semimodule with an additional multiplication operation $A\times A \to A$
that satisfies $(a+b)c=ac+bc$ and $c(a+b)=ca+cb$ and 
$(\alpha a)b = \alpha(ab) = a(\alpha b)$,
and there exists $1_A\in A$ with $1_A a = a = a 1_A$ for all $a\in A$.

For an $R$-semialgebra $A$, let $A\times A$ denote the product semialgebra with addition $(a_1,b_1)+(a_2,b_2)=(a_1+a_2,b_1+b_2)$ and scalar multiplication $\alpha(a,b)=(\alpha a,\alpha b)$ and multiplication $(a_1,b_1)\cdot(a_2,b_2)=(a_1\cdot a_2,b_1\cdot b_2)$. A congruence on an $R$-semialgebra $A$ is defined as an $R$-subsemialgebra of $A\times A$ that is reflexive ($(a,a)\in\sim$), symmetric ($(a,b)\in\sim$ implies $(b,a)\in\sim$), and transitive ($(a,b)\in\sim$ and $(b,c)\in\sim$ implies $(a,c)\in\sim$).
For an $R$-semialgebra $A$ and a congruence $\sim$ we can define the quotient $R$-semialgebra $A/\!_{\sim}$ on the equivalence classes of $\sim$ via
$[a]+[b]=[a+b]$, $\alpha[a]=[\alpha a]$, $[a]\cdot[b]=[a\cdot b]$ where the well-definedness can be seen as follows.
If $[a]=[a']$ and $[b]=[b']$, then $(a,a')\in\sim$ and $(b,b')\in\sim$,
and hence $(a+b,a'+b')\in\sim$, i.e., $[a+b] = [a'+b']$
Moreover, if $[a]=[a']$, then $(a,a')\in\sim$ and hence $(\alpha a,\alpha a')=\alpha(a,a')\in\sim$, i.e., $[\alpha a]=[\alpha a']$.
Also, if $[a]=[a']$ and $[b]=[b']$, then $(a,a')\in\sim$ and $(b,b')\in\sim$,
and hence $(a\cdot b,a'\cdot b')\in\sim$, i.e., $[a\cdot b] = [a'\cdot b']$.
The semialgebra axioms are satisfied for $A/\!_{\sim}$ by inheritance from $A$.

A map $\varphi:M\to N$ between two $R$-semimodules is called \emph{$R$-linear} if $\forall \alpha\in R, m,n\in M: \varphi(\alpha m+n)=\alpha\varphi(m)+\varphi(n)$.
Let $M$ denote the free $R$-semimodule generated by~$\Var$,
i.e., $M$ is the set of finite $R$-linear combinations of elements from $\Var$.
The \emph{endomorphism monoid} $\End$ is the set of those $R$-linear maps from $M$ to~$M$ that send every variable to a \emph{finite} linear combination of variables.

For $M$ the free $R$-semimodule generated by~$\Var$,
let $\tensor^\bullet(M)$ denote the tensor algebra, i.e., set of finite $R$-linear combinations of sequences of variables, where the product of sequences is given by the sequence concatenation.
The homogeneous degree $d$ component of $\tensor^\bullet(M)$ is defined as being the $R$-linear span of the set of sequences of exactly $d$ variables, and is denoted by $\tensor^d(M)$.
These sequences of variables are written with a tensor symbol as the divider symbol, e.g., $x_1 \otimes x_2$.
Consider the following two quotients of $\tensor^\bullet(M)$,
which we define carefully, because in general $R$ might not have additive inverses.
\begin{itemize}
\item Define the congruence $\sim$ as the $R$-linear span of the $(f,f')$ with $f,f'\in\sT_d$ such that $\exists\pi\in\aS_d$ with $f=\pi f'$.
All $(f,f)$ are generators, and for each generator $(f,f')\in\sim$ we have $(f',f)\in\sim$.
Note that $f\sim f'$ iff $f$ can be transformed into $f'$ by applying a permutation to each monomial, not necessarily the same for each. This implies the $\sim$ is transitive.
Note also that $\sim$ is closed under the semialgebra product, because if $f\sim f'$ and $f''\sim f'''$, then from the permutations that transform $f$ into $f'$ and those that transform $f''$ into $f'''$ we can readily obtain the permutations that transform $f f''$ into $f' f'''$.
Let $S^\bullet(M) := \tensor^\bullet(M)/\!_{\sim}$.
The notation for the equivalence class of $v_1\otimes\cdots\otimes v_d$ is $v_1\cdots v_d$, or alternatively $[v_1\otimes\cdots\otimes v_d]_{S}$.
\item Define the congruence $\backsim$ as generated by the elements $(\ell\otimes\ell,0)$ and $(0,\ell\otimes\ell)$, $\ell \in M$,
and $(1_M,1_M)$.
Since $(1_M,1_M)\in\backsim$, we have $(f,f)\in\backsim$, thus $\backsim$ is reflexive.
we have $f\backsim f'$ if and only if $f+h+\sum_i h_i \otimes\ell_i\otimes\ell_i \otimes m_i = f'+h+\sum_j h'_j \otimes\ell'_j\otimes\ell'_j \otimes m'_j$,
hence $f'\backsim f$. Hence, $\backsim$ is symmetric.
Transitivity is seen as follows. If $f\backsim f'$ and $f'\backsim f''$, then
\[
f + h+ n \otimes\ell\otimes\ell \otimes m = f'+h+  n' \otimes\ell'\otimes\ell' \otimes m'
\]
and
\[
  f' +\bar h+  \bar n' \otimes\bar\ell'\otimes\bar\ell' \otimes \bar m'
= f''+\bar h+  n'' \otimes \ell''\otimes\ell'' \otimes m'',
\]
which implies that
\[
  f' +\bar h+  \bar n' \otimes\bar\ell'\otimes\bar\ell' \otimes \bar m'  +h+  n' \otimes\ell'\otimes\ell' \otimes m'
= f''+\bar h+  n'' \otimes \ell''\otimes\ell'' \otimes m''  +h+\ n' \otimes\ell'\otimes\ell' \otimes m'
\]
hence
\[
f + h+ n \otimes\ell\otimes\ell \otimes m +  \bar h+  \bar n' \otimes\bar\ell'\otimes\bar\ell' \otimes \bar m'
= f''+\bar h+  n'' \otimes \ell''\otimes\ell'' \otimes m''  +h+\ n' \otimes\ell'\otimes\ell' \otimes m'
\]
and thus $f\backsim f''$.
Let $\Wedge^\bullet(M) := \tensor^\bullet(M)/\!_{\backsim}$.
The notation for the equivalence class of $v_1\otimes\cdots\otimes v_d$ is $v_1\wedge\cdots \wedge v_d$,
or alternatively $[v_1\otimes\cdots\otimes v_d]_{\Wedge}$.
Note that $\ell\wedge\ell = 0$ by definition.
This implies $\ell\wedge\ell'+\ell\wedge\ell'=0$,
because
$\ell\wedge\ell'+\ell\wedge\ell'=
\ell\wedge\ell'+\ell\wedge\ell' + (\ell\wedge\ell) + (\ell'\wedge\ell')
=\ell\wedge(\ell+\ell') + \ell'\wedge(\ell+\ell')
=(\ell+\ell')\wedge(\ell+\ell')=0.
$
Moreover, this implies $\ell\wedge\ell'\wedge\ell''=\ell'\wedge\ell''\wedge\ell$, because
$\ell\wedge\ell'\wedge\ell''
=
\ell\wedge\ell'\wedge\ell''
+
(\ell+\ell'')\wedge(\ell+\ell'')\wedge\ell'
=
\ell\wedge\ell'\wedge\ell''
+
\ell\wedge\ell''\wedge\ell'
+
\ell''\wedge\ell\wedge\ell'
=
\ell''\wedge\ell\wedge\ell'
$.
\end{itemize}

Let $\End$ be the monoid of $R$-linear maps from $M_1$ to $M_1$.
We will not use a notation involving $M$ in the future, but we will write $\sT := \tensor^\bullet(M)$, $\sS := S^\bullet(M)$, $\sW := \Wedge^\bullet(M)$. The homogeneous degree $d$ components are written as $\sT_d$, $\sS_d$, and $\sW_d$.

\begin{remark}
Arbitrary semirings instead of commutative semirings can be used for most parts when applying the following changes.
\begin{itemize}
\item
The semialgebra axiom is changed from $\alpha (ab) = (\alpha a)b = a(\alpha b)$
to $(\alpha\beta)(ab)=(\alpha a)(\beta b)$ for standard basis vectors $a$ and $b$, see \cite[eq.~(1.4)]{Wei06}.
Here, a basis means that every element can be expressed uniquely as an $R$-linear combination of basis vectors.
The standard basis of $\sW_d$ is different over semirings compared to the standard basis over rings,
because over semirings $y\wedge x$ and $x\wedge y$ can be $R$-linearly independent, while over rings we have
$y\wedge x = y\wedge x + (x-y)\wedge(x-y) - x\wedge x - y \wedge y = -x \wedge y$.
\item Even over non-commutative rings, the multiplication operation in $\sS$ is not necessarily commutative.
Similarly, taking the product of two elements in $\sW_{2d}$, $d\geq 1$, is not necessarily commutative.
\item The definition of $\boxtimes$ in \S\ref{sec:multilin} works via $R$-bilinear continuation, which introduces a similar subtlety as for semialgebras, which can be resolved in the same way by choosing standard bases. As a consequence, \eqref{eq:boxtimescommutative} can fail if $R$ is not commutative.
\end{itemize}
It is unclear if the lower bound arguments involving Lemma~\ref{lem:rankofid} would still hold.
We avoid these complications and thus work over commutative semirings.
This is also the generality in which mathlib \cite{mathlib2020} defines the exterior algebra in \texttt{Mathlib/LinearAlgebra/ExteriorAlgebra/Basic.lean}.
\end{remark}

\subsection{Matrices}

$R^n$ is an $R$-semimodule with the standard basis $\{e_1,\ldots,e_n\}$,
i.e., every element in $R^n$ has a unique expression as an $R$-linear combination of $e_1,\ldots,e_n$.
The dual $R$-semimodule $(R^n)^*$ is defined as the set of $R$-linear maps $R^n\to R$.
Every $\phi\in(R^n)^*$ is uniquely determined by its values $\phi(e_i)$,
and a basis of $(R^n)^*$ is given by the $e_i^*$, which are defined via $e_i^*(e_j)=\delta_{i,j}$,
where $\delta_{i,j}=1$ if $i=j$, and $\delta_{i,j}=0$ otherwise.

For an $R$-semimodule $M$, define
$\Mat_n(M) := R^n \otimes M \otimes (R^n)^\ast$.
The multiplication $\Mat_n(M)\times \Mat_n(M)\to \Mat_n(M)$ is defined via
\[
(e_i \otimes m \otimes e_j^*) \cdot (e_{i'} \otimes m' \otimes e_{j'}^*)
= e_j^*(e_i) \, e_i \otimes (m m') \otimes e_{j'}^*
\]
extended bilinearly.
We embed $A_d$ into $\Mat_{1}(A_d)$ via $f\mapsto e_1 \otimes f \otimes e_1^*$.
We embed $\Mat_{n-1}$ into $\Mat_n$ as the notation suggests:
$(e_i \otimes m \otimes e_j^*) \mapsto (e_i \otimes m \otimes e_j^*)$.
 
For $F = \sum_{i,j} e_i \otimes m_{i,j} \otimes e_j^*$ we write $[F]_{i,j} = m_{i,j}$,
i.e., we use brackets $[\,]$ to extract matrix entries.
We use this notation, because our matrices will come from families of matrices,
and hence possibly already have subscripts and/or superscripts.

\subsection{Restrictions}\label{subsec:restrictions}

The monoid $\End$ acts linearly on $\sT_d$ via
\[
T(x_{i_1}\otimes\cdots\otimes x_{i_d}) = T(x_{i_1})\otimes\cdots\otimes T(x_{i_d}),
\]
extended linearly.
The action commutes with the map to the quotient $\sS_d$ or $\sW_d$, in other words
\[
T(x_{i_1}\cdots x_{i_d}) = T(x_{i_1}) \cdots T(x_{i_d}),\qquad
T(x_{i_1}\wedge\cdots \wedge x_{i_d}) = T(x_{i_1}) \wedge\cdots\wedge T(x_{i_d}).
\]
This action lifts to $\Mat(A_d)$ entry-wise, i.e., $[T(F)]_{a,b}:=T([F]_{a,b})$.
Since $A_d$ is an $R$-semimodule, for every $T_{\textup{left}}\in\Mat(R)$ we define
$T_{\textup{left}} \cdot F$ via matrix multiplication, i.e., $[T_{\textup{left}} F]_{a,b}:=\sum_{i}[T_{\textup{left}}]_{a,i}[F]_{i,b}$.
Moreover, for $T_{\textup{right}}\in\Mat(R)$ we define $F \cdot T_{\textup{right}}$ analogously.
All three actions on $\Mat(A_d)$ commute.
We make this a left action of $\Mat(R)\times\End\times\Mat(R)$ on $\Mat(A_d)$ via
\[
(T_{\textup{left}},T,T_{\textup{right}}) F := T_{\textup{left}} \cdot T(F) \cdot T_{\textup{right}}^\transpose
\]
For $F,F'\in\Mat(A_d)$ we define $F\leq F'$ if there exists $T\in\End$, $T_{\textup{left}}\in \Mat(R)$, $T_{\textup{right}}\in \Mat(R)$
such that $(T_{\textup{left}},T,T_{\textup{right}}) F' = F$.
We say that $F'$ \emph{restricts to} $F$.
Clearly, $\leq$ is transitive.

\subsection{Number of essential variables}

For $x_j\in\Var$, let $\partial_i(x_j)=1$ if $i=j$, and $0$ otherwise.
We define the linear map $\partial_i:\sT_d \to \sT_{d-1}$ via
\begin{equation}
\label{eq:partiali}
\partial_i (x_{j_1}\otimes\cdots\otimes x_{j_d})
\ := \
(\partial_i x_{j_1})\otimes x_{j_2}\otimes\cdots\otimes x_{j_d}
 \ + \ 
x_{j_1} \otimes (\partial_i x_{j_2})\otimes\cdots\otimes x_{j_d}
 \ + \cdots + \ 
x_{j_1} \otimes \cdots\otimes (\partial_i x_{j_d})
\end{equation}
The linear map $J:\sT_d \to \sT_{d-1}\otimes \sT_1$ is defined via
\[
J(t) = \sum_{i\in\IN} \partial_i(t) \otimes x_i.
\]
The monoid $\End$ acts on $\sT_{d-1}\otimes \sT_1$ via
\[
T(N\otimes \ell) := T(N) \otimes T(\ell)
\]
for $T\in\End$.
The map $J$ is clearly equivariant, i.e., $T(J(f))=J(T(f))$.

Analogously, we define the linear map $\partial_i:\sS_d\to \sS_{d-1}$ via
\begin{equation}
\label{eq:partialis}
\partial_i (x_{j_1} \cdots x_{j_d})
\ := \
(\partial_i x_{j_1}) x_{j_2} \cdots x_{j_d}
 \ + \ 
x_{j_1} (\partial_i x_{j_2}) \cdots x_{j_d}
 \ + \cdots + \ 
x_{j_1} \cdots (\partial_i x_{j_d}).
\end{equation}
This is well-defined, because permuting the variable positions in $x_{j_1} \cdots x_{j_d}$
in the left-hand side of \eqref{eq:partialis} just permutes the summands on the right-hand side.
We define the map $J:\sS_d \to \sS_{d-1}\otimes \sS_1$ via
$
J(f) = \sum_{i\in\IN} \partial_i(f) \otimes x_i.
$
The monoid $\End$ acts on $\sS_{d-1}\otimes \sS_1$ via
$
T(N\otimes \ell) := T(N) \otimes T(\ell)
$
for $T\in\End$.
The map $J$ is clearly equivariant, i.e., $T(J(f))=J(T(f))$.

For $d\geq 3$, the linear map
$\partial_i:\sW_d\to \sW_{d-1}$ is defined via
\begin{equation}
\label{eq:partialiw}
\partial_i (x_{j_1} \wedge\cdots\wedge x_{j_d})
\ := \
(\partial_i x_{j_1}) x_{j_2} \wedge\cdots\wedge x_{j_d}
 \ + \ 
(1\ 2) x_{j_1} \wedge(\partial_i x_{j_2}) \wedge\cdots\wedge x_{j_d}
 \ + \cdots + \ 
(1\ 2)^d x_{j_1} \wedge\cdots\wedge (\partial_i x_{j_d}),
\end{equation}
where 
$(1\ 2)(v_1\wedge \cdots \wedge v_d) = v_2\wedge v_1 \wedge v_3 \wedge \cdots \wedge v_d$.
This is well-defined, which we see as follows:
For $\partial_i(\ell)=\alpha$ we have
$[\partial_i(\ell\otimes\ell\otimes x_{j_3}\otimes\cdots\otimes x_{j_d})]_\Wedge
=
\alpha \ell \wedge x_{j_3}\wedge\cdots\wedge x_{j_d}
+
(1\ 2) \alpha \ell \wedge x_{j_3}\wedge\cdots\wedge x_{j_d}
+ \ell \wedge\ell \wedge \ldots
= 0$.
Moreover, we have
$[\partial_i(\ell\otimes x_{j_2}\otimes\ell\otimes x_{j_4}\otimes\cdots\otimes x_{j_d})]_\Wedge
=
\alpha x_{j_2}\wedge\ell\wedge x_{j_4}\wedge\cdots\wedge x_{j_d}
+ \alpha \ell \wedge x_{j_2}\wedge x_{j_4}\wedge\cdots\wedge x_{j_d}
+ 0 = 0.
$
All other cases are analogous.
The linear map $J:\sW_d \to \sW_{d-1}\otimes \sW_1$ is defined via
$
J(t) = \sum_{i\in\IN} \partial_i(t) \otimes x_i,
$
and the monoid $\End$ acts on $\sW_{d-1}\otimes \sW_1$ via
$
T(N\otimes \ell) := T(N) \otimes T(\ell)
$
for $T\in\End$.
The map $J$ is clearly equivariant, i.e., $T(J(f))=J(T(f))$.
For $d=2$, if $-1\in R$, then the same construction works by replacing the action of $(1 \ 2)$ by the rescaling with $-1$. We will ignore the case $A_2=\sW_2$ in this paper.

In all three cases (with the exception of $\sW_2$ with $-1\notin R$), we obtain an equivariant map $J : A_d\to A_{d-1}\otimes A_1$.

The \emph{rank} of an element of $N \in A_{d-1}\otimes A_1$ is defined as the smallest $r$ such that $f$ can be written as
\[
N = \sum_{i=1}^r N_i \otimes \ell_i
\]
for some $N_i\in A_{d-1}$ and $\ell_i \in A_1$.

We know that
\begin{equation}\label{eq:Jmatrix}
\text{if $h\leq f$, then $\rk(J(h))\leq\rk(J(f))$,}
\end{equation}
because if $T\in\End$ and $f\in A_d$, $h\in A_d$ with $T(f)=h$
and
$J(f)=\sum_{i=1}^r N_i \otimes \ell_i$,
then
$J(h)=J(T(f))=\sum_{i=1}^r T(N_i) \otimes T(\ell_i)$ due to the equivariance.

Clearly, $\rk(J(f))$ is at most the number of variables in which $f$ is specified.
If $R$ is a field, then $\rk(J(f))$ is sometimes called the \emph{number of essential variables} of~$f$.

This concept generalizes to matrices as follows.
For $F\in\Mat(A_d)$ define $J:\Mat(A_d)\to\Mat(A_{d-1}\otimes A_1)$ via
$[J(F)]_{a,b} := J([F]_{a,b})$.
The rank of an element $M \in \Mat_{m}(A_{d-1}\otimes A_1)$
is defined as the smallest $r$ such that $M$ can be written as
\[
M = \sum_{i=1}^r c_i \otimes N_i \otimes \ell_i \otimes j_i.
\]
for some $c_i\in R^m$, $j_i\in (R^m)^\ast$,
$N_i\in A_{d-1}$, and $\ell_i \in A_1$.

We know that
\begin{equation}\label{eq:Jmatrixformatrices}
\text{if $H\leq F$, then $\rk(J(H))\leq\rk(J(F))$,}
\end{equation}
because if $(T_\textup{left},T,T_\textup{right})\in\Mat(R)\times\End\times\Mat(R)$ and $F\in \Mat(A_d)$, $H\in \Mat(A_d)$ with $(T_\textup{left},T,T_\textup{right})(F)=H$
and
$J(F)=\sum_{i=1}^r c_i \otimes N_i \otimes \ell_i \otimes j_i$,
then
$J(H)=J((T_\textup{left},T,T_\textup{right})(F))=\sum_{i=1}^r T_\textup{left}(c_i) \otimes T(N_i) \otimes T(\ell_i) \otimes T_{\textup{right}}^\transpose(j_i)$.

We also have that $\rk(J(F))$ is at most the number of variables in which $F$ is specified.
And we also say that $\rk(J(F))$ is the number of essential variables.

\subsection{Rank}

\begin{lemma}\label{lem:rankofid}
Over every commutative semiring $R$, the $n\times n$ identity matrix has rank $n$.
\end{lemma}
\begin{proof}
Clearly $\rk(I_n)\leq n$, because $I_n=\sum_{i=1}^n e_i\otimes e_i^\ast$.
For the other direction, assume that $\rk(I_n)\leq r < n$.
Then there exist $N\in R^{n\times r}$ and $M\in R^{r\times n}$ with $NM=I_n$.
We pad $N$ and $M$ with zeros to obtain $n\times n$ matrices $N'$ and $M'$,
where the last column of $N'$ is zero and the last row of $M'$ is zero.
We have $N'M'=I_n$.
Now apply \cite[Theorem]{RS84} to see that $M'N'=I_n$.
This is a contradiction, because the bottom-right entry of $M'N'$ is zero.
\end{proof}

\section{Width}\label{sec:linalg}
In this section we introduce the width $\w(F)$ of an element $F\in \Mat(A_d)$.
The pair $(\w(F),d)$ measures the computational complexity of~$F$.

Let $R$ be a commutative semiring and let $A \in \{\sT,\sS,\sW\}$.
Let $X^{(k)}_{n} := \big( x_{i,j}^{(k)} \big)_{1 \leq i,j \leq n}$ denote the $n \times n$ matrix whose entry at position $(i,j)$ is the variable $x_{i,j}^{(k)}\in\Var$.

\begin{definition}[The computant matrix]
The \emph{computant matrix} $\Xi_{n,d}\in\Mat_n(A_d)$ is defined as
\[
X^{(1)}_{n} \cdots X^{(d)}_{n}.
\]
We denote the bottom-right entry with $\Gamma\!_{n,d} := [\Xi_{n,d}]_{n,n}$.
For $A=\sS$, $\Gamma\!_{n,d}$ is called the iterated matrix multiplication polynomial.
\end{definition}
More explicitly, we have
\begin{equation}\label{eq:Gammaij}
[\Xi_{n,d}]_{a,b} \ = \ \sum_{1 \leq i_1,\ldots,i_{d-1} \leq n} x^{(1)}_{a,i_1} \, x^{(2)}_{i_1,i_2} \, x^{(3)}_{i_2,i_3} \, \cdots \, x^{(d-1)}_{i_{d-2},i_{d-1}} \, x^{(d)}_{i_{d-1},b}.
\end{equation}
Note that
\begin{equation}\label{eq:Gammaindeppos}
\textup{$[\Xi_{n,d}]_{a,b}$ is the same as $[\Xi_{n,d}]_{a',b'}$ up to renaming variables.}
\end{equation}

It will turn out the $\Xi_{n,d}$ has the better properties than $\Gamma\!_{n,d}$ from a geometric and representation theoretic perspective, see \S\ref{subsec:differentsettings}. We use $\Gamma\!_{n,d}$ to connect more easily to computational complexity theory.

\begin{claim}
\label{cla:GammvslargerGamma}
$\Xi_{n,d} \leq \Xi_{n+1,d}$ and $\Gamma\!_{n,d} \leq \Gamma\!_{n+1,d}$.
\end{claim}
\begin{proof}
For $\Xi_{n,d} \leq \Xi_{n+1,d}$, just map every variable $x^{(k)}_{n+1,i}$ and $x^{(k)}_{i,n+1}$ to zero, $1\leq i \leq n+1$.
The same map proves $[\Xi_{n,d}]_{1,1} \leq [\Xi_{n+1,d}]_{1,1}$, which by \eqref{eq:Gammaindeppos} implies $\Gamma\!_{n,d}\leq\Gamma\!_{n+1,d}$.
\end{proof}
\begin{definition}[Width]
For $F\in \Mat(A_d)$, the \emph{width} $\w(F)$ is defined as the smallest $n$ such that $F \leq \Xi_{n,d}$.
\end{definition}

\begin{example}\label{exa:Xiw}
$\w(\Xi_{n,d})=n$.
\end{example}
\begin{proof}
By definition, $\w(\Xi_{n,d})\leq n$.
We claim that $\rk(J(\Xi_{n,d})) = n^2 d$,
which then finishes the proof by~\eqref{eq:Jmatrixformatrices}.
Since $\Xi_{n,d}$ is defined in $n^2 d$ many variables, we have $\rk(J(\Xi_{n,d})) \leq n^2 d$.
For every one of the $n^2 d$ variables $x_{\mathbf{i}}$ in $\Xi_{n,d}$
we can reconstruct $\mathbf i$
from any monomial in $\partial_{\mathbf i} \Xi_{n,d}$,
hence $J(\Xi_{n,d})$ contains an $(n^2 d)\times(n^2 d)$ permutation matrix as a submatrix,
so the claim follows with Lemma~\ref{lem:rankofid}.
\end{proof}

Recall that restriction is defined via a product of three monoids in \S\ref{subsec:restrictions}.
The following lemma shows that restricting $\Xi_{n,d}$ can be done by only one of those monoids.
\begin{lemma}\label{lem:Xirestrict}
Let $F\in\Mat_n(A_d)$.
We have $F\leq\Xi_{n,d}$ if and only if there exists $T\in\End$ with $T(\Xi_{n,d})=F$.
\end{lemma}
\begin{proof}
\begin{eqnarray*}
(T_\textup{left},\id,\id)\Xi_{n,d}
&=&
T_\textup{left} (X_n^{(1)}\cdots X_n^{(d)})
\\
&=&
(T_\textup{left} X_n^{(1)})\cdots X_n^{(d)})
\end{eqnarray*}
Note that
\[
[T_\textup{left} X_n^{(1)}]_{i,j} = \sum_{a}[T_\textup{left}]_{i,a}[X_n^{(1)}]_{a,j}
= \sum_{a}[T_\textup{left}]_{i,a} x_{a,j}^{(1)}
\]
Define $T\in\End$ via $T(x_{i,j}^{(1)}) = \sum_{a}[T_\textup{left}]_{i,a} x_{a,j}^{(1)}$
and $T(x_{i,j}^{(k)})=x_{i,j}^{(k)}$ for $k>1$,
and observe that
$T_\textup{left} X_n^{(1)} = T(X_n^{(1)})$ and $X_n^{(k)} = T(X_n^{(k)})$ for all $k>1$.
Therefore,
\[
(T_\textup{left} X_n^{(1)})\cdot X_n^{(d)} = T(X_n^{(1)}\cdot X_n^{(d)}).
\]
and hence
\[
(T_\textup{left},\id,\id)\Xi_{n,d} = (\id,T,\id)\Xi_{n,d}.
\]
Analogously for $T_\textup{right}$.
\end{proof}

For elements of $A_d$, the width can be defined via restrictions of $\Gamma\!_{n,d}$, as the following claim shows.
\begin{claim}\label{cla:widthfviaGamma}
For $f\in A_d$, we have that $\w(f)$ is the smallest $n$ such that $f \leq \Gamma\!_{n,d}$.
\end{claim}
\begin{proof}
We write $\Gamma'\!_{n,d}=[\Xi_{n,d}]_{1,1}$ and note that $\Gamma'\!_{n,d}\leq\Gamma\!_{n,d}$
and $\Gamma\!_{n,d}\leq\Gamma'\!_{n,d}$ by renaming variables, see \eqref{eq:Gammaindeppos}.
If $f\leq \Gamma\!_{n,d}$, then clearly $f\leq\Xi_{n,d}$.
For the other direction, let $f\leq \Xi_{n,d}$.
According to Lemma~\ref{lem:Xirestrict} there exists $T\in\End$ with
$T(\Xi_{n,d})=f$, i.e., $[T(\Xi_{n,d})]_{1,1}=f$ and $[T(\Xi_{n,d})]_{a,b}=0$ for all $(a,b)\neq(1,1)$.
Since $[\Xi_{n,d}]_{1,1}$ does not involve any $x_{i,j}^{(1)}$ for $i\neq 1$,
the element $[T(\Xi_{n,d})]_{1,1}$ is independent of $T(x_{i,j}^{(1)})$ for $i\neq 1$.
Analogously, the element $[T(\Xi_{n,d})]_{1,1}$ is independent of $T(x_{i,j}^{(d)})$ for $j\neq 1$.
Let $Z\in\End$ map all elements $x_{i,j}^{(1)}$ to $0$ for $i\neq 1$,
and map all elements $x_{i,j}^{(d)}$ to $0$ for $j\neq 1$,
and act as the identity on all other variables.
Then $[TZ\Xi_{n,d}]_{1,1}=[T\Xi_{n,d}]_{1,1} = f$
and $[TZ\Xi_{n,d}]_{a,b}=[T\Xi_{n,d}]_{a,b} = 0$ for all $(a,b)\neq(1,1)$.
But $Z\Xi_{n,d}=\Gamma'\!_{n,d}$, which implies $T\Gamma'\!_{n,d}=f$, hence $f\leq \Gamma'\!_{n,d} \leq \Gamma\!_{n,d}$.
\end{proof}

\subsection{Algebraic Branching Programs}
Several different variants of algebraic branching programs (ABPs) exist,
and we use a variant that allows the ABP to reflect the structure of our recurrence relation.
An ABP is a directed acyclic graph
in which each vertex has a layer from $0,\ldots,d$,
and each edge is either going from a vertex in layer $i$ to another vertex in layer $i$ and is labeled by an element of $R$ (these are called constant edges),
or the edge is going from a vertex in layer $i$ to a vertex in layer $i+1$ and is labeled by an element of $A_1$.
A subset of the vertices in layer 0 are called \emph{sources} $s_1,\ldots,s_a$
and a subset of the vertices in layer $d$ is called \emph{sinks} $t_1,\ldots,t_b$.
To an ABP $G$ we assign its computation result~$F$, which is the matrix whose $(a,b)$ entry is the 
weighted sum over all $s_a$-$t_b$-paths,
where each summand is weighted by the product of its edge labels.
We say that $G$ computes~$F$.
The width of an ABP is defined as the maximum number of vertices in any layer, but in layer $0$ and $d$ we only count sources and sinks.

\begin{claim}
For $F\in\Mat(A_d)$, the width of a smallest ABP computing $F$ is $\w(F)$.
\end{claim}
\begin{proof}
If $\w(F)\leq r$, then $F\leq\Xi_{n,d}$, so there is $T\in\End$ with $T(\Xi_{n,d})=F$.
We take the width $n$ ABP that computes $\Xi_{n,d}$ and apply $T$ to each edge label to obtain a width $n$ ABP that computes $F$.
For the other direction, we are given a width $r$ ABP that computes $F$.
Figure~\ref{fig:compresseduncompressed} gives an illustration.
\begin{figure}
\centering
\begin{subfigure}[b]{0.47\textwidth}
\centering
\begin{tikzpicture}[xscale=1.8,yscale=1.17]
\node[draw,circle,inner sep=0,minimum width=5mm] (s) at (0,0) {$s$};
\node[draw,circle,inner sep=0,minimum width=5mm] (v1a) at (1,-1) {};
\node[draw,circle,inner sep=0,minimum width=5mm] (v1b) at (1,0) {};
\node[draw,circle,inner sep=0,minimum width=5mm] (v1c) at (1,1) {};
\node[draw,circle,inner sep=0,minimum width=5mm] (v2ab) at (2,-0.5) {};
\node[draw,circle,inner sep=0,minimum width=5mm] (v2bc) at (2,0.5) {};
\node[draw,circle,inner sep=0,minimum width=5mm] (v3a) at (3,-1) {};
\node[draw,circle,inner sep=0,minimum width=5mm] (v3b) at (3,0) {};
\node[draw,circle,inner sep=0,minimum width=5mm] (v3c) at (3,1) {};
\node[draw,circle,inner sep=0,minimum width=5mm] (t) at (4,0) {$t$};
\draw[-Triangle] (s) to node[pos=0.5,fill=white,inner sep=0,minimum size=0.5cm] {$x_1$} (v1a);
\draw[-Triangle] (s) to node[pos=0.5,fill=white,inner sep=0,minimum size=0.5cm] {$x_3$} (v1b);
\draw[-Triangle] (s) to node[pos=0.5,fill=white,inner sep=0,minimum size=0.5cm] {$x_5$} (v1c);
\draw[-Triangle] (v1a) to node[pos=0.5,fill=white,inner sep=0,minimum size=0.5cm] {$x_2$} (v2ab);
\draw[-Triangle] (v1b) to node[pos=0.5,fill=white,inner sep=0,minimum size=0.5cm] {$x_4$} (v2ab);
\draw[-Triangle] (v1c) to node[pos=0.5,fill=white,inner sep=0,minimum size=0.5cm] {$x_6$} (v2bc);
\draw[-Triangle] (v2ab) to node[left,pos=0.5,inner sep=0,minimum size=0.5cm] {$\ \ 1$} (v2bc);
\draw[-Triangle] (v2ab) to node[pos=0.5,fill=white,inner sep=0,minimum size=0.5cm] {$x_7$} (v3a);
\draw[-Triangle] (v2bc) to node[pos=0.5,fill=white,inner sep=0,minimum size=0.5cm] {$x_9$} (v3b);
\draw[-Triangle] (v2bc) to node[pos=0.5,fill=white,inner sep=0,minimum size=0.5cm] {$x_{11}$} (v3c);
\draw[-Triangle] (v3a) to node[pos=0.5,fill=white,inner sep=0,minimum size=0.5cm] {$x_{12}$} (t);
\draw[-Triangle] (v3b) to node[pos=0.5,fill=white,inner sep=0,minimum size=0.5cm] {$x_{10}$} (t);
\draw[-Triangle] (v3c) to node[pos=0.5,fill=white,inner sep=0,minimum size=0.5cm] {$x_{8}$} (t);
\end{tikzpicture}
\caption{An ABP with a constant edge.}
\label{subfig:compressed}
\end{subfigure}
\hfill
\begin{subfigure}[b]{0.47\textwidth}
\centering
\begin{tikzpicture}[xscale=1.8,yscale=1.17]
\node[draw,circle,inner sep=0,minimum width=5mm] (s) at (0,0) {$s$};
\node[draw,circle,inner sep=0,minimum width=5mm] (v1a) at (1,-1) {};
\node[draw,circle,inner sep=0,minimum width=5mm] (v1b) at (1,0) {};
\node[draw,circle,inner sep=0,minimum width=5mm] (v1c) at (1,1) {};
\node[draw,circle,inner sep=0,minimum width=5mm] (v2ab) at (2,-0.5) {};
\node[draw,circle,inner sep=0,minimum width=5mm] (v2bc) at (2,0.5) {};
\node[draw,circle,inner sep=0,minimum width=5mm] (v3a) at (3,-1) {};
\node[draw,circle,inner sep=0,minimum width=5mm] (v3b) at (3,0) {};
\node[draw,circle,inner sep=0,minimum width=5mm] (v3c) at (3,1) {};
\node[draw,circle,inner sep=0,minimum width=5mm] (t) at (4,0) {$t$};
\draw[-Triangle] (s) to node[pos=0.5,fill=white,inner sep=0,minimum size=0.5cm] {$x_1$} (v1a);
\draw[-Triangle] (s) to node[pos=0.5,fill=white,inner sep=0,minimum size=0.5cm] {$x_3$} (v1b);
\draw[-Triangle] (s) to node[pos=0.5,fill=white,inner sep=0,minimum size=0.5cm] {$x_5$} (v1c);
\draw[-Triangle] (v1a) to node[pos=0.3,fill=white,inner sep=0.5mm] {$x_2$} (v2ab);
\draw[-Triangle] (v1b) to node[pos=0.25,fill=white,inner sep=0.5mm] {$x_4$} (v2ab);
\draw[-Triangle] (v1a) to node[pos=0.3,fill=white,inner sep=0.5mm] {$x_2$} (v2bc);
\draw[-Triangle] (v1b) to node[pos=0.5,fill=white,inner sep=0,minimum size=0.5cm] {$x_4$} (v2bc);
\draw[-Triangle] (v1c) to node[pos=0.5,fill=white,inner sep=0,minimum size=0.5cm] {$x_6$} (v2bc);
\draw[-Triangle] (v2ab) to node[pos=0.5,fill=white,inner sep=0,minimum size=0.5cm] {$x_7$} (v3a);
\draw[-Triangle] (v2bc) to node[pos=0.5,fill=white,inner sep=0,minimum size=0.5cm] {$x_9$} (v3b);
\draw[-Triangle] (v2bc) to node[pos=0.5,fill=white,inner sep=0,minimum size=0.5cm] {$x_{11}$} (v3c);
\draw[-Triangle] (v3a) to node[pos=0.5,fill=white,inner sep=0,minimum size=0.5cm] {$x_{12}$} (t);
\draw[-Triangle] (v3b) to node[pos=0.5,fill=white,inner sep=0,minimum size=0.5cm] {$x_{10}$} (t);
\draw[-Triangle] (v3c) to node[pos=0.5,fill=white,inner sep=0,minimum size=0.5cm] {$x_{8}$} (t);
\end{tikzpicture}
\caption{An ABP with the constant edge removed.}
\label{subfig:uncompressed}
\end{subfigure}
\caption{Removal of constant edges from ABPs.}
\label{fig:compresseduncompressed}
\end{figure}
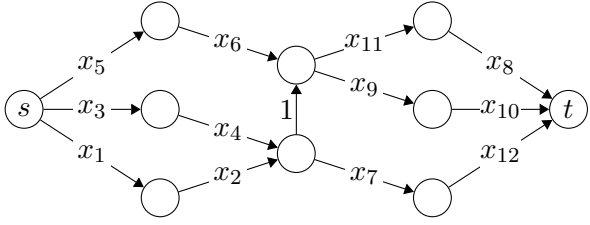
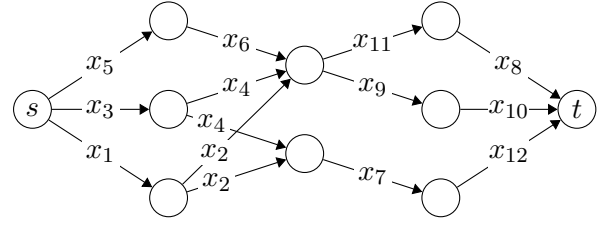
If we have edges in a program \emph{within a vertex layer},
then these edges must be constant edges, i.e., edges labeled by elements from~$R$,
and the set of these edges is not allowed to contain a directed cycle.
We now use the following process.
If there are any constant edges, then
we pick a vertex $v$ in some vertex layer $k<d$ that has no incoming constant edges but outgoing constant edges,
and then for each $u$ in vertex layer $k+1$ we compute $\ell_{v,u} := \sum_{p \in\mathcal P(v,u)} \prod_{e \in p} p \in A_1$.
We then remove all constant edges that are going out of $v$ and add edges from $v$ to all $u$ with label $\ell_{v,u}$.
The ABP still computes the same, but has less constant edges.
A dual approach is also possible, so that all constant edges can be removed:
we pick a vertex $v$ in some vertex layer $k>1$ that has no outgoing constant edges but incoming constant edges,
and then for each $u$ in vertex layer $k-1$ we compute $\ell_{u,v} := \sum_{p \in\mathcal P(u,v)} \prod_{e \in p} p \in A_1$.
We then remove all constant edges that are going into $v$ and add edges from all $u$ to $v$ with label $\ell_{u,v}$.
From the resulting ABP without constant edges we can readily read off a $T$ with $T(\Xi_{n,d})=F$.
\end{proof}

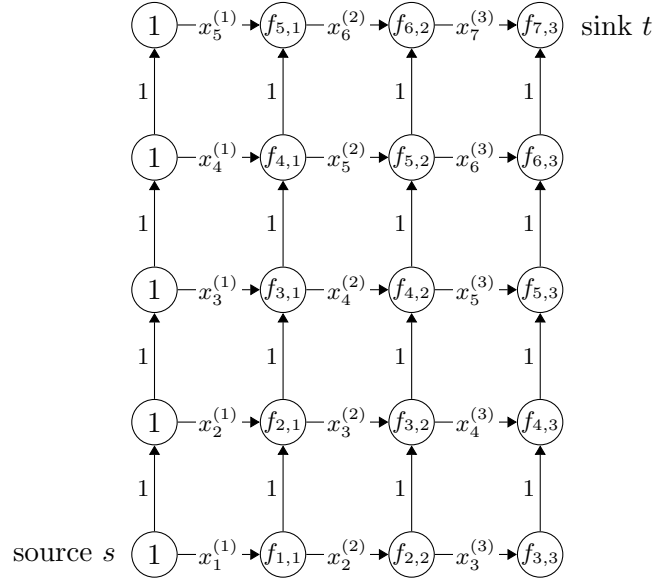
\begin{figure}[htbp]
\centering
\begin{tikzpicture}[xscale=1.7,yscale=1.75]
\node (slabel) at (0.3,1) {\textup{source} $s$};
\node[fill=white,draw,circle,inner sep=0,minimum size=0.6cm] (v11) at (1,1) {$1$};
\node[fill=white,draw,circle,inner sep=0,minimum size=0.6cm] (v21) at (1,2) {$1$};
\node[fill=white,draw,circle,inner sep=0,minimum size=0.6cm] (v31) at (1,3) {$1$};
\node[fill=white,draw,circle,inner sep=0,minimum size=0.6cm] (v41) at (1,4) {$1$};
\node[fill=white,draw,circle,inner sep=0,minimum size=0.6cm] (v51) at (1,5) {$1$};
\node[fill=white,draw,circle,inner sep=0,minimum size=0.6cm] (v12) at (2,1) {\footnotesize $f_{1,1}$};
\node[fill=white,draw,circle,inner sep=0,minimum size=0.6cm] (v22) at (2,2) {\footnotesize $f_{2,1}$};
\node[fill=white,draw,circle,inner sep=0,minimum size=0.6cm] (v32) at (2,3) {\footnotesize $f_{3,1}$};
\node[fill=white,draw,circle,inner sep=0,minimum size=0.6cm] (v42) at (2,4) {\footnotesize $f_{4,1}$};
\node[fill=white,draw,circle,inner sep=0,minimum size=0.6cm] (v52) at (2,5) {\footnotesize $f_{5,1}$};
\node[fill=white,draw,circle,inner sep=0,minimum size=0.6cm] (v13) at (3,1) {\footnotesize $f_{2,2}$};
\node[fill=white,draw,circle,inner sep=0,minimum size=0.6cm] (v23) at (3,2) {\footnotesize $f_{3,2}$};
\node[fill=white,draw,circle,inner sep=0,minimum size=0.6cm] (v33) at (3,3) {\footnotesize $f_{4,2}$};
\node[fill=white,draw,circle,inner sep=0,minimum size=0.6cm] (v43) at (3,4) {\footnotesize $f_{5,2}$};
\node[fill=white,draw,circle,inner sep=0,minimum size=0.6cm] (v53) at (3,5) {\footnotesize $f_{6,2}$};
\node[fill=white,draw,circle,inner sep=0,minimum size=0.6cm] (v14) at (4,1) {\footnotesize $f_{3,3}$};
\node[fill=white,draw,circle,inner sep=0,minimum size=0.6cm] (v24) at (4,2) {\footnotesize $f_{4,3}$};
\node[fill=white,draw,circle,inner sep=0,minimum size=0.6cm] (v34) at (4,3) {\footnotesize $f_{5,3}$};
\node[fill=white,draw,circle,inner sep=0,minimum size=0.6cm] (v44) at (4,4) {\footnotesize $f_{6,3}$};
\node[fill=white,draw,circle,inner sep=0,minimum size=0.6cm] (v54) at (4,5) {\footnotesize $f_{7,3}$};
\node (tlabel) at (4.6,5) {\textup{sink} $t$};
\draw[-Triangle] (v11) to node[left,pos=0.5,inner sep=0,minimum size=0.5cm] {\footnotesize$\ \ 1$} (v21);
\draw[-Triangle] (v21) to node[left,pos=0.5,inner sep=0,minimum size=0.5cm] {\footnotesize$\ \ 1$} (v31);
\draw[-Triangle] (v31) to node[left,pos=0.5,inner sep=0,minimum size=0.5cm] {\footnotesize$\ \ 1$} (v41);
\draw[-Triangle] (v41) to node[left,pos=0.5,inner sep=0,minimum size=0.5cm] {\footnotesize$\ \ 1$} (v51);
\draw[-Triangle] (v12) to node[left,pos=0.5,inner sep=0,minimum size=0.5cm] {\footnotesize$\ \ 1$} (v22);
\draw[-Triangle] (v22) to node[left,pos=0.5,inner sep=0,minimum size=0.5cm] {\footnotesize$\ \ 1$} (v32);
\draw[-Triangle] (v32) to node[left,pos=0.5,inner sep=0,minimum size=0.5cm] {\footnotesize$\ \ 1$} (v42);
\draw[-Triangle] (v42) to node[left,pos=0.5,inner sep=0,minimum size=0.5cm] {\footnotesize$\ \ 1$} (v52);
\draw[-Triangle] (v13) to node[left,pos=0.5,inner sep=0,minimum size=0.5cm] {\footnotesize$\ \ 1$} (v23);
\draw[-Triangle] (v23) to node[left,pos=0.5,inner sep=0,minimum size=0.5cm] {\footnotesize$\ \ 1$} (v33);
\draw[-Triangle] (v33) to node[left,pos=0.5,inner sep=0,minimum size=0.5cm] {\footnotesize$\ \ 1$} (v43);
\draw[-Triangle] (v43) to node[left,pos=0.5,inner sep=0,minimum size=0.5cm] {\footnotesize$\ \ 1$} (v53);
\draw[-Triangle] (v14) to node[left,pos=0.5,inner sep=0,minimum size=0.5cm] {\footnotesize$\ \ 1$} (v24);
\draw[-Triangle] (v24) to node[left,pos=0.5,inner sep=0,minimum size=0.5cm] {\footnotesize$\ \ 1$} (v34);
\draw[-Triangle] (v34) to node[left,pos=0.5,inner sep=0,minimum size=0.5cm] {\footnotesize$\ \ 1$} (v44);
\draw[-Triangle] (v44) to node[left,pos=0.5,inner sep=0,minimum size=0.5cm] {\footnotesize$\ \ 1$} (v54);
\draw[-Triangle] (v11) to node[midway,fill=white,inner sep=1pt] {\footnotesize$x_1^{(1)}$} (v12);
\draw[-Triangle] (v12) to node[midway,fill=white,inner sep=1pt] {\footnotesize$x_2^{(2)}$} (v13);
\draw[-Triangle] (v13) to node[midway,fill=white,inner sep=1pt] {\footnotesize$x_3^{(3)}$} (v14);
\draw[-Triangle] (v21) to node[midway,fill=white,inner sep=1pt] {\footnotesize$x_2^{(1)}$} (v22);
\draw[-Triangle] (v22) to node[midway,fill=white,inner sep=1pt] {\footnotesize$x_3^{(2)}$} (v23);
\draw[-Triangle] (v23) to node[midway,fill=white,inner sep=1pt] {\footnotesize$x_4^{(3)}$} (v24);
\draw[-Triangle] (v31) to node[midway,fill=white,inner sep=1pt] {\footnotesize$x_3^{(1)}$} (v32);
\draw[-Triangle] (v32) to node[midway,fill=white,inner sep=1pt] {\footnotesize$x_4^{(2)}$} (v33);
\draw[-Triangle] (v33) to node[midway,fill=white,inner sep=1pt] {\footnotesize$x_5^{(3)}$} (v34);
\draw[-Triangle] (v41) to node[midway,fill=white,inner sep=1pt] {\footnotesize$x_4^{(1)}$} (v42);
\draw[-Triangle] (v42) to node[midway,fill=white,inner sep=1pt] {\footnotesize$x_5^{(2)}$} (v43);
\draw[-Triangle] (v43) to node[midway,fill=white,inner sep=1pt] {\footnotesize$x_6^{(3)}$} (v44);
\draw[-Triangle] (v51) to node[midway,fill=white,inner sep=1pt] {\footnotesize$x_5^{(1)}$} (v52);
\draw[-Triangle] (v52) to node[midway,fill=white,inner sep=1pt] {\footnotesize$x_6^{(2)}$} (v53);
\draw[-Triangle] (v53) to node[midway,fill=white,inner sep=1pt] {\footnotesize$x_7^{(3)}$} (v54);
\end{tikzpicture}
\caption{An ABP that computes $f_{7,3}$.
The vertex label at each vertex $v$ is the weighted sum of all $s$-$v$-paths.}
\label{fig:cprogh}
\end{figure}
\begin{example}\label{exa:elemsym}
The elementary symmetric polynomial is defined as
\[
e_{n,d} = \sum_{1 \leq i_1<i_2<\cdots<i_d\leq n} x_{i_1} \cdots x_{i_d}.
\]
Its layered variant is defined as
\[
f_{n,d} = \sum_{1 \leq i_1<i_2<\cdots<i_d\leq n} x^{(1)}_{i_1} \cdots x^{(d)}_{i_d}.
\]
Clearly $e_{n,d}\leq f_{n,d}$ by mapping $x_{i}^{(k)}$ to $x_i$.
We have
\[
f_{n,d}= f_{n-1,d-1} \cdot x_n^{(d)} + f_{n-1,d}
\]
and this recursion can be used directly to construct a width $n$ ABP that computes $f_{n,d}$,
see Figure~\ref{fig:cprogh}, thus $\w(f_{n,d})\leq n$.
\end{example}
We will use the technique in Example~\ref{exa:elemsym} several times in this paper to construct ABPs.
Another important example is the following.
\begin{proposition}\label{pro:wdetpoly}
Let $R$ be a commutative ring and let $A=\sS$.
We have $\w(\det_d) \in\poly(d)$.
\end{proposition}
\begin{proof}
Let $X_n = (x_{i,j})_{\substack{1\leq i \leq n\\1 \leq j \leq n}} \in \Mat_n(\sS_1)$.
Define
\begin{equation}\label{eq:cpcS}
\cpc_{n,d} \ = \ \sum_{S\in\binom{\{1,\ldots,n\}}{d}}\det([X_n]_{S,S}) \in \sS_d,
\end{equation}
where $[X_n]_{S,S}$ is the square submatrix of $X_n$ with row indices $S$ and column indices $S$.
Let $\pow_{n,i} := [X_n^i]_{n,n}$, and note $\w(\pow_{n,i}) \leq n$.
We use \cite[eq.~(2.8)]{Ike25}:
\begin{equation}
\label{eq:S}
\cpc_{n-1,d} \ = \ \sum_{i=0}^d (-1)^{i} \, \cpc_{n,d-i} \, \pow_{n,i}.
\end{equation}
We use that $\pow_{n,0}=1$ to rewrite equation~\eqref{eq:S} in the form
\begin{equation}
\label{eq:cpcrecurse}
\cpc_{n,d} \ = \ \cpc_{n-1,d} + \sum_{i=1}^{d} (-1)^{i+1} \, \cpc_{n,d-i} \, \pow_{n,i}.
\end{equation}
This recursion immediately gives an ABP that computes $\det_d$, see Figure~\ref{fig:charpolyprogram}.
\begin{figure}
\centering
\begin{tikzpicture}[xscale=2.5,yscale=2]
\node[fill=white,draw,circle,inner sep=0,minimum size=0.6cm] (v40) at (0,-1) {$s$};
\node[fill=white,draw,circle,inner sep=0,minimum size=0.6cm] (v30) at (0,0) {};
\draw[-Triangle] (v40) to node[left,pos=0.7,inner sep=0,minimum size=0.5cm] {\footnotesize$\ \ 1$} (v30);
\node[fill=white,draw,circle,inner sep=0,minimum size=0.6cm] (v20) at (0,1) {};
\node[fill=white,draw,circle,inner sep=0,minimum size=0.6cm] (v10) at (0,2) {};
\draw[-Triangle] (v30) to node[left,pos=0.7,inner sep=0,minimum size=0.5cm] {\footnotesize$\ \ 1$} (v20);
\draw[-Triangle] (v20) to node[left,pos=0.7,inner sep=0,minimum size=0.5cm] {\footnotesize$\ \ 1$} (v10);
\node[fill=white,draw,circle,inner sep=0,minimum size=0.5cm] (v11) at (1,0) {$\cpc_{1,1}$};
\draw[-Triangle] (v30) to node[midway,fill=white,inner sep=1pt] {\footnotesize$\pow_{1,1}$} (v11);
\node[fill=white,draw,circle,inner sep=0,minimum size=0.5cm] (v21) at (1,1) {$\cpc_{2,1}$};
\draw[-Triangle] (v20) to node[pos=0.5,fill=white,inner sep=0] {\footnotesize$\pow_{2,1}$} (v21);
\draw[-Triangle] (v11) to node[left,pos=0.7,inner sep=0,minimum size=0.5cm] {\footnotesize$\ \ 1$} (v21);
\node[fill=white,draw,circle,inner sep=0,minimum size=0.5cm] (v31) at (1,2) {$\cpc_{3,1}$};
\draw[-Triangle] (v10) to node[pos=0.5,fill=white,inner sep=0] {\footnotesize$\pow_{3,1}$} (v31);
\draw[-Triangle] (v21) to node[left,pos=0.7,inner sep=0,minimum size=0.5cm] {\footnotesize$\ \ 1$} (v31);
\node[fill=white,draw,circle,inner sep=0,minimum size=0.6cm] (v22) at (2,1) {$\cpc_{2,2}$};
\draw[-Triangle] (v20) to[bend left] node[pos=0.5,fill=white,inner sep=0] {\footnotesize$-\pow_{2,2}$} (v22);
\draw[-Triangle] (v21) to node[pos=0.5,fill=white,inner sep=1pt] {\footnotesize$\pow_{2,1}$} (v22);
\node[fill=white,draw,circle,inner sep=0,minimum size=0.6cm] (v32) at (2,2) {$\cpc_{3,2}$};
\draw[-Triangle] (v31) to node[pos=0.5,fill=white,inner sep=1pt] {\footnotesize$\pow_{3,1}$} (v32);
\node[fill=white,draw,circle,inner sep=0,minimum size=0.6cm] (v33) at (3,2) {$\cpc_{3,3}$};
\draw[-Triangle] (v32) to node[pos=0.5,fill=white,inner sep=1] {\footnotesize$\pow_{3,1}$} (v33);
\draw[-Triangle] (v22) to node[left,pos=0.7,inner sep=0,minimum size=0.5cm] {\footnotesize$\ \ 1$} (v32);
\draw[-Triangle] (v10) to[bend left,looseness=1] node[pos=0.5,fill=white,inner sep=0] {\footnotesize$-\pow_{3,2}$} (v32);
\draw[-Triangle] (v31) to[bend left,looseness=1] node[pos=0.5,fill=white,inner sep=0] {\footnotesize$-\pow_{3,2}$} (v33);
\draw[-Triangle] (v10) to[bend left=40] node[pos=0.5,fill=white,inner sep=0] {\footnotesize$\pow_{3,3}$} (v33);
\end{tikzpicture}
\caption{An ABP that computes $\det_3 = \cpc_{3,3}$.
}
\label{fig:charpolyprogram}
\end{figure}
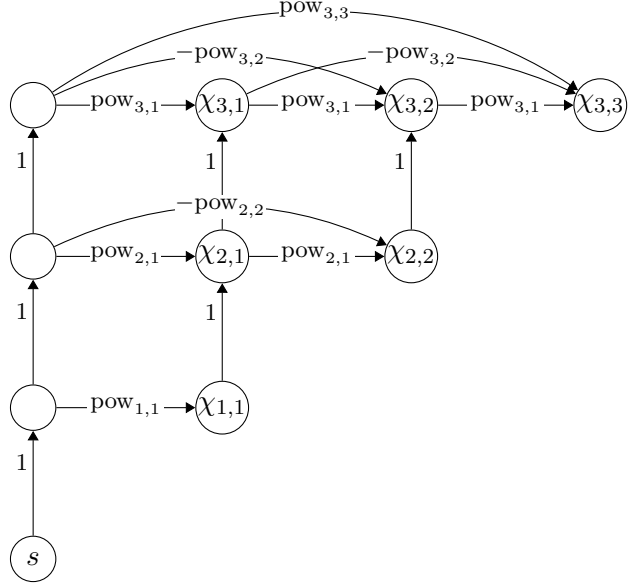
This shows $\w(\det_d)\in O(d^4)$. There are many other proofs of $\w(\det_d)\in\poly(d)$.
\end{proof}

We have
\begin{equation}\label{eq:monleqdet}
\mon_d \leq \det_d
\end{equation}
via setting all off-diagonal variables to zero and renaming variables.

There is the following duality between the degree and the width with respect to subadditivity and submultiplicativity.
\begin{claim}
\label{cla:duality}
Let $R$ be a commutative semiring and $A\in\{\sT,sS,\sW\}$.
Let $F,G\in A_d$. We have
\begin{align*}
\deg(F+G) &\leq \max(\deg(F),\deg(G))
\\
\deg(FG) &\leq \deg(F)+\deg(G)
\\
\w(F+G) &\leq \w(F)+\w(G)
\\
\w(FG) &\leq \max(\w(F),\w(G)),
\end{align*}
where we use the conventions $\deg(0)=0$ and $\w(0)=0$.
\end{claim}
\begin{proof}
The statements for the degree are obvious.
The statements for the width are classical and can be proved by explicit constructions that appear for example in~\cite{Val79}.
The width-of-sum statement is obtained by taking an ABP for $F$ and an ABP $G$ and obtain an ABP for the sum $F+G$ by identifying the respective sources with each other and identifying the respective sinks with each other.
The width-of-product statement is obtained by taking an ABP $F$ and an ABP $G$ and obtain an ABP for the product $FG$ by identifying the sinks in $F$ with the respective sources in $G$, where one adds dummy sources of sinks to make the numbers match.
\end{proof}

It is sometimes useful to convert from multiple sources and sinks to single source and single sink by
adding a new source $s$ in vertex layer $-1$ and a new sink in vertex layer $d+1$
and add edges $(s,s_i)$ with label $e_i$
and edges $(t_j,t)$ with label $e_j^*$.

\section{Higher-order computation}\label{sec:multilin}
In this section we define the Kronecker product~$\boxtimes$.
We start with the definition for $\sS$ in characteristic zero in \eqref{eq:boxtimescharzero}, and then turn to arbitrary commutative semirings in \S\ref{subsec:boxtimes}.
In \S\ref{subsec:kronaction} we then generalize the construction to $\sT$ and $\sW$.
We define the hypercomputant in \S\ref{subsec:higherordercomp}, which in characteristic 0 is $\hGamma\!_{n,d} = \Gamma\!_{n,d}^{\,\boxtimes 2}$.
We phrase Valiant's conjecture via the width $\w(\hGamma\!_{n,d})$ in \S\ref{subsec:Valiantsconj}.
We show that $\HC_d\leq \hGamma\!_{d,d}$ for the Hamiltonian cycle polynomial $\HC_d$,
and $\Clique_{n,k} \leq \hGamma\!_{n,\binom{k+1}{2}}$ for the Clique polynomial in \S\ref{sec:higherorderwidth}.

\subsection{Polarization and Restitution}
Let $R$ be a commutative semiring, and $A=\sS$.
Let
\[
\Symt_d \ := \ \big(\sT_d\big)^{\aS_d} = \{s \in \sT_d \mid \forall \pi \in \aS_d: \pi s = s\}
\]
denote the set of symmetric tensors.
The set $\Symt_d$ is an $R$-semimodule, i.e., it is closed under addition and closed under multiplication with elements of~$R$.
The \emph{polarization} map
\[
\pol: \sS_d \to \Symt_d
\]
is defined via
\[
\pol(x_{i_1}\cdots x_{i_d}) = \sum_{\pi\in\aS_d} x_{i_{\pi(1)}}\otimes\cdots\otimes x_{i_{\pi(d)}}
\]
and $R$-linear extension.
The restitution map
$\res: \sT_d \to \sS_d$
is the canonical map to the quotient,
i.e., $\res(x_{i_1}\otimes\cdots \otimes x_{i_d}) = x_{i_1}\cdots x_{i_d}$.
One can see that $\res\circ\pol = d! \cdot \textup{id}_{\sS_d}$,
which makes $\res$ difficult to use in finite characteristic.
We will treat these concerns in \S\ref{subsec:boxtimes}.
The Kronecker product $\boxtimes$ is a bilinear map
\[
\sT_d \times \sT_d \to \sT_d
\]
defined via
\begin{equation}
\label{eq:boxtimestensors}
(x_{i_1}\otimes\cdots \otimes x_{i_d}) \boxtimes (x_{j_1}\otimes\cdots \otimes x_{j_d})
 \ = \ 
x_{i_1,j_1}\otimes\cdots \otimes x_{i_d,j_d},
\end{equation}
and extended $R$-bilinearly.
\begin{remark}
The $\boxtimes$ operation is well-known in the setting of fast matrix multiplication, see e.g.\ \cite[p.\ 23]{Bla13} (where the symbol $\otimes$ is used instead of $\boxtimes$).
If $M_n = \sum_{i,j,k=1}^n x_{i,j}\otimes x_{j,k} \otimes x_{k,i} \in \sT_3$ denotes the $n\times n$ matrix multiplication tensor,
then $M_n\boxtimes M_m = M_{nm}$.
\end{remark}
For $s,s'\in\Symt_d$, we have $s\boxtimes s' \in \Symt_d$.
Hence, for two polynomials $h,f\in \sS_d$, if $\tfrac{1}{d!}\in R$, then we can construct a new polynomial
\begin{equation}
\label{eq:boxtimescharzero}
h \boxtimes f \ := \ \tfrac{1}{d!}\,\res(\pol(h)\boxtimes\pol(f)).
\end{equation}

Several examples of polynomials of this form exist in the literature, see Example~\ref{exa:boxtimes} below,
but before discussing those,
in \S\ref{subsec:boxtimes}
we generalize \eqref{eq:boxtimescharzero}, so that it works over arbitrary semirings.

\subsection{The Kronecker product of polynomials}
\label{subsec:boxtimes}
Let $R$ be a commutative semiring, and let $A=\sS$.
For $f,h \in A_d$ we define $f\boxtimes h \in A_d$ via
\begin{equation}
\label{eq:defboxtimes}
(x_{i_1} \cdots x_{i_d}) \boxtimes (x_{j_1} \cdots x_{j_d}) := \sum_{\pi\in\aS_d} x_{i_1,j_{\pi(1)}} \cdots x_{i_d,j_{\pi(d)}}
\end{equation}
and bilinear extension.
First note that this is a well-defined map
$A_d \times A_d \to A_d$,
because for any $\sigma,\tau\in\aS_d$ we have
\begin{eqnarray*}
(x_{i_{\sigma(1)}} \cdots x_{i_{\sigma(d)}}) \boxtimes (x_{j_{\tau(1)}} \cdots x_{j_{\tau(d)}})
& = &
\sum_{\pi\in\aS_d} x_{i_{\sigma(1)},j_{\tau(\pi(1))}} \cdots x_{i_{\sigma(d)},j_{\tau(\pi(d))}}
\\
& = &
\sum_{\pi\in\aS_d} x_{i_{1},j_{\sigma^{-1}(\tau(\pi(1)))}} \cdots x_{i_{d},j_{\sigma^{-1}(\tau(\pi(d)))}}
\\
& \stackrel{\varrho:=\sigma^{-1}\tau\pi}{=} &
\sum_{\varrho\in\aS_d} x_{i_{1},j_{\varrho(1)}} \cdots x_{i_{d},j_{\varrho(d)}}
 \ = \ (x_{i_1} \cdots x_{i_d}) \boxtimes (x_{j_1} \cdots x_{j_d}).
\end{eqnarray*}
In fact, note that the same reshuffling argument shows that
\eqref{eq:defboxtimes}
is the same as
\begin{equation}
\label{eq:defboxtimesII}
\sum_{\pi\in\aS_d} x_{i_{\pi(1)},j_1} \cdots x_{i_{\pi(d)},j_d}.
\end{equation}
This readily shows the associativity of the Kronecker product:
\begin{eqnarray}
\big((x_{i_1} \cdots x_{i_d}) \boxtimes (x_{j_1} \cdots x_{j_d})\big) \boxtimes (x_{k_1} \cdots x_{k_d})
&=&\sum_{\substack{\sigma\in\aS_d, \, \pi\in\aS_d}} x_{i_{\sigma(1)},j_1,k_{\pi(1)}} \cdots x_{i_{\sigma(d)},j_d,k_{\pi(d)}}\nonumber
\\
&=&(x_{i_1} \cdots x_{i_d}) \boxtimes \big((x_{j_1} \cdots x_{j_d}) \boxtimes (x_{k_1} \cdots x_{k_d})\big).
\label{eq:kronprodassoc}
\end{eqnarray}
The equality of \eqref{eq:defboxtimes} and \eqref{eq:defboxtimesII} also imply that
\begin{equation}\label{eq:boxtimescommutative}
\textup{both $f\boxtimes h \leq h\boxtimes f$ and $h\boxtimes f \leq f\boxtimes h$}
\end{equation}
by renaming each $x_{i,j}$ to $x_{j,i}$.

\begin{proposition}
\label{pro:boxtimesequivariance}
If $f \leq h$ and $\widetilde f \leq \widetilde h$, then
$f\boxtimes \widetilde f \leq h\boxtimes \widetilde h$.
\end{proposition}
\begin{proof}
For $T,T'\in\cE$ we define $P \in \cE$ as
$P(x_{i,j}) := \sum_{i',j'} \alpha_{i',i}\alpha'_{j',j} x_{i',j'}$,
where
$\alpha_{i',i}\in R$ and $\alpha'_{i',i}\in R$
are defined via
$T(x_i) = \sum_{i'} \alpha_{i',i} x_{i'}$ and $T'(x_j) = \sum_{j'} \alpha'_{j',j} x_{j'}$.
\begin{eqnarray*}
T(f)\boxtimes T(h)
&=&
T(\sum_a \alpha_a x_{a,1}\cdots x_{a,d})
\boxtimes
T'(\sum_b \alpha_b x_{b,1}\cdots x_{b,d})
\\
&=&
\sum_{a,b} \alpha_a \alpha_b
\Big(
T(x_{a,1}\cdots x_{a,d})
\boxtimes
T'(x_{b,1}\cdots x_{b,d})
\Big)
\\
&\stackrel{(\ast)}{=}&
\sum_{a,b} \alpha_a \alpha_b
\Big(
P((x_{a,1}\cdots x_{a,d}) \boxtimes (x_{b,1}\cdots x_{b,d}))
\Big)
\\
&=&
\Big(
P((\sum_a \alpha_a x_{a,1}\cdots x_{a,d}) \boxtimes \sum_b \alpha_b(x_{b,1}\cdots x_{b,d}))
\Big)
\\
&=&
P(f\boxtimes h)
\end{eqnarray*}
We verify $(\ast)$ in the following calculation, which finishes the proof.
\begin{align*}
T(f)\boxtimes T'(h)
&=
T(x_{i_1}\cdots x_{i_d})\boxtimes T'(x_{j_1}\cdots x_{j_d})
\\
&=
\Big(\sum_{i'_1} \alpha_{i'_1,i_1} x_{i'_1}\Big) \cdots \Big(\sum_{i'_d} \alpha_{i'_d,i_d} x_{i'_d}\Big)
\boxtimes
\Big(\sum_{j'_1} \alpha_{j'_1,j_1} x_{j'_1}\Big) \cdots \Big(\sum_{j'_d} \alpha_{j'_d,j_d} x_{j'_d}\Big)
\\
&=
\Big(\sum_{i'_1} \alpha_{i'_1,i_1}\Big)\cdots\Big(\sum_{i'_d} \alpha_{i'_d,i_d}\Big)
\Big(\sum_{j'_1} \alpha_{j'_1,j_1}\Big)\cdots\Big(\sum_{j'_d} \alpha_{j'_d,j_d}\Big)
 x_{i'_1} \cdots  x_{i'_d}
\boxtimes
 x_{j'_1} \cdots  x_{j'_d}
\\
&=
\sum_{\pi\in\aS_d}
\Big(\sum_{\substack{i'_1,\ldots,i'_d\\j'_1,\ldots,j'_d}} \alpha_{i'_1,i_1}\cdots \alpha_{i'_d,i_d} \cdot \alpha_{j'_1,j_1}\cdots\alpha_{j'_d,j_d}\Big)
x_{i'_1,j'_{\pi(1)}} \cdots  x_{i'_d,j'_{\pi(d)}}
\\
&=
\sum_{\pi\in\aS_d}
\Big(\sum_{\substack{i'_1,\ldots,i'_d\\j'_1,\ldots,j'_d}} \alpha_{i'_1,i_1}\cdots \alpha_{i'_d,i_d} \cdot \alpha_{j'_1,\pi^{-1}(j_1)}\cdots\alpha_{j'_d,\pi^{-1}(j_d)}\Big)
x_{i'_1,j'_1} \cdots  x_{i'_d,j'_d}
\\
&=
\sum_{\sigma\in\aS_d}
\Big(\sum_{\substack{i'_1,\ldots,i'_d\\j'_1,\ldots,j'_d}} \alpha_{i'_1,i_1}\cdots \alpha_{i'_d,i_d} \cdot \alpha_{j'_1,\sigma(j_1)}\cdots\alpha_{j'_d,\sigma(j_d)}\Big)
x_{i'_1,j'_1} \cdots  x_{i'_d,j'_d}
\\
&=
\sum_{\sigma\in\aS_d}
\Big(\sum_{i'_1,j'_1} \alpha_{i'_1,i_1}\alpha_{j'_1,j_{\sigma(1)}}  x_{i'_1,j'_1}\Big)
\cdots
\Big(\sum_{i'_d,j'_d} \alpha_{i'_d,i_d}\alpha_{j'_d,j_{\sigma(d)}}  x_{i'_d,j'_d}\Big)
\\
&=
P\Big( \sum_{\sigma\in\aS_d} x_{i_1,j_{\sigma(1)}}\cdots x_{i_d,j_{\sigma(d)}} \Big)
\\
&=
P((x_{i_1}\cdots x_{i_d})\boxtimes (x_{j_1}\cdots x_{j_d}))
\\
&=
P(f\boxtimes h).\qedhere
\end{align*}
\end{proof}

For every commutative ring $R$ we note that
\begin{equation}\label{eq:detdet}
\det_d \leq \det_d\boxtimes \det_d
\end{equation}
as follows.
We first recall $\mon_d \leq \det_d$ from \eqref{eq:monleqdet}.
Using Proposition~\ref{pro:boxtimesequivariance} we obtain $\mon_d \boxtimes \det_d \leq \det_d\boxtimes \det_d$,
which is a polynomial in triple-indexed variables.
We then map every $x_{i,j,k}$ to zero for which $j\neq k$, and we map $x_{i,j,k}$ to $x_{i,j}$ otherwise.
The result is $\det_d$, which proves \eqref{eq:detdet}.
In characteristic 0, we can alternatively prove \eqref{eq:detdet} by restricting $\det_d$ to $\mon_d$
and further to $x^d$ and then observe that $x^d \boxtimes \det_d = d!\,\det_d$.

\begin{example}
\label{exa:boxtimes}
There are several examples of polynomials in the literature that are of the form $f\boxtimes h$,
for example the following (whenever the determinant appears in the subsequent examples, we assume we are working over a ring~$R$, not just a semiring).
\begin{enumerate}
\item The permanent polynomial: $\per_d := \mon_d\boxtimes \mon_d$, see e.g., \cite{Val79}, \cite[eq.~(3.1)]{Min84}.
\item The hyperpermanent (also called the multidimensional permanent): $\mon_d^{\boxtimes m}$, see e.g., \cite{DG87,Tic15,CP16}.
\item The rectangular permanent: $\mon_d\boxtimes e_{n,d}$, which is just called the permanent in \cite[eq.~(1.1)]{Min84}, see also \cite[Exa~4.6]{BE19}.
\item The mixed discriminant: $\det_d\boxtimes \mon_d$, see e.g., \cite{LC83,Pan85,Bap89,Gur05,CP16}.
\item Cayley's first hyperdeterminant: $\det_d^{\boxtimes m}$, \cite{Cay46, Bar95, HL13, CP16, BI17, AY23}, also called the combinatorial hyperdeterminant, also called the Pascal determinant in \cite[\S8.3]{Lan11} and \cite{Lan15}, or the quantum permanent in \cite[\S3]{Gur04}.
\end{enumerate}
\begin{remark}
Recall $\mon_d\leq \det_d$ from \eqref{eq:monleqdet}.
Using Proposition~\ref{pro:boxtimesequivariance} it follows $\mon_d\boxtimes \mon_d \leq \mon_d\boxtimes\det_d \leq \det_d\boxtimes\det_d$.
These two restrictions are separate proofs in the literature, see \cite{Gur04,Gur05}.
Also note that this implies $\w(\mon_d\boxtimes \mon_d)\leq \w(\mon_d\boxtimes\det_d) \leq \w(\det_d\boxtimes\det_d)$.
\end{remark}
\end{example}

\subsection{The Kronecker action}
\label{subsec:kronaction}
In this section we generalize the $\boxtimes$ construction from \S\ref{subsec:boxtimes} from $\sS$ to $\sT$ and $\sW$.
It is unclear if there is any reasonable Kronecker product construction for two elements in $\sW$ with output in $\sW$, so we take a different approach.
The Kronecker product of two polynomials is replaced by a product of a symmetric tensor with an algebra element, i.e., an action of the algebra of order $d$ symmetric tensors on $\sT_d$ or $\sW_d$ or $\sS_d$.
On $\sS_d$ this gives the same construction as in \S\ref{subsec:boxtimes}, see Claim~\ref{cla:boxtimesisboxtimes}.
We partially follow the exposition from \cite{IOT25}.

\begin{claim}\label{cla:coinvfunct}
Let $R$ be a commutative semiring.
Let $\varphi:\sT_d\to \sT_d$ be an $R$-linear $\aS_d$-equivariant map.
Then $\varphi$ induces an $R$-linear map $\psi_{S}:\sS_d\to \sS_d$, $\psi([f]_{\sS})=[\varphi(f)]_{\sS}$.
Moreover, $\varphi$ induces an $R$-linear map $\psi_{\sW}:\sW_d\to \sW_d$, $\psi([f]_{\sW})=[\varphi(f)]_{\sW}$.
\end{claim}
\begin{proof}
We verify the well-definedness of~$\psi_{\sS}$.
Let $f,f'\in \sT_d$ with $[f]_{\sS}=[f']_{\sS}$ be arbitrary, i.e., $f\sim f'$.
The task is to show $\varphi(f)\sim\varphi(f')$.
We first assume that $(f,f')$ is a generator of $\sim$, i.e., we have
$f = \pi f'$.
Applying $\varphi$ gives $\varphi(f) = \varphi(\pi f') = \pi \varphi(f')$,
which is a generator of $\sim$ on $W$, thus, $\varphi(f)\sim \varphi(f')$.
More generally, consider the case where $f$ and $f'$ are not necessarily generators, but
$f\sim f'$, $(f,f') = \sum_i\alpha_i (h_i,h'_i)$, where $\alpha_i\in R$ and $(h_i,h'_i)$ are generators of $\sim$.
Then $(\varphi(f),\varphi(f')) = \sum_i\alpha_i \underbrace{(\varphi(h_i),\varphi(h'_i))}_{\in \sim} \in \sim$. In other words, $\varphi(f)\sim\varphi(f')$.

We now verify the well-definedness of~$\psi_{\sW}$ in a similar way.
Let $f,f'\in \sT_d$ with $[f]_{\sW}=[f']_{\sW}$ be arbitrary, i.e., $f\backsim f'$.
The task is to show $\varphi(f)\backsim\varphi(f')$.
We have
\[
f + h + n \otimes \ell \otimes \ell \otimes m = f' + h + n' \otimes \ell' \otimes \ell' \otimes m'.
\]
Thus
\[
\varphi(f) + \varphi(h) + \varphi(n) \otimes \varphi(\ell) \otimes \varphi(\ell) \otimes \varphi(m) = \varphi(f') + \varphi(h) + \varphi(n') \otimes \varphi(\ell') \otimes \varphi(\ell') \otimes \varphi(m'),
\]
which shows $\varphi(f)\backsim\varphi(f')$.
\end{proof}

We have an $R$-bilinear map $\sT_1\times \sT_1\to\sT_1$, $(v_1,v_2)\mapsto v_1\boxtimes v_2$, where $v_1,v_2\in\Var$, see \S\ref{subsec:variables}.
This lifts to an $R$-bilinear map $\boxtimes:\sT_d\times\sT_d\to\sT_d$.
Note that
\begin{equation}\label{eq:pifhpifpih}
\pi(f\boxtimes h) = (\pi f)\boxtimes (\pi h).
\end{equation}
Let $\psi:\Symt_d \times \sT_d \to \sT_d$ be the restriction of this map to $\Symt_d \times \sT_d$.
Let $A\in\{\sT,\sS,\sW\}$.
Let $s\in \Symt_d$. Then there exists a linear map $\psi'_s:\sT_d\to \sT_d, \ f \mapsto \psi(s,f)$
that is $\aS_d$-equivariant: $\psi'_s(\pi f) = \psi(s,\pi f) = \psi(\pi s,\pi f) \stackrel{\eqref{eq:pifhpifpih}}{=} \pi \psi(s,f) = \pi \psi'_s(f)$.
We write $[f]$ for the image of $f\in\sT_d$ in $A_d$.
We apply Claim~\ref{cla:coinvfunct} (in the case $A_d=\sT_d$ there is no need to invoke Claim~\ref{cla:coinvfunct}) to obtain a map $\psi_s : A_d\to A_d$, $[f]\mapsto [\psi(s,f)]$.
The well-defined map
\[
\kappa: \Symt_t \to \Hom(A_d, A_d), \qquad s \mapsto \psi_s
\]
is $R$-linear, because
$(\kappa(s+\alpha t))([f])
=
\psi(s+\alpha t,f)
=
\psi(s,f) + \alpha \psi(t,f)
=
\kappa(s)([f])+\alpha\kappa(t)([f])
=
(\kappa(s)+\alpha\kappa(t))([f]).
$
Altogether, this gives an $R$-bilinear map
\begin{equation}\label{eq:ASBSCS}
\boxtimes : \Symt_d \times A_d \to A_d, \qquad (s,[f]) \mapsto \psi_s([f]) = [\psi(s,f)].
\end{equation}

We now interpret this over the standard basis. We treat this for variables $x_1,x_2,\ldots$, as the more general case is analogous, but notationally unpleasant.
For a tensor $t\in\sT_d$, define the stabilizer $\stab(t)\subseteq\aS_d = \{\pi\in\aS_d\mid \pi t = t\}$.
Let $\aS_d/\stab(t) = \{\tau \pi \mid \tau\in\aS_d, \, \pi \in \stab(t)\}$ denote the set of cosets.
Then $\sigma t := \tau t$ is well-defined for $\sigma\in\aS_d/\stab(t)$ with coset representative $\tau$.
Let $t_{\textup{sym}}:=\sum_{\sigma \in \aS_d/\stab(t)} \sigma t$,
which is an element in $\Symt_d$.
The $R$-semimodule $\Symt_d$ of symmetric tensors is freely generated by the tensors of the form
\[
(x_1^{\otimes \la_1} \otimes x_2^{\otimes \la_2} \otimes \cdots)_{\textup{sym}},
\]
where $\sum_i\la_i = d$.
For a list $\la$ with $\sum_i\la_i = d$ let $\aS_d(1^{\la_1}2^{\la_2}\cdots)$ denote the orbit
of the length $d$ list 
\begin{equation}
\label{eq:list}
(\underbrace{1,1,\ldots,1}_{\la_1 \textup{ many}},\underbrace{2,2,\ldots,2}_{\la_2 \textup{ many}},\ldots),
\end{equation}
i.e., the set of all length $d$ lists that arise form \eqref{eq:list} by permuting the positions of the numbers.
For example, $\aS_4(1^2 2^2) = \{(1,1,2,2),(1,2,1,2),(1,2,2,1),(2,1,1,2),(2,1,2,1),(2,2,1,1)\}$.
Let $A \in \{\sT,\sS,\sW\}$.
For a tensor $t \in \sT_d$, let $[t]$ denote its equivalence class in $A_d$.
For a symmetric tensor $s\in \Symt_d$ and a $f\in A_d$,
\eqref{eq:ASBSCS} gives an interpretation on the standard basis as
\begin{equation}
\label{eq:boxtimesaction}
(x_1^{\otimes \la_1} \otimes x_2^{\otimes \la_2} \otimes \cdots)_{\textup{sym}} \boxtimes
[x_{i_1}\otimes \cdots \otimes x_{i_d}] \ := \ 
\sum_{(j_1,\ldots,j_d)\in\aS_d(1^{\la_1}2^{\la_2}\cdots)} [x_{j_1,i_1} \otimes \cdots \otimes x_{j_d,i_d}],
\end{equation}

\begin{claim}\label{cla:boxtimesisboxtimes}
For two polynomials $h,f\in \sS_d$, we have
$h \boxtimes f = \pol(h) \boxtimes f$,
where $h \boxtimes f$ is defined in \eqref{eq:defboxtimes}.
\end{claim}
\begin{proof}
It suffices to prove this for $h=[x_{i_1} \otimes\cdots\otimes x_{i_d}]_S$ and $f=[x_{j_1} \otimes\cdots\otimes x_{j_d}]_S$,
because the statement then follows from bilinearity.
Define $\la$ via $\la_j = |\{i_k = j\}|$.
We have $h = [x_1^{\otimes \la_1} \otimes x_2^{\otimes \la_2}\otimes \cdots]_S$.
We have
\[
\pol(h) = \sum_{\pi\in\aS_d} x_{i_{\pi(1)}} \otimes\cdots\otimes x_{i_{\pi(d)}}
= (\la_1!\,\la_2!\cdots)! \, (x_1^{\otimes\la_1}\otimes x_2^{\otimes\la_2}\otimes\cdots)_\sym
\]
\eqref{eq:boxtimesaction} gives
\[
\pol(h)\boxtimes f = 
(\la_1!\,\la_2!\cdots)! \sum_{(j_1,\ldots,j_d)\in\aS_d(1^{\la_1}2^{\la_2}\cdots)} [x_{j_1,i_1} \otimes \cdots \otimes x_{j_d,i_d}],
\]
which according to \eqref{eq:defboxtimesII} coincides with~$h\boxtimes f$.
\end{proof}

Note that there are interesting polynomials $s\boxtimes f$ for which $s\in\Symt_d$ is \emph{not} a polarization.
For example, the polynomial
$(x^{\otimes i}y^{\otimes d-i})_\sym \boxtimes \det_d$ appears in \cite{IOT25}.

The following claim is the generalization of \eqref{eq:kronprodassoc}.
It states that the semialgebra $(\Symt_d,+,\boxtimes)$
acts on $A_d$.
\begin{claim}
\label{cla:boxtimesaction}
Let $R$ be a commutative semiring, and let $A\in\{\sT,\sS,\sW\}$.
Let 
$s',s\in\Symt_d$ and $f\in A_d$. Then
$(s'\boxtimes s) \boxtimes f = s' \boxtimes (s \boxtimes f)$.
Here, $s'\boxtimes s$ is defined as in \eqref{eq:boxtimestensors}.
\end{claim}
\begin{proof}
Let $[f]$ denote the coset of $f\in\sT_d$ in $A_d$.
We prove the statement for
$s'=(x_1^{\otimes\la_1}\otimes x_2^{\otimes\la_2}\otimes\cdots)_\sym$,
$s=(x_1^{\otimes\mu_1}\otimes x_2^{\otimes\mu_2}\otimes\cdots)_\sym$,
$f=[x_{i_1} \otimes \cdots \otimes x_{i_d}]$.
We observe
\[
s'\boxtimes s = \sum_{\substack{(j'_1,\ldots,j'_d)\in\aS_d(1^{\la_1}2^{\la_2}\cdots)\\(j_1,\ldots,j_d)\in\aS_d(1^{\mu_1}2^{\mu_2}\cdots)}} x_{j'_1,j_1} \otimes \cdots \otimes x_{j'_d,j_d}.
\]
According to\eqref{eq:boxtimesaction} we have
\[
(s'\boxtimes s) \boxtimes f = \sum_{\substack{(j'_1,\ldots,j'_d)\in\aS_d(1^{\la_1}2^{\la_2}\cdots)\\(j_1,\ldots,j_d)\in\aS_d(1^{\mu_1}2^{\mu_2}\cdots)}} [x_{j'_1,j_1,i_1}\otimes\cdots\otimes x_{j'_d,j_d,i_d}].
\]
On the other hand,
\[
s\boxtimes f = \sum_{(j_1,\ldots,j_d)\in\aS_d(1^{\mu_1}2^{\mu_2}\cdots)} [x_{j_1,i_1}\otimes\cdots\otimes x_{j_d,i_d}].
\]
Moreover,
\[
s'\boxtimes(s\boxtimes f) \ = \ 
\sum_{\substack{(j'_1,\ldots,j'_d)\in\aS_d(1^{\la_1}2^{\la_2}\cdots)\\(j_1,\ldots,j_d)\in\aS_d(1^{\mu_1}2^{\mu_2}\cdots)}}
[x_{j'_1,j_1,i_1}\otimes\cdots\otimes x_{j'_d,j_d,i_d}].
\]
Hence, $(s'\boxtimes s) \boxtimes f=s'\boxtimes(s\boxtimes f)$.
\end{proof}

Note that for $s\in \Symt_d\subseteq\sT_d$ and $T\in\End$ we have $T(s)\in \Symt_d$,
in other words, $\End$ acts linearly on $\Symt_d$.
We write $s\leq t$ if there exists $T\in\End$ such that $T(t)=s$.
The next proposition is the generalized variant of Proposition~\ref{pro:boxtimesequivariance}.
\begin{proposition}
\label{pro:boxtimesequivarianceaction}
Let $s,t\in \Symt_d$ and $f,g \in A_d$.
If $s\leq t$ and $f \leq h$, then
$s\boxtimes f \leq t\boxtimes h$.
\end{proposition}
\begin{proof}
Without loss of generality, we consider $s,t,f,h$ in variables $x_1,x_2,\ldots$, so that $s\boxtimes f$ and $t\boxtimes h$ are double-indexed.
Given $E_1\in\End$ with $E_1(t)=s$ and
given $E_2\in\End$ with $E_2(h)=t$,
we create $E_3\in\End$ via $E_3(x_{j,i})=\sum_{j',i'} \alpha_{j',j}\,\beta_{i',i}\,x_{j',i'}$,
where $E_1(x_j)=\sum_{j'}\alpha_{j',j} \, x_{j'}$ and $E_2(x_i)=\sum_{i'}\beta_{i',i} \, x_{i'}$.
Let
$t = \sum_a \iota_a \, x_{j_{1,a}}\otimes \cdots \otimes x_{j_{d,a}} \in \Symt_d$ be arbitrary.
Then
$E_1(t) = \sum_a \iota_a(\sum_{j'_1}\alpha_{j'_1,j_{1,a}} x_{j'_1})  \otimes \cdots \otimes (\sum_{j'_d}\alpha_{j'_d,j_{d,a}} x_{j'_d})
= \sum_a \iota_a\sum_{j'_1,\ldots,j'_d} (\alpha_{j'_1,j_{1,a}}\cdots\alpha_{j'_d,j_{d,a}}) (x_{j'_1}\otimes\cdots\otimes x_{j'_d})
$.
Let
$h = \sum_b \kappa_b [x_{i_{1,b}}\otimes \cdots \otimes x_{i_{d,b}}]$ be arbitrary.
Then
$E_2(h)
=
\sum_b \kappa_b [(\sum_{i'_1}\beta_{i'_1,i_{1,b}} \, x_{i'_1})\otimes \cdots \otimes (\sum_{i'_d}\beta_{i'_d,i_{d,b}} \, x_{i'_d})]
=
\sum_b \kappa_b\sum_{i'_1,\ldots,i'_d}(\beta_{i'_1,i_{1,b}}\cdots \beta_{i'_d,i_{d,b}}) [x_{i'_1}\otimes \cdots \otimes x_{i'_d}]
$.
Moreover, $t\boxtimes h = \sum_{a,b} \iota_a \kappa_b [x_{j_{1,a},i_{1,b}}\otimes\cdots\otimes x_{j_{d,a},i_{d,b}}]$.
We now calculate

$E_1(t)\boxtimes E_2(h) = \sum_{\substack{j'_1,\ldots,j'_d\\i'_1,\ldots,i'_a}}
\sum_{a,b} \iota_a \kappa_b (\alpha_{j'_1,j_{1,a}}\cdots\alpha_{j'_d,j_{d,a}})
(\beta_{i'_1,i_{1,b}}\cdots \beta_{i'_d,i_{d,b}})
[x_{j'_1,i'_1}\otimes \cdots \otimes x_{j'_d,i'_d}]
=
\sum_{a,b} \iota_a \kappa_b [
(\sum_{j'_1,i'_1} \alpha_{j'_1,j_{1,a}}\,\beta_{i'_1,i_{1,b}}\,x_{j'_1,i'_1})
\otimes\cdots\otimes
(\sum_{j'_d,i'_d} \alpha_{j'_d,j_{d,a}}\,\beta_{i'_d,i_{d,b}}\,x_{j'_d,i'_d})
]
=
E_3(t\boxtimes h)$.
\end{proof}

\subsection{Higher order computation}
\label{subsec:higherordercomp}

We use 0-based indexing of the rows and columns for the following discussion.
Let $R$ be a commutative semiring, let $U,V,W$ be $R$-semimodules, and let $\varphi:U\times V\to W$ be $R$-bilinear.
Let $E\in\Mat_{\ell}(U)$ and $F\in\Mat_{m}(V)$.
The Kronecker product $E\boxtimes F\in\Mat_{\ell m}(W)$ with respect to $\varphi$ is defined as
\begin{equation}\label{eq:kronactionmatrix}
[E\boxtimes F]_{i m + i',j m + j'} := \varphi([E]_{i,j},[F]_{i',j'}).
\end{equation}
In the following Claim~\ref{cla:boxtimeskronprod}, we lift Claim~\ref{cla:boxtimesaction} to matrices.
\begin{claim}\label{cla:boxtimeskronprod}
Let $R$ be a commutative semiring, and let $A\in\{\sT,\sS,\sW\}$.
Let $S'\in\Mat_m(\Symt_d)$, $S\in\Mat_\ell(\Symt_d)$, $F\in\Mat_k(A_d)$.
Then
$(S'\boxtimes S) \boxtimes F = S' \boxtimes (S \boxtimes F)$,
where the bilinear maps used for the two Kronecker products of matrices are defined in \eqref{eq:boxtimestensors} and \eqref{eq:defboxtimes}.
\end{claim}
\begin{proof}\mbox{~}

\vspace{-1.4cm}

\begin{eqnarray*}
[(S'\boxtimes S) \boxtimes F]_{k\ell i + k i' + i'',k\ell j + k j' + j''}
&=&
[(S'\boxtimes S) \boxtimes F]_{k(\ell i + i') + i'',k(\ell j + j') + j''}
\\
&=&
[S'\boxtimes S]_{\ell i + i',\ell j + j'} \boxtimes [F]_{i'',j''}
\\
&=&
([S']_{i,j}\boxtimes [S]_{i',j'}) \boxtimes [F]_{i'',j''}
\\
&\stackrel{\textup{Cla.}~\ref{cla:boxtimesaction}}{=}&
[S']_{i,j}\boxtimes ([S]_{i',j'} \boxtimes [F]_{i'',j''})
\\
&=&
[S']_{i,j}\boxtimes ([S\boxtimes F]_{\ell i'+i'',\ell j'+j''})
\\
&=&
[S'\boxtimes(S\boxtimes F)]_{k\ell i+(\ell i'+i''),k\ell j + (\ell j'+j'')}
\\
&=&
[S'\boxtimes(S\boxtimes F)]_{k\ell i+\ell i'+i'',k\ell j + \ell j'+j''}
\end{eqnarray*}

\vspace{-0.7cm}

\end{proof}

Let $R$ be a commutative semiring, $A=\sS$.
We define $\pol:\Mat(\sS_d)\to\Mat(\Symt_d)$ entrywise, i.e., $[\pol(F)]_{i,j}:=\pol(F_{i,j}).$
We define the matrix of symmetric tensors
\begin{equation}\label{eq:defsfHnd}
\textsf{H}_{n,d} := \pol(\Xi_{n,d}) \ \in \Mat_n(\Symt_d)
\end{equation}
and its bottom-right entry
\begin{equation}\label{eq:defnonsfHnd}
H_{n,d} \ := \ \pol(\Gamma\!_{n,d}) \ \in \Symt_d.
\end{equation}
We keep this specific matrix of symmetric tensors also when studying $A_d\in\{\sT_d,\sW_d\}$.

Now, let $A\in\{\sS,\sT,\sW\}$.
We define the \emph{hypercomputant matrix} as
\[
\hXi_{n,d} \ := \ \textsf{H}_{n,d} \boxtimes \Xi_{n,d}
\]
and its bottom-right entry $\hGamma_{n,d}\in A_d$ as
\[
\hGamma\!_{n,d} \ := \ H_{n,d} \boxtimes \Gamma\!_{n,d}.
\]
\begin{remark}
We have
\begin{equation}\label{eq:GammaXi}
\w(\hGamma\!_{n,d}) \leq \w(\hXi_{n,d}) \leq n^4\w(\hGamma\!_{n,d}),
\end{equation}
which can be seen by constructing an ABP for each of the $n^4$ entries of $\hXi_{n,d}$ separately,
because all entries of the matrix $\hXi_{n,d}$ are the same up to renaming variables.
\end{remark}

\begin{remark}
Note that if $A=\sS$, then $\hGamma\!_{n,d} = \Gamma\!_{n,d}^{\,\boxtimes 2}$.
In $\sW_d$ we cannot write $\Gamma\!_{n,d}^{\,\boxtimes 2}$, because there is no such multiplication. In $A_d=\sT_d$ we have $\hGamma\!_{n,d}\neq \Gamma\!_{n,d}^{\,\boxtimes 2} = \Gamma\!_{n^2,d}$, which is a less interesting tensor for our purposes.
\end{remark}

Combining \eqref{eq:Gammaij}, \eqref{eq:boxtimesaction}, and \eqref{eq:kronactionmatrix}, we see that
\[
[\hXi_{n,d}]_{an+a',bn+b'} \ = \ \sum_{\pi\in\aS_d}\sum_{\substack{
1 \leq i_1,\ldots,i_{d-1} \leq n
\\
1 \leq j_1,\ldots,j_{d-1} \leq n
}}
(x^{(1)}_{a,i_1}\boxtimes x^{(\pi(1))}_{a',j_1})
\cdot
(x^{(2)}_{i_1,i_2}\boxtimes x^{(\pi(2))}_{j_1,j_2})
\cdots
(x^{(d)}_{i_d,b}\boxtimes x^{(\pi(d))}_{j_d,b'})
\]

We end this subsection with some observations.

\begin{claim}
$\Xi_{n^2,d}\leq \hXi_{n,d}$ and $\Gamma\!_{n^2,d}\leq \hGamma\!_{n,d}$.
\end{claim}
\begin{proof}
We prove the claim for $\Gamma$, the proof for $\Xi$ being completely analogous.
For this proof, we let the variables $x_{i,j}^{(k)}$ in $\Gamma\!_{n,d}$ be indexed at zero, i.e., $0\leq i<n$ and $0\leq j<n$.
Consider $\hGamma\!_{n,d}$ and map the variables $x_{i,j}^{(k)}\boxtimes x_{i',j'}^{(k')}$ to zero if $k\neq k'$.
Then rename the variables $x_{i,j}^{(k)}\boxtimes x_{i',j'}^{(k)}$
to $x_{in+i',jn+j'}^{(k)}$ to obtain $\Gamma\!_{n^2,d}$.
\end{proof}

\begin{claim}
\label{cla:whGammand}
For fixed $d\geq 2$ we have $\w(\hXi_{n,d})=\Omega(n^2)$
and $\w(\hGamma_{n,d})=\Omega(n^2)$,
with the exception of $A_2=\sW_2$.
\end{claim}
\begin{proof}
This is analogous to the proof of Example~\ref{exa:Xiw}.
We start with the observation that for $\Xi_{n,d}$, $\hXi_{n,d}$, $\Gamma\!_{n,d}$, and $\hGamma\!_{n,d}$ the number of essential variables equals the number of variables used to define them,
which is seen as follows.
For every one of the $n^4 d^2$ variables $x_{\mathbf{i}}$
we can reconstruct $\mathbf i$
from any monomial in $\partial_{\mathbf i} \Xi_{n,d}$.
Hence $J(\hXi_{n,d})$ contains a $(n^4 d^2)\times(n^4 d^2)$ permutation matrix as a submatrix,
so the claim follows with Lemma~\ref{lem:rankofid}.
The proofs for $\Xi_{n,d}$, $\Gamma\!_{n,d}$, and $\hGamma\!_{n,d}$ are analogous.

The number of essential variables of $\Xi_{w,d}$ is $w^2 d$,
and the number of essential variables of $\hXi_{n,d}$ is $n^4 d^2$.
If $\w(\hXi_{n,d})\leq w$, then $\hXi_{n,d}\leq \Xi_{w,d}$,
and hence by \eqref{eq:Jmatrixformatrices} we have
$n^4 d^2\leq w^2 d$, thus $w\geq n^2 d^{1/2}$.

For the second part we use Claim~\ref{cla:widthfviaGamma}.
The number of essential variables of $\Gamma\!_{w,d}$ is $2w+(d-2)w^2$,
and the number of essential variables of $\hGamma\!_{n,d}$ is $(2n+(d-2)n^2)^2$.
If $\w(\hGamma\!_{n,d})\leq w$, then by Claim~\ref{cla:widthfviaGamma}
we have $\hGamma\!_{n,d}\leq \Gamma\!_{w,d}$,
hence by \eqref{eq:Jmatrix} we have
$(2n+(d-2)n^2)^2 \leq 2w+(d-2)w^2$.
Therefore, for fixed $d$ we have $w\in\Omega(n^2)$.
\end{proof}

\begin{example}
\label{exa:smallhyperwidth}
Let $A=\sS$.
We use Proposition~\ref{pro:boxtimesequivariance} to show that the examples in Example~\ref{exa:boxtimes} are restrictions of $\hGamma_{n,d}$ as follows.
\begin{itemize}
\item The permanent: $\mon_d\boxtimes\mon_d = \hGamma\!_{1,d}$
\item The hyperpermanent: $\mon_d^{\boxtimes m} = \Gamma\!_{1,d}^{\,\boxtimes m}$
\item The rectangular permanent: $\mon_d\boxtimes e_{n,d} \leq \hGamma\!_{n,d}$, because $e_{n,d} \leq \Gamma\!_{n,d}$, see Example~\ref{exa:elemsym}
\item The mixed discriminant: $\det_d\boxtimes\mon_d \leq \hGamma\!_{\w(\det_d),d}$.
\item Cayley's first hyperdeterminant: $\det_d^{\boxtimes m} \leq \Gamma\!_{\w(\det_d),d}^{\,\boxtimes m}$.
\end{itemize}
We will see more examples in~\S\ref{sec:higherorderwidth} and in~\S\ref{sec:boolean}.
\end{example}

\subsection{Valiant's conjecture}
\label{subsec:Valiantsconj}

Valiant's original conjecture ``$\VBP\neq\VNP$'' (over characteristic different from 2, this is also the \emph{determinant vs permanent} conjecture) is equivalent to the following conjecture.
\begin{conjecture}[Valiant's conjecture]\label{conj:flagship}
Let $R$ be a field and let $A=\sS$.
Valiant's conjecture for $R$ is
\[
\w(\hXi_{n,d})\notin\poly(n,d).\qedhere
\]
\end{conjecture}
By \eqref{eq:GammaXi}, Conjecture~\ref{conj:flagship} is equivalent to $\w(\hGamma\!_{n,d})\notin\poly(n,d)$.
Note that Conjecture~\ref{conj:flagship} depends on~$R$ and $A$ and the question can also be asked when $R$R is not a field and/or when $A\neq\sS$.
The case $R=\IC$ is discussed in \S\ref{sec:PNP}.
The statement $\w(\hGamma\!_{n,d})\notin\poly(n,d)$ is true for $A=\sT$ when $R$ is an arbitrary commutative semiring, see \S\ref{subsec:noncomm}.
The statement $\w(\hGamma\!_{n,d})\notin\poly(n,d)$ is also true for $A=\sS$ when $R$ is a zerosumfree commutative semiring, for example $R=\IR_{\geq 0}$, see \S\ref{subsec:monotone}.
So far there are no results involving $A=\sW$.
The following proposition shows that it is sufficient to study second-order computation only and not also higher orders.

\begin{proposition}\label{pro:collapse}
$\w(\hXi_{n,d})\in\poly(n,d)$ if and only if for every $m$ we have $\w(\textup{\textsf{H}}_{n,d}^{\boxtimes m}\boxtimes \Xi_{n,d})\in\poly(n,d)$.
\end{proposition}
\begin{proof}
One direction is clear.
For the other direction, fix any $m$ and consider $\textsf{H}_{n,d}^{\boxtimes m}\boxtimes \Xi_{n,d}$.
Let $k\in\IN$ such that $\forall d,n: \w(\hXi_{n,d})\leq n^k d^k$,
in other words $\hXi_{n,d}\leq \Xi_{n^k d^k,d}$.
We will mark with $(\ast)$ when we use $\textsf{H}_{n,d}\leq \textsf{H}_{n',d}$ for $n'\geq n$, which is obtained analogously to Claim~\ref{cla:GammvslargerGamma}.
\begin{eqnarray*}
\textsf{H}_{n,d}^{\boxtimes m}\boxtimes \Xi_{n,d}
&\stackrel{\textup{Cla.}~\ref{cla:boxtimesaction}}{=}&
\textsf{H}_{n,d}^{\boxtimes m-1}\boxtimes (\textsf{H}_{n,d}\boxtimes \Xi_{n,d})
\ \stackrel{\textup{Pro.}~\ref{pro:boxtimesequivarianceaction}}{\leq} \ 
\textsf{H}_{n,d}^{\boxtimes m-1}\boxtimes \Xi_{n^k d^k,d}
\\
& \stackrel{\textup{Cla.}~\ref{cla:boxtimesaction}}{=} & 
\textsf{H}_{n,d}^{\boxtimes m-2}\boxtimes (\textsf{H}_{n,d}\boxtimes \Xi_{n^k d^k,d})
\ \stackrel{(\ast)}{\leq} \ 
\textsf{H}_{n,d}^{\boxtimes m-2}\boxtimes (\textsf{H}_{n^k d^k,d}\boxtimes \Xi_{n^k d^k,d})
\\
&\stackrel{\textup{Pro.}~\ref{pro:boxtimesequivarianceaction}}{\leq}&
\textsf{H}_{n,d}^{\boxtimes m-2}\boxtimes \Xi_{n^{k^2} d^{(k+1)^2-1},d}
\ \leq \ 
\textsf{H}_{n,d}^{\boxtimes m-3}\boxtimes \Xi_{n^{k^3} d^{(k+1)^3-1},d}
\\
&\leq& \cdots \ \ \leq \ \Xi_{n^{k^m} d^{(k+1)^m-1},d} \ .
\end{eqnarray*}
Since $k$ and $m$ are constant, this means $\w(\textsf{H}_{n,d}^{\boxtimes m}\boxtimes \Xi_{n,d})\in\poly(n,d)$.
We remark that for $A=\sS$ a better construction would split $\Xi_{n,d}^{\boxtimes m}=\Xi_{n,d}^{\boxtimes m/2}\boxtimes \Xi_{n,d}^{\boxtimes m/2}$ recursively, but this does not work for $A\in\{\sT,\sW\}$.
\end{proof}

If Conjecture~\ref{conj:flagship} is false,
then from Example~\ref{exa:smallhyperwidth}
it follows that all examples in Example~\ref{exa:boxtimes}
have small width.
Moreover, all examples in the following \S\ref{sec:higherorderwidth} then also have small width.
More examples are given in \S\ref{sec:PNP}.

\subsection{Hamiltonian cycles and cliques}\label{sec:higherorderwidth}
In this section we prove that the Hamiltonian cycle polynomial and the Clique polynomial are restrictions of $\hGamma\!_{n,d}$.
A general construction that gives the polynomials only up to a factor of low degree is presented in \S\ref{sec:boolean}.

\subsubsection*{Hamiltonian cycles}
Let $R$ be a semiring, and let $A\in\{\sT,\sS,\sW\}$.
We define the Hamiltonian cycle algebra element, see \cite{Val79} where the variant without superscripts is discussed:
\begin{eqnarray*}
\HC_d &:=& \sum_{\substack{\phi\in\aS_{d-1}}} [x^{(1)}_{d,\phi(1)} \otimes x^{(2)}_{\phi(1),\phi(2)} \otimes\cdots\otimes x^{(d)}_{\phi(d-1),d}]
\ \in A_d.
\end{eqnarray*}

Note $H_{1,d}=\pol(\mon_d)$ up to renaming of variables.
\begin{proposition}
\label{pro:HCleqGamma}
$\HC_d \ \leq \ H_{1,d}\boxtimes\Gamma\!_{d,d} \ \leq \ \hGamma\!_{d,d}$.
\end{proposition}
\begin{proof}
The second restriction is obvious. We treat the first restriction.
We start with writing $\Gamma\!_{d,d} \in A_d$ as
\[
\Gamma\!_{d,d} \ = \ \sum_{\phi:[d-1]\to[d]} [x^{(1)}_{d,\phi(1)} \otimes x^{(2)}_{\phi(1),\phi(2)} \otimes x^{(3)}_{\phi(2),\phi(3)} \otimes\cdots\otimes x^{(d)}_{\phi(d-1),d}].
\]
Hence,
\[
\pol(\mon_d)\boxtimes \Gamma\!_{d,d} = \sum_{\substack{\phi:[d-1]\to[d]\\\pi\in\aS_d}} [x_{\pi(1),d,\phi(1)}^{(1)} \otimes x_{\pi(2),\phi(1),\phi(2)}^{(2)} \otimes\cdots\otimes x_{\pi(d),\phi(d-1),d}^{(d)}].
\]
Let $T\in\End$ as follows:
$T(x^{(k)}_{a,b,c})=x^{(k)}_{b,c}$ if $a=c$, and $T(x^{(k)}_{a,b,c})=0$ otherwise.
We calculate
\[
T(\pol(\mon_d)\boxtimes \Gamma\!_{d,d})
=
\sum_{\substack{\phi\in\aS_d\\\phi(d)=d}} [x_{d,\phi(1)}^{(1)} \otimes x_{\phi(1),\phi(2)}^{(2)} \otimes\cdots\otimes x_{\phi(d-1),d}^{(d)}]
\ = \ \HC_d.
\qedhere
\]
\end{proof}

The statement ``$\w(\HC_d) \notin \poly(d)$ over $\sS$ for a commutative semiring $R$'' is equivalent to Conjecture~\ref{conj:flagship},
which was proved in~\cite{Val79}, where it is stated only for fields~$R$.
The corresponding Hamiltonian cycle decision is problem 9 of Karp's original 21 NP-complete problems~\cite{Karp72}.

\subsubsection*{Cliques}
Let $R$ be a commutative semiring, and let $A=\sS$.
One can find different variants of the so-called Clique polynomial in the literature,
for example
the variant in \cite{Sch76} is uncolored without loops,
the variant in \cite{BE19} is uncolored with loops,
the variant in \cite{KPR23} is colored without loops.
We use a colored variant with loops:
\begin{equation}
\label{eq:defclique}
{}_k\Clique_{n}
\ := \ \sum_{1 \leq v_1 < v_2 < \cdots < v_k \leq n} \, \prod_{\substack{1\leq i\leq j\leq k}} x_{v_i,i,j,v_j}
\ \in \ \sS_{\binom{k+1}{2}}.
\end{equation}
For convenience, we also use the following notation.
\[
\Cliqued_{n,d} := \begin{cases}
{}_k\Clique_{n} & \textup{if } d = \binom{k+1}{2},
\\
0 & \textup{otherwise}.
\end{cases}
\]

\begin{proposition}\label{pro:wclique}
$\Cliqued_{n,d} \leq \hGamma\!_{n,d}$.
\end{proposition}
\begin{proof}
Let $d := \binom{k+1}{2}$, because otherwise there is nothing to prove.
For additional clarity, we use the auxiliary variable names
$x_{a,(i,j)} := x_{a,i,j}$ and $x_{(i,j),b} := x_{i,j,b}$,
and we write $x_{a,(i,j),(i',j'),b} := x_{a,(i,j)}\boxtimes x_{(i,j),b}$.
The variable $x_{a,(i,j),(i,j),b}$ will later be mapped to the variable $x_{a,i,j,b}$ in the clique polynomial.
We show that
\begin{equation}\label{eq:cliquehGamma}
{}_k\Clique_{n} \ \leq \ 
\underbrace{\left(\sum_{1 \leq a_1, \ldots, a_k \leq n}\,\prod_{1 \leq i \leq j \leq k} x_{a_i,(i,j)} \right)}_{=:\,h_{n,d}}
\boxtimes
\underbrace{\left(\sum_{1 \leq b_1 < \cdots < b_k \leq n}\,\prod_{1 \leq i \leq j \leq k} x_{(i,j),b_j} \right)}_{=:\,f_{n,d} \ = \ {}_k f_n}
\ \leq \ \hGamma\!_{n,d}
\end{equation}
in a two-step process (note the strict inequality between the $b_j$, whereas the $a_i$ can be in any order).
We first observe that
\[
h_{n,d} \ = \
\prod_{i=1}^k \, \sum_{a_i=1}^n \, \prod_{i \leq j\leq k} x_{a_i,(i,j)},
\]
which immediately shows $\w(h_{n,d})\leq n$, see Claim~\ref{cla:duality}.
Moreover,
\[
{}_k f_{n} = \left(\prod_{i=1}^{k} x_{(i,k),n}\right) \cdot {}_{k-1}f_{n-1} \, + \, {}_k f_{n-1},
\]
which shows that $\w(f_{n,d})\leq n$.
These two width bounds together prove the right inequality of \eqref{eq:cliquehGamma}, see Proposition~\ref{pro:boxtimesequivariance}.
We now prove the left inequality of \eqref{eq:cliquehGamma}.
Let $T\in\End$ with $T(x_{a,(i,j),(i',j'),b})=0$ if $i\neq j'$ or $j\neq j'$,
and $T(x_{a,(i,j),(i,j),b})=x_{a,(i,j),(i,j),b}$.
We expand
\[
T\left(\left(\sum_{1 \leq a_1 ,\ldots, a_k \leq n}\,\prod_{1 \leq i \leq j \leq k} x_{a_i,(i,j)} \right)
\boxtimes
\left(\sum_{1 \leq b_1 < \cdots < b_k \leq n}\,\prod_{1 \leq i \leq j \leq k} x_{(i,j),b_j} \right)\right)
\]
\[
\mbox{~}\qquad =
\sum_{\substack{1 \leq a_1 , \ldots , a_k \leq n\\1 \leq b_1 < \cdots < b_k \leq n}}
\,
\prod_{1 \leq i \leq j \leq k} x_{a_i,(i,j),(i,j),b_j}
\]
Note that every monomial involves a variable $x_{a_i,(i,i),(i,i),b_i}$ for each $i$.
Therefore, if we map every variable $x_{b,(i,i),(i,i),b'}$ to zero for which $b\neq b'$,
then for each monomial that is not mapped to zero we have $\forall i:a_i=b_i$.
so that all variables are of the form $x_{b_i,(i,j),(i,j),b_j}$.
The resulting polynomial is
\[
\sum_{1 \leq b_1 < \cdots < b_k \leq n}
\,
\prod_{1 \leq i \leq j \leq k} x_{b_i,(i,j),(i,j),b_j}.
\]
Renaming $x_{b,(a,c),(a,c),b'}$ to $x_{b,a,c,b'}$ gives the desired ${}_k\Clique_{n}$.
\end{proof}

The clique polynomial was the first polynomial family that was proved $\VW[1]$-complete, see \cite{BE19}.
The corresponding clique problem is problem 3 of Karp's original 21 NP-complete problems~\cite{Karp72},
and its parameterized version is $\W[1]$-complete, see e.g.\ \cite[Thm~21.2.5]{DF13}.

\section{Exponents}\label{sec:exponents}

In this section we analyze the growth of $\w(\hXi_{n,d})$ and $\w(\hGamma\!_{n,d})$ for fixed~$d$ and growing~$n$.

\begin{definition}%
\label{def:hypercompexpo}
We define
\begin{align*}
\xi_d &:= \limsup_{n\to\infty}(\log_n(\w(\hXi_{n,d}))) = \inf\{r\mid \w(\hXi_{n,d}\in O(n^r)\},
\\
\gamma_d &:= \limsup_{n\to\infty}(\log_n(\w(\hGamma\!_{n,d}))) = \inf\{r\mid \w(\hGamma\!_{n,d}\in O(n^r)\}.
\end{align*}
These definitions depend on $R$ and $A$.
\end{definition}
By \eqref{eq:GammaXi} we have
\begin{equation}\label{eq:gammaxifourplusgamma}
\gamma_d\leq\xi_d\leq 4+\gamma_d.
\end{equation}

For algebra elements in $A_d$ whose variables have two superscripts,
for each monomial the superscript can be listed as a matrix of nonnegative integers called the superscript matrix,
whose entry in position $(i,j)$ is the number of times a variable with superscript $(i),(j)$ appears in the monomial.
Note that all superscript matrices for $\hXi_{n,d}$ are permutation matrices corresponding to $\pi\in\aS_d$.
For a fixed $\pi\in\aS_d$, the result of setting all monomial with superscript matrix $\neq \pi$ to zero
is denote by $\hXi_{n,d,\pi}$.
By construction, we have
\begin{equation}\label{eq:hXisumhXipi}
\hXi_{n,d}=\sum_{\pi\in\aS_d}\hXi_{n,d,\pi}
\end{equation}
Therefore,
\begin{equation}\label{eq:xidasmax}
\xi_d = \max\{ \inf\{\log_n(\w(\hXi_{n,d,\pi}))\} \mid \pi\in\aS_d\}.
\end{equation}
As usual, the lower right entries are denoted by $\hGamma\!_{n,d,\pi}$,
$\hGamma\!_{n,d}=\sum_{\pi\in\aS_d}\hGamma\!_{n,d,\pi}$,
$\gamma_d = \max\{ \inf\{\log_n(\w(\hGamma\!_{n,d,\pi}))\} \mid \pi\in\aS_d\}$.

\begin{remark}
Over $R=\IC$, the torus $T = (\IC\setminus\{0\})^{d\times d}$ acts on $\IC^{n^4 d^2}$ via
$\diag(\alpha_{1,1},\ldots,\alpha_{d,d}) x_{\mathbf i}^{(k)}\boxtimes x_{\mathbf j}^{(\ell)} = \alpha_{k,\ell} \, (x_{\mathbf i}^{(k)}\boxtimes x_{\mathbf j}^{(\ell)})$.
The irreducible polynomial representations of $T$ are indexed by $d\times d$ matrices of nonnegative integers.
The linear action of $T$ lifts to $A_d$ and $\Mat_n(A_d)$ in the usual way: $t[v_1\otimes\cdots\otimes v_d]=[t v_1\otimes\cdots\otimes t v_d]$
and entrywise for matrices.
We have that
$\hXi_{n,d,\pi}$ is the image of the projection of $\hXi_{n,d}$ to the $\pi$-isotypic component.
\end{remark}

Note that Claim~\ref{cla:whGammand} implies that
\begin{equation}\label{eq:atleasttwo}
\textup{$\xi_{d}\geq 2$
and $\gamma_{d}\geq 2$.}
\end{equation}
\begin{claim}\label{cla:nsquaredlowerbound}
$\w(\hXi_{n,d,\pi})\geq n^2$.
\end{claim}
\begin{proof}
The number of essential variables of $\hXi_{n,d,\pi}$ is $n^4 d$,
thus If $\w(\hXi_{n,d})\leq w$, then
$\hXi_{n,d}\leq \Xi_{w,d}$,
and hence by \eqref{eq:Jmatrixformatrices} we have
$n^4 d\leq w^2 d$, thus $w\geq n^2$.
\end{proof}
Since $\hXi_{n,d,\id} = \Xi_{n^2,d}$ up to renaming variables,
Claim~\ref{cla:nsquaredlowerbound} is the best lower bound we can get for $\w(\hXi_{n,d,\pi})$
if we do not take any information about $\pi$ into account.

\subsection{Tensors and flattenings}
\label{subsec:noncomm}

Given $d,j\in\IN$, $0\leq j \leq d$, there is a canonical linear isomorphism $\sT_d \to \big(\sT_j\big) \otimes \big(\sT_{d-j}\big)$.
Since we have a chosen basis and dual basis, we obtain a map
$\mathcal F_j: \sT_d \to \sT_j \otimes \big(\sT_{d-j} \big)^*$,
which is a map into a space of matrices.
For a matrix $F\in \Mat_m(\sT_d)$ we write $F = \sum_{a,b} e_a \otimes [F]_{a,b} \otimes e^*_b$ with standard basis vectors $e_a\in R^m$, $e_b\in R^m$.
Every linear map $\mathcal F:\sT_d\to\Mat_m(\sT_d)$
lifts to a linear map $\mathcal F:\Mat_\ell(\sT_d)\to\Mat_{\ell m}(\sT_d)$
via
$\mathcal F(F) = \sum_{a,b} e_a \otimes \mathcal F([F]_{a,b}) \otimes e^*_b$,
i.e., matrix entries are replaced by matrices, forming a block matrix.
Let $\rk$ denote the decomposition rank function,
i.e., for $M\in U\otimes U'$ we have
$\rk(M)=\min\{r\mid \exists l_i\in U,\ell_i\in U' : M=\sum_{i=1}^r l_i\otimes \otimes\ell_i\}$.

\begin{theorem}[\cite{Nis91}]
\label{thm:Nisanmatrix}
Let $F\in\Mat(\sT_d)$.
Then $\w(F)\geq \max\{\rk\mathcal F_j(F)\mid 0\leq j\leq d\}$.
If $R$ is a field, then we have equality.
\end{theorem}
\begin{proof}
For a sequence of variables $p=(p_1,\ldots,p_k)$ we write $p_{\otimes} := p_1 \otimes p_2 \otimes \cdots \otimes p_k$.
Let $V_j(G)$ denote the set of vertices of $G$ in vertex layer $j$.

For every vertex pair $v,w\in G$ we define $(\nsp{v,\rightsquigarrow,w}) := \sum_{p\in P(G,v,w)}p_\otimes$.
Note that
\[
[F]_{a,b} = \sum_{v \in V_j(G)} (\nsp{s_a,\rightsquigarrow,v})\otimes(\nsp{v,\rightsquigarrow,t_b}) \ .
\]
In other words,
\begin{equation}\label{eq:brakfmatrix}
F = \sum_{a,b} \sum_{v \in V_j(G)} e_a \otimes (\nsp{s_a,\rightsquigarrow,v})\otimes (\nsp{v,\rightsquigarrow,t_b}) \otimes e^*_b
= \sum_{v \in V_j(G)} \underbrace{\sum_{a,b} e_a \otimes (\nsp{s_a,\rightsquigarrow,v})\otimes (\nsp{v,\rightsquigarrow,t_b}) \otimes e^*_b}_{=: (\rightsquigarrow v \rightsquigarrow)} \ .
\end{equation}

We apply $\mathcal F_j$ and analyze the rank:

\begin{eqnarray*}
\mathcal F_j((\nsp{\rightsquigarrow,v,\rightsquigarrow})) &=& \sum_{a,b} e_a \otimes \mathcal F_j((\nsp{s_a,\rightsquigarrow,v})\otimes (\nsp{v,\rightsquigarrow,t_b})) \otimes e^*_b
= \sum_{a,b} (e_a \otimes (\nsp{s_a,\rightsquigarrow,v}) \otimes (\nsp{v,\rightsquigarrow,t_b})^* \otimes e^*_b)
\\
&=& \big(\sum_{a} e_a \otimes (\nsp{s_a,\rightsquigarrow,v}) \big) \otimes \big(\sum_{b} (\nsp{v,\rightsquigarrow,t_b})^* \otimes e^*_b\big),
\end{eqnarray*}
which is a rank $\leq 1$ matrix.
By \eqref{eq:brakfmatrix} and the definition of rank as the decomposition rank
$\rk(\mathcal F_j(F))\leq \#V_j(G)$,
which implies the first part of the theorem.

For the second part of the theorem, let $R=\IF$ be a field.
For the sake of contradiction, assume that $\w(F)>\max\{\rk\mathcal F_j(F)\mid 0 \leq j \leq d\}$.
Among all ABPs of width $\w(F)$ that compute $F$, take one that has the smallest index $j$ for which $\#V_j(G) = \w(F)$.
In particular $\#V_j(G) > \rk\mathcal F_j(F)$.
We will now see that there exists an ABP $G'$ of width $\leq\w(F)$ that computes $F$ with $\#V_j(G')<\#V_j(G)$,
which is a contradiction to the minimal choice of $j$.

We have
\[
F = \sum_{v \in V_j(G)} (\nsp{\rightsquigarrow,v,\rightsquigarrow}),
\]
but for one vertex $\tilde v\in V_j(G)$ we have a linear combination
\[
\mathcal F_j((\nsp{\rightsquigarrow,\tilde v,\rightsquigarrow})) = \sum_{v \in V_j(G)\setminus\{\tilde v\}} \alpha_v \mathcal F_j((\nsp{\rightsquigarrow,v,\rightsquigarrow})),
\]
with $\alpha_v \in \IF$.
We use that $\mathcal F_j$ is an isomorphism to conclude
\begin{equation}
\label{eq:expressvtildematrix}
(\nsp{\rightsquigarrow,\tilde v,\rightsquigarrow}) = \sum_{v \in V_j(G)\setminus\{\tilde v\}} \alpha_v \, (\nsp{\rightsquigarrow,v,\rightsquigarrow}).
\end{equation}
We construct $G'$ as follows: We delete the vertex $\tilde v$ from $G$, and we rescale all incoming
(or all outgoing if $j=0$)
edges of all $v\in V_j(G)$ with $(1+\alpha_v)$.
Note that $G'$ computes the matrix
\[
\sum_{v \in V_j(G)\setminus\{\tilde v\}} (1+\alpha_v) (\nsp{\rightsquigarrow,v,\rightsquigarrow}) \stackrel{\eqref{eq:brakfmatrix}+\eqref{eq:expressvtildematrix}}{=} F.
\]
Since the number of vertices in $G$ and $G'$ differ only in vertex layer $j$, this gives the desired contradiction to the choice of~$j$.
\end{proof}

For $\pi\in\aS_d$, $k\in\{0,\ldots,d\}$, let
\begin{eqnarray}
\nonumber\beta(\pi,k) &:=& \{0,\pi(1)-1,\pi(1),\pi(2)-1,\pi(2),\ldots,\pi(k)-1,\pi(k)\}
\\
\label{eq:betapi}
&& \cap \, \{\pi(k+1)-1,\pi(k+1),\pi(k+2)-1,\pi(k+2),\ldots,\pi(d)-1,\pi(d),d\}.
\end{eqnarray}
Let $\beta(\pi):=\max\{\#\beta(\pi,k)\mid 0\leq k\leq d\}$.

\begin{proposition}
\label{pro:Nisanoptimal}
Let $R$ be a commutative semiring, and let $A=\sT$.
We have $\w(\Xi_{n,d,\pi}) = n^{\beta(\pi)+1}$.
\end{proposition}
\begin{proof}
To prove the lower bound, we use Theorem~\ref{thm:Nisanmatrix}.
We have
\[
\hXi_{n,d,\pi} =
\sum_{\substack{i_0,i_1,\ldots,i_d\\j_0,j_1,\ldots,j_d}}
e_{i_0,j_0}
\otimes (x^{(\pi(1))}_{i_{\pi(1)-1},i_{\pi(1)}} \boxtimes x^{(1)}_{j_0,j_1})
\otimes
\cdots
\otimes (x^{(\pi(d))}_{i_{\pi(d)-1},i_{\pi(d)}} \boxtimes x^{(d)}_{j_{d-1},j_d})
\otimes e_{i_d,j_d}
\]
The monomials on the left factor of $\mathcal F_k(\hXi_{n,d,\pi})$ are
\[
e_{i_0,j_0}
\otimes (x^{(\pi(1))}_{i_{\pi(1)-1},i_{\pi(1)}} \boxtimes x^{(1)}_{j_0,j_1})
\otimes
\cdots
\otimes (x^{(\pi(k))}_{i_{\pi(k)-1},i_{\pi(k)}} \boxtimes x^{(k)}_{j_{k-1},j_k})
\]
The monomials on the right factor of $\mathcal F_k(\hXi_{n,d,\pi})$ are
\[
(x^{(\pi(k+1))}_{i_{\pi(k+1)-1},i_{\pi(k+1)}} \boxtimes x^{(k+1)}_{j_k,j_{k+1}})
\otimes
\cdots
\otimes (x^{(\pi(d))}_{i_{\pi(d)-1},i_{\pi(d)}} \boxtimes x^{(d)}_{j_{d-1},j_d})
\otimes e_{i_d,j_d}
\]
The set of $o$ for which $i_o$ appears on the left-hand side of the tensor is
$\{0,\pi(1)-1,\pi(1),\pi(2)-1,\pi(2),\ldots,\pi(k)-1,\pi(k)\}$.
The set of $o$ for which $i_o$ appears on the right-hand side of the tensor is
$\{\pi(k+1)-1,\pi(k+1),\pi(k+2)-1,\pi(k+2),\ldots,\pi(d)-1,\pi(d),d\}$.
We say that $o\in\{0,\ldots,d\}$ is a \emph{boundary point} if $i_o$ appears on both sides of the tensor,
otherwise $o$ is called an \emph{inner point}.
We denote by $\beta(\pi,k)$ the set of boundary points.
One monomial from the left-hand side tensored with a monomial from the right-hand side
gives a monomial of $\hXi_{n,d,\pi}$ if and only if they agree on $j_k$ and on every $i_o$ for boundary points $o$.
Now we consider the set $M$ of those monomials for which $i_o=1$ for inner points $o$, and $j_o=1$ for all $o\neq k$.
For each monomial in $M$ from the left-hand side there is exactly one monomial in $M$ from the right-hand side
so that their tensor product gives a monomial in $\hXi_{n,d,\pi}$.
This means that if we restrict the matrix $\mathcal F_k(\hXi_{n,d,\pi})$ to the rows and columns of those monomials, we get a permutation matrix.
The rank of this matrix is full, see Lemma~\ref{lem:rankofid}.
The number of such monomials from the left-hand side equals the number of monomials from the right-hand side,
and it is $n^{\beta(n,k)+1}$,
which proves the lower bound.
We now prove the upper bound by adapting a construction of Grenet~\cite{Gre11}.
Let $B_k$ denote the set of boundary points and let $I_k$ be the set of inner points of the left-hand side.
For every map $C_k\to\{1,\ldots,n\}$ and a number $c_k\in\{1,\ldots,n\}$ we define the vector
\[
F_{k,C_k,c_k} :=
\sum_{\substack{\forall o\in I_k: i_o \in \{1,\ldots,n\}\\j_0,j_1,\ldots,j_{k-1}\in\{1,\ldots,n\}\\\forall o\in B_k: i_o=C(o)\\j_k=c_k}}
e_{i_0,j_0}
\otimes (x^{(\pi(1))}_{i_{\pi(1)-1},i_{\pi(1)}} \boxtimes x^{(1)}_{j_0,j_1})
\otimes
\cdots
\otimes (x^{(\pi(k))}_{i_{\pi(k)-1},i_{\pi(k)}} \boxtimes x^{(k)}_{j_{k-1},j_k})
\]
which is the sum of the subset of monomials of the left-hand side of $\mathcal F_k(\hXi_{n,d,\pi})$ that satisfy
$i_o=C_k(o)$ for all boundary points $o$,
and $j_k=c_k$.
The ABP we construct has one vertex for each of these $n^{\beta(n,k)+1}$ many tensors in degree $k$
at which the ABP computes the vector $F_{k,C_k,c_k}$.
Edges with appropriate labels can be chosen to make this ABP compute $\hXi_{n,d,\pi}$, because when running $k$ from $0$ to $d$,
every position $o$ transitions at most once (exactly once for all $o\neq d$) from boundary to inner, and never in the other direction. A transition from boundary to inner is handled by summation,
for example, if $\pi(k+1)=o$ and $o-1$ transitions from boundary to inner and $o$ is a new boundary point, then for any $C_{k+1}$ and $c_{k+1}$ we have
\[
F_{k+1,C_{k+1},c_{k+1}} = \sum_{i,j=1}^n F_{k,C_{k+1}\setminus\{o\}\cup \{o-1\mapsto i\},j} \otimes \left( x^{\pi(k+1)}_{i,C_{k+1}(o)} \boxtimes x^{(k+1)}_{j,c_{k+1}}\right),
\]
where $C_{k+1}\setminus\{o\}\cup \{o-1\mapsto i\}$ arises from the map $C_{k+1}$ by making it undefined at $o$ and setting $C_{k+1}(o-1)=i$.
\end{proof}

Let $\beta(d) := \max\{\beta(\pi)\mid \pi\in\aS_d\}$.
By \eqref{eq:xidasmax} we have that for $A_d=\sT_d$ it holds $\xi_d=\beta(d)+1$.
Note that the upper bound $\w(\hXi_{n,d,\pi})\leq n^{\beta(\pi)+1}$ also holds for $A_d\in\{\sS_d,\sW_d\}$, because these are quotients of $\sT_d$.

Let $\beta'(\pi,k)=\beta(\pi,k)\setminus\{0,d\}$
and $\beta'(\pi)=\max\{\#\beta(\pi,k)\mid 0\leq k \leq d\}$.
The same arguments as for $\hXi_{n,d,\pi}$ show that for $A_d=\sT_d$ we have
$\w(\hGamma\!_{n,d,\pi})=n^{\beta'(\pi)}$
and for $A_d\in\{\sS_d,\sW_d\}$ we have
$\w(\hGamma\!_{n,d,\pi})\leq n^{\beta'(\pi)}$.

and $\beta'(d) := \max\{\beta'(\pi)\mid \pi\in\aS_d\}$.
\begin{lemma}
We have $\beta(d) = \begin{cases}
d & \textup{ if $d$ is odd}, \\
d+1 & \textup{ if $d$ is even}.
\end{cases}$
\\And we have $\beta'(d) = d-1$.
\end{lemma}
\begin{proof}
We consider $\beta(d)$ first.
For any $\pi$ with $\pi(1)=2, \pi(2)=4, \pi(3)=6, \ldots, \pi_{\lfloor d/2\rfloor}=2\lfloor d/2\rfloor$
we see that $\beta(\pi) \geq 1+2\lfloor d/2\rfloor$, which is $d+1$ for even $d$, and $d$ for odd $d$.
Clearly $d+1$ is an upper bound for $\beta(d)$.
We phrase $\beta(d)$ as an edge coloring process of the graph on $d+1$ vertices $v_i$,
where $v_i$ is adjacent to $v_{i-1}$ and $v_{v+1}$, and $v_0$ has an additional half-edge,
and $v_d$ has an additional half-edge. Initially, all edges are colored blue,
the half-edge at $v_0$ is colored red and the half-edge at $v_d$ is colored blue.
We now color edges (not half-edges) from blue to red step by step.
$\beta(\pi,k)$ is the number of vertices that are adjacent to a blue and a red (half-)edge after $k$ recolorings.
To choose $\pi$ so that $\beta(\pi)$ is maximized, it is not beneficial to ever choose an edge that has a red neighbor edge,
because the number of boundary vertices between red and blue does not increase,
and in any future step, coloring this edge back blue can only increase the number of boundary points.
The same idea also works for coloring the edge $(v_0,v_1)$.
Hence, it is optimal to color the $(v_1,v_2)$, $(v_3,v_4)$, and so on.
Since $d$ is odd, we can only do this $\frac{d-1}{2}$ many times, after which the number of boundary points will go down again.
Starting with $1$ boundary point $v_0$,
each such recoloring increases the number of boundary points by $2$, hence $\beta(d)=1+2\frac{d-1}{2}=d$.

For $\beta'$, when counting boundary points, we do not count $v_0$ and $v_d$, and we can ignore the half-edges.
For $d$ odd, coloring $(v_1,v_2)$, $(v_3,v_4)$, and so on, gives
$\beta'(d)\geq 2\frac{d-1}{2} = d-1$.
For $d$ even, coloring $(v_1,v_2)$, $(v_3,v_4)$, and so on, gives
$\beta'(d)\geq 2(\frac{d}{2}-1)+1 = d-1$, since $v_d$ is not being counted.
The upper bound comes from a similar coloring consideration.
Coloring $(v_0,v_1)$ or $(v_{d-1},v_d)$ can increase the number of boundary points for subsequent rounds by 1,
but it comes at the cost of only increasing the number of boundary points only by 1 instead of 2.
Hence, effectively, the number of boundary points can only be increased by 2 each round.
We make a case distionction based on the color of $(v_0,v_1)$ and $(v_{d-1},v_d)$ in the coloring with maximized number of boundary points that count.
If $(v_0,v_1)$ is red and $(v_{d-1},v_d)$ is red, then we get at most $2+2\lfloor\frac{d-3}{2}\rfloor=\leq d-1$ many boundary points that count.
If $(v_0,v_1)$ is red and $(v_{d-1},v_d)$ is blue, then we get at most $1+2\lfloor\frac{d-2}{2}\rfloor=\leq d-1$ many boundary points that count.
Analogously for $(v_0,v_1)$ blue and $(v_{d-1},v_d)$ red.
If $(v_0,v_1)$ is blue and $(v_{d-1},v_d)$ is blue, then we get at most $2\lceil\frac{d-2}{2}\rceil=\leq d-1$ many boundary points that count.
Hence, in all cases we see $\beta'(d)\leq d-1$.
\end{proof}

\subsection{Optimal speedups for polynomials if $R$ is zerosumfree}\label{subsec:monotone}
A commutative semiring is called \emph{zerosumfree} if
for all $r_1,r_2 \in R$ we have
$r_1+r_2=0$ implies $r_1=r_2=0$.
For $\pi\in\aS_d$, $\sigma\in\aS_d$, $k\in\{0,\ldots,d\}$, we define
\begin{eqnarray}
\nonumber\beta(\pi,\sigma,k) &:=&
\{1\}\times\Big(
\{0,\pi(\sigma(1))-1,\pi(\sigma(1)),\pi(\sigma(2))-1,\pi(\sigma(2)),\ldots,\pi(\sigma(k))-1,\pi(\sigma(k))\} 
\\
\nonumber
&& \hspace{1.5cm}\cap \, \{\pi(\sigma(k+1))-1,\pi(\sigma(k+1)),\pi(\sigma(k+2))-1,\pi(\sigma(k+2)),\ldots,\pi(\sigma(d))-1,\pi(\sigma(d)),d\} 
\Big)
\\
\nonumber
&\cup&
\{2\}\times\Big(
\{0,\sigma(1)-1,\sigma(1),\sigma(2)-1,\sigma(2),\ldots,\sigma(k)-1,\sigma(k)\} 
\\
\label{eq:betapisigmak}
&& \hspace{1.5cm}\cap \, \{\sigma(k+1)-1,\sigma(k+1),\sigma(k+2)-1,\sigma(k+2),\ldots,\sigma(d)-1,\sigma(d),d\} 
\Big)
\end{eqnarray}
As a remark, if one compares to \eqref{eq:betapi}, one observes that $\beta(\pi,\id,k) = \beta(\pi,k)+1$.
Now, let $\beta(\pi,\sigma):=\max\{\#\beta(\pi,\sigma,k)\mid 0\leq k\leq d\}$.
Let $\beta_\textup{opt}(\pi):=\min\{\#\beta(\pi,\sigma)\mid \sigma\in\aS_d\}$.
Let $\beta_\textup{opt}(d) := \max\{\beta_\textup{opt}(\pi)\mid \pi\in\aS_d\}$.

\begin{proposition}\label{pro:monotoneupperbound}
For a commutative semiring $R$ and $A=\sS$ we have $\w(\hXi_{n,d\,pi})\leq n^{\beta_{\textup{opt}}(\pi)}$.
\end{proposition}
\begin{proof}
The upper bound uses essentially the same construction as in the proof of Proposition~\ref{pro:Nisanoptimal}.
The only difference is that we are free to choose $\sigma\in\aS_d$ instead of having $\sigma=\id$,
and after fixing $\sigma$ we can write
\begin{equation}\label{eq:hXindpifixedsigma}
\hXi_{n,d,\pi} =
\sum_{\substack{i_0,i_1,\ldots,i_d\\j_0,j_1,\ldots,j_d}}
e_{i_0,j_0}
\otimes (x^{(\pi(\sigma(1)))}_{i_{\pi(\sigma(1))-1},i_{\pi(\sigma(1))}} \boxtimes x^{(\sigma(1))}_{j_{\sigma(1)-1},j_{\sigma(1)}})
\cdot \cdots \cdot
(x^{(\pi(\sigma(d)))}_{i_{\pi(\sigma(d))-1},i_{\pi(\sigma(d))}} \boxtimes x^{(\sigma(d))}_{j_{\sigma(d)-1},j_{\sigma(d)}})
\otimes e_{i_d,j_d}.
\end{equation}
Now $\beta(\pi,\sigma,k)$ is the set of all boundary points,
so we have $n^{\#\beta(\pi,\sigma,k)}$ many vertices in the ABP in vertex layer~$k$.
\end{proof}
\begin{remark}
Proposition~\ref{pro:monotoneupperbound} also works for $A=\sW$ if $R$ is a commutative ring.
If $R$ is a commutative semiring, then we can only allow even permutations $\sigma$
in the definition of $\beta_{\textup{opt}}(\pi)$,
which increases its value by at most 4, because it is easy to see that
$\forall \pi,\sigma: \beta(\pi,(1 \ 2)\sigma) \leq 4+\beta(\pi,\sigma)$.
\end{remark}
Note that better constructions are possible for $A=\sS$
by not treating every $\pi$ individually,
but by \eqref{eq:xidasmax} this cannot give a better exponent.

\begin{proposition}\label{pro:monotonelowerbound}
For a zerosumfree commutative semiring $R$ and $A=\sS$ we have $n^{\beta_{\textup{opt}}(\pi)}/d \leq \w(\hXi_{n,d,\pi})$.
\end{proposition}
\begin{proof}
The proof technique is standard, see \cite{KPR23} and references therein.
Let $G$ be an ABP that computes $\hXi$. Edges with label zero are removed.
We call the pair $(k,k')$ the \emph{edge layer pair} of the variable $x_{i,j}^{(k)}\boxtimes x_{i',j'}^{(k')}$.
Let $\sigma_p := \big(\sigma(1),\ldots,\sigma(d)\big)$ be the list of the second indices of the edge layer pairs of variables on $p$ read from source to sink.
If $p$ and $q$ have a common vertex $v$ in some vertex layer $k$, then
\begin{equation}\label{eq:sigmapsigmaq}
\{\sigma_p(1),\ldots,\sigma_p(k)\}=\{\sigma_q(1),\ldots,\sigma_q(k)\}
 \ \textup{ and } \ 
\{\sigma_p(k+1),\ldots,\sigma_p(d)\}=\{\sigma_q(k+1),\ldots,\sigma_q(d)\},
\end{equation}
because otherwise the path that starts at a source, goes via $p$ until $v$,
and continues from $v$ via $q$ to a sink would contribute to the coefficient of a monomial that is not a monomial of $\hXi$.
We define analogously to \eqref{eq:betapisigmak}
\begin{eqnarray*}
\beta(v) &:=&
\{1\}\times\Big(
\{0,\pi(\sigma(1))-1,\pi(\sigma(1)),\pi(\sigma(2))-1,\pi(\sigma(2)),\ldots,\pi(\sigma(k))-1,\pi(\sigma(k))\} 
\\
&& \hspace{1.5cm}\cap \, \{\pi(\sigma(k+1))-1,\pi(\sigma(k+1)),\pi(\sigma(k+2))-1,\pi(\sigma(k+2)),\ldots,\pi(\sigma(d))-1,\pi(\sigma(d)),d\} 
\Big)
\\
&\cup&
\{2\}\times\Big(
\{0,\sigma(1)-1,\sigma(1),\sigma(2)-1,\sigma(2),\ldots,\sigma(k)-1,\sigma(k)\} 
\\
&& \hspace{1.5cm}\cap \, \{\sigma(k+1)-1,\sigma(k+1),\sigma(k+2)-1,\sigma(k+2),\ldots,\sigma(d)-1,\sigma(d),d\} 
\Big),
\end{eqnarray*}
where $\sigma$ can be chosen arbitrarily from $\{\sigma_p \mid v \in p\}$, and $\beta(v)$ is independent of this choice because of \eqref{eq:sigmapsigmaq}.
We call $\beta(v)$ the set of \emph{boundary points} of~$v$.
To every vertex $v$ we assign its \emph{bag},
which is a function $\textup{bag}(v):\beta(v)\to\{1,\ldots,n\}$ as follows.
We take a direct source-sink path $p$ that uses $v$ and
a variable $x_{i,j}^{(k)}\boxtimes x_{i',j'}^{(k')}$ on $p$
such that
$(1,k-1) \in \beta(v)$,
and we set $\textup{bag}(v)((1,k-1))=i$.
Analogously, for $(1,k) \in \beta(v)$ we set $\textup{bag}(v)((1,k))=j$,
and
for $(2,k'-1) \in \beta(v)$ we set $\textup{bag}(v)((2,k'-1))=i'$,
and
for $(2,k') \in \beta(v)$ we set $\textup{bag}(v)((2,k'))=j'$.
This setting is consistent, because of the structure of the monomials of the entries of $\hXi$.
Moreover, this completely defines $\textup{bag}(v)$.
The definition of $\textup{bag}(v)$ does not depend on the choice of~$p$,
because for $p$ and $q$ through $v$ with different bags for $v$ (but still, \eqref{eq:sigmapsigmaq} holds)
we see that there is also the path that starts at a source and goes via $p$ until it reaches $v$ and continues via $q$ to a sink,
and that path would not correspond to a monomial in $\hXi_{n,d,\pi}$.
The bag size of $v$ is defined as $\#\beta(v)$.
If $v$ is the vertex in vertex layer $k$ on a direct source-sink path $p$,
then $\#\beta(v)=\#\beta(\pi,\sigma_p,k)$.
By definition of $\beta(\pi,\sigma_p)$ we have that
at least one of the vertices has bag size $\beta(\pi,\sigma_p)$.
Note that $\beta_\textup{opt}(\pi)\leq \beta(\pi,\sigma_p)$,
therefore for every direct source-sink path $p$ 
at least one of the vertices has bag size at least $\beta_\textup{opt}(\pi)$.

We assign a \emph{position} in $\{1,2\}\times\{0,\ldots,d\}$ to every subscript of a variable
$x_{i,j}^{(k)}\boxtimes x_{i',j'}^{(k')}$ by saying that the position of $i$ is $(1,k-1)$,
the position of $j$ is $(1,k)$, the position of $i'$ is $(2,k'-1)$,
and the position of $j'$ is $(2,k')$.
There are $2d+2$ many positions, and the choice of a value from $\{1,\ldots,n\}$ for each position corresponds bijectively to a monomial in $\hXi_{n,d,\pi}$.
By definition, in every monomial of $\hXi_{n,d,\pi}$, all indices that have the same position have the same value.
The bag of a vertex $v$ fixes $\#\beta(v)$ many values of positions,
hence there can be at most $n^{2d+2-\#\beta(v)}$ many monomials
whose computation gets a contribution from direct source-sink paths that use $v$.

Every monomial has at least one vertex with $\#\beta(v)\geq\beta_{\textup{opt}}$,
therefore this vertex is only involved in the computation of at most $n^{2d+2-\beta_{\textup{opt}}}$ many monomials.
Since we have $n^{2d+2}$ many monomials, we have at least $n^{2d+2}/n^{2d+2-\beta_{\textup{opt}}}=n^{\beta_{\textup{opt}}}$ many vertices, thus $G$ has width at least $n^{\beta_{\textup{opt}}}/d$.
\end{proof}

For zerosumfree $R$ and $A=\sS$,
combining Proposition~\ref{pro:monotoneupperbound} and Proposition~\ref{pro:monotonelowerbound}
gives
$\beta_\textup{opt}(d) = \xi_d$.
In analogy to $\beta_\textup{opt}(d)$ we define
$\beta'(\pi,\sigma,k) := \beta(\pi,\sigma,k)\setminus\{(1,0),(1,d),(2,0),(2,d)\}$,
and $\beta'(\pi,\sigma):=\max\{\#\beta'(\pi,\sigma,k)\mid 0\leq k\leq d\}$,
and $\beta'_\textup{opt}(\pi):=\min\{\#\beta'(\pi,\sigma)\mid \sigma\in\aS_d\}$,
and
$\beta'_\textup{opt}(d) := \max\{\beta'_\textup{opt}(\pi)\mid \pi\in\aS_d\}$.
By construction, $\beta(\pi,\sigma,k) \leq 4+\beta'(\pi,\sigma,k) \leq 4+\beta(\pi,\sigma,k)$.
We see $\beta'_\textup{opt}(d) = \gamma_d$ in the same way as
$\beta_\textup{opt}(d) = \xi_d$.
By brute force enumeration of $(\pi,\sigma,k)$ we obtain the following values.
\begin{center}
\begin{tabular}{r|ccccccc}
$d$&2&3&4&5&6&7
\\
\hline
$\xi_d$ for $A=\sT$ & 
4&4&6&6&8&8
\\
$\xi_d$ for zerosumfree $R$ and $A=\sS$ &
4&4&6&6&6&8
\\
\hline
$\gamma_d$ for $A=\sT$ & 
2&3&4&5&6&7
\\
$\gamma_d$ for zerosumfree $R$ and $A=\sS$ &
2&3&4&4&5&6
\end{tabular}
\end{center}
Note the difference between $\xi_6$ for $A=\sT$
and $\xi_6$ for $R$ zerosumfree and $A=\sS$.
The next proposition shows that such differences cannot be large.
\begin{proposition}
For zerosumfree commutative semirings $R$ and $A=\sS$ we have
\[
\beta_\textup{opt}(d) = \Omega(d).\qedhere
\]
\end{proposition}
\begin{proof}
This follows from \cite[Thm 3.4]{AR08} and the fact that $\beta'_{\textup{opt}}(\pi)$
is the same as the pathwidth of a permutation graph for $\pi$, which is called $C(\pi)$ in \cite[Def 2.6]{AR08}.
\end{proof}

\subsection{Monotonicity of the exponents}
\label{subsec:affine}
In this subsection we prove that the sequences $(\xi_d)_{d\in\IN}$
and $(\gamma_d)_{d\in\IN}$ grow monotonously,
where for $A=\sW$ we require that $R$ is a commutative ring.

An affine algebraic branching program (aABP) is
defined in the same way as an ABP with the only difference
that we allow affine $R$-linear combinations of variables as edge labels.
We also do not require the all source-sink paths to have the same length.
The size of an aABP is the number of its vertices, not counting sources or sinks.
We write $\asize(F)$ for the smallest size of an aABP computing $F$.
For $F\in \Mat(A_d)$, we clearly have
\begin{equation}
\label{eq:asizeleqw}
\asize(F)\leq (d-1) \cdot \w(F).
\end{equation}
The following proposition is classical and shows the other direction.
\begin{proposition}[Homogenization]
\label{pro:homogenization}
Let $H$ be a matrix with entries in $\sT$
or with entries in $\sS$ or with entries in $\sW$.
and let $F$ be a homogeneous component of $H$.
Then $\w(F) \leq \asize(H)$.
In particular, $\w(F) \leq \asize(F)$.
\end{proposition}
\begin{proof}
A vertex in an ABP computes a matrix, by considering only those source-sink paths that use that vertex.
Let $d:=\deg(F)$.
We take a minimal size aABP $G$ that computes $H$
and we replace every vertex $v$ with $d+1$ many vertices $v_0,\ldots,v_d$.
We will set the edges so that if a vertex $v$ computes some matrix $M$, then $v_i$ will compute the $i$-th homogeneous component of $M$, for $0\leq i \leq d$,
and we will not compute higher degree components.
If $\ell + \alpha$ is an edge label in $G$ from $v$ to $w$
with $\ell$ homogeneous linear and $\alpha$ constant,
then we add edges with label $\ell$ from each $v_i$ to $w_{i+1}$,
and we add edges with label $\alpha$ from each $v_i$ to $w_i$.
For the sources $s_a$ we delete all $s_{a,i}$ with $i>0$,
and for the sinks $t_b$ we delete all $t_{b,i}$ with $i<d$.
This finishes the construction.
The result is an ABP that has at most $\asize(H)$ many vertices
in each vertex layer but the first and last, in other words, the width of the ABP is at most $\asize(H)$.
Hence $\w(H)\leq\asize(H)$.
\end{proof}

\begin{corollary}\label{cor:xigammamonotone}
For $A\in\{\sT,\sS\}$ the sequences $(\xi_d)_{d\in\IN}$
and $(\gamma_d)_{d\in\IN}$ grow monotonously.
\end{corollary}
\begin{proof}
Let $\pi\in\aS_d$ with $\pi(d)=d$, and let $\kappa\in\aS_{d-1}$ be defined via $\forall 1\leq i \leq d-1: \kappa(i)=\pi(i)$.
Take an ABP of minimal size that computes $\hXi_{n,d,\pi}$.
Set all variables $x_{i,j}^{(k)}\boxtimes x_{i',j'}^{(\pi(k))}$ to 1 for which $k=d$ (hence also $\pi(k)=d$).
The resulting aABP computes $\hXi_{n,d-1,\pi}$.
Hence $\asize(\hXi_{n,d-1,\kappa}) \leq \asize(\hXi_{n,d,\pi})$.
By Proposition~\ref{eq:asizeleqw} and Proposition~\ref{pro:homogenization}
this implies
$\w(\hXi_{n,d-1,\kappa}) \leq d\cdot\w(\hXi_{n,d,\pi})$.
Hence, $\forall\kappa\in\aS_{d-1}:\inf\{\log_n(\w(\hXi_{n,d-1,\kappa}))\} \leq \xi_d$,
and therefore $\xi_{d-1}\leq\xi_d$ by \eqref{eq:xidasmax}.
The proof of $\gamma_{d-1}\leq\gamma_d$ is analogous.
\end{proof}

\begin{proposition}
Let $R$ be a commutative ring, and $A=\sW$.
Then the sequences $(\xi_d)_{d\in\IN}$
and $(\gamma_d)_{d\in\IN}$ grow monotonously.
\end{proposition}
\begin{proof}
Let $\pi\in\aS_d$ with $\pi(d)=d$, and let $\kappa\in\aS_{d-1}$ be defined via $\forall 1\leq i \leq d-1: \kappa(i)=\pi(i)$.
Crucial property over commutative semirings even over wedges:
If in $\hXi_{n,d,\pi}$
we map each $x_{i,j}^{(d)}\boxtimes x_{i',j'}^{(d)}$ to $y$ if and only if $i=i'$ and $j=j'$,
otherwise we map it to zero,
we obtain $(\hXi_{n,d-1,\kappa})\wedge y$.
If in $\hXi_{n,d}$ we additionally map all $x_{i,j}^{(d)}\boxtimes x_{i',j'}^{(d')}$ to zero for which $d\neq d'$,
we obtain $\hXi_{n,d-1}\wedge y$.
Given an aABP that computes $\hXi_{n,d-1}\wedge y$,
we split every vertex into two vertices, one for $y$-degree 0 and one for $y$-degree 1 (higher $y$-degrees are deleted).
We now replace every $y$ that appears in an edge layer of the same parity as $d$ with $1$,
and we replace $y$ with $-1$ otherwise.
The resulting aABP computes $\hXi_{n,d-1}$,
hence $\asize(\hXi_{n,d-1})\leq \asize(2 \hXi_{n,d})$,
and thus $\xi_{d-1}\leq\xi_d$, see \eqref{eq:asizeleqw} and Proposition~\ref{pro:homogenization}.
The proof of $\gamma_{d-1}\leq\gamma_d$ is analogous.
\end{proof}

\section{Coefficients of polynomials given by Boolean circuits}\label{sec:boolean}

In this section we study elements of $A_d$ whose coefficients are computed by small Boolean circuits, defined below.
Let $R$ be a commutative semiring and $A\in\{\sT,\sS,\sW\}$.
We define the negation $\overline{0} := 1$ and $\overline{1} := 0$.
For a finite list $\mathbf b$ we define $\overline{\mathbf b}$ as the negation applied on all elements.
We also use the negation symbol on variables, for example $\overline{x_i}$, but this is a new variable instead of a negation.
We define
\[
(x,\overline{x})^{\mathbf{b},\overline{\mathbf{b}}} \ := \ x_1^{b_1} \overline{x}_1^{\overline{b}_1} x_2^{b_2} \overline{x}_2^{\overline{b}_2} \cdots x_{m}^{b_{m}} \overline{x}_{m}^{\overline{b}_m} \ \in A_{2m}.
\]

A Boolean circuit $B$ is a directed acyclic graph with a single outdegree 0 vertex,
whose indegree 0 vertices (called \emph{inputs}) are labeled with variables, 
and whose other vertices (called \emph{computation gates}) are each labeled with an operation $\AND,\OR,\NOT$,
where the vertices labeled with $\AND$ or $\OR$ have indegree 2,
and the vertices labeled with $\NOT$ have indegree 1.
We will always assume that $B$ is connected, i.e., there is at least one path from each input to the output,
which can be achieved by taking any gate $v$ that connects to the output and replace it with $v\AND(x\OR\NOT x)$ for an unused input $x$.
On input $\mathbf b\in\{0,1\}^m$, $B$ computes a Boolean value by induction over the circuit structure in the obvious way.
The size $|B|$ is defined as the number of the computation gates of~$B$.
An input $\textbf b$ with $B(\textbf b)=1$ is called a satisfying assignment of $B$.

In this section we prove the following theorem in which we use the dummy variables $\epsilon^{(k)}\in\Var$.
\begin{theorem}
\label{thm:booleancktsascoeffs}
Given a Boolean circuit $B$ on $m$ inputs. Then
\[
\sum_{\substack{B(\mathbf{b})=1}} (x,\overline{x})^{\mathbf b,\overline{\mathbf b}} \prod_{i=m+1}^d \epsilon^{(i)} \ \leq \ \hGamma\!_{n,d},
\]
with $n \leq (3|B|+m)^2$ and $d \leq 3|B|+m+1$.
\end{theorem}

The proof procedes in several steps, where the first steps are classical.
We first use the Tseytin transformation (Lemma~\ref{lem:tseytin}) to convert $B$ into a product of clauses.
These clauses can then be arithmetized in low width, which gives part of the second factor in the $\boxtimes$ construction.
The construction of the first factor $h$ is a variant of the elementary symmetric polynomial, see Lemma~\ref{lem:widthh}.

A \emph{literal} is a variable or its negation. A \emph{clause} is defined as the $\OR$ of 3 literals.

\begin{lemma}[Tseytin transformation \cite{Tse68,Tse83}]
\label{lem:tseytin}
For every Boolean circuit $B$ in $m$ variables $\mathbf b$ there exists a set $\mathcal C$ of clauses
in $r>m$ variables $(\mathbf{b},o,\mathbf{b'})$
with the following property:
for every $\mathbf b\in\{0,1\}^m$ there exists a unique $(o,\mathbf b') \in \{0,1\}^{r-m}$
that satisfies $\mathcal C(\mathbf b,o,\mathbf b')=1$.
Moreover, for this $(o,\mathbf b')$ we have $B(\mathbf b)=o$.
\end{lemma}
\begin{proof}
Given $B$ on $m$ variables $\mathbf b$, we first attach a new wire to the output gate, but this wire has \emph{no second endpoint}, and we call it the output wire.
We assign a variable $a_e$ to each wire $e$:
If the wire $e$ is connected to the $i$-th input, then $a_e = b_i$.
If the wire $e$ is the output wire, then $a_e = o$.
Otherwise, a new variable $b'_j$ is created and assigned to the wire.
This gives a variable vector $(\mathbf b,o,\mathbf b')$.
A Boolean assignment to these variables can be interpreted as a Boolean assignment to wires.
We now introduce a set of clauses that is true if and only if the Boolean values on the wires respect the gate operations.
For example,
for every $\NOT$ gate with input wire $e_1$ and output wire $e_2$ we observe that the logical expression $a_{e_2} = \NOT(a_{e_1})$
can be expressed as the $\AND$ of these clauses:
$\OR(\overline{a_{e_1}},\overline{a_{e_2}})$,
$\OR(a_{e_1},a_{e_2})$.
Moreover, for every $\AND$ gate with input wires $e_1,e_2$ and output wire $e_3$ we observe (for example via case distinction over all $2^3=8$ cases)
that the logical expression $a_{e_3} = \AND(a_{e_1},a_{e_2})$
can be expressed as the $\AND$ of these clauses:
$\OR(\overline{a_{e_1}},\overline{a_{e_2}},a_{e_3})$,
$\OR(a_{e_1},\overline{a_{e_3}})$,
$\OR(a_{e_2},\overline{a_{e_3}})$.
Analogously,
$a_{e_3} = \OR(a_{e_1},a_{e_2})$
can be expressed as the $\AND$ of these clauses:
$\OR(a_{e_1},a_{e_2},\overline{a_{e_3}})$,
$\OR(\overline{a_{e_1}},a_{e_3})$,
$\OR(\overline{a_{e_2}},a_{e_3})$.
Let $\mathcal C$ be the set of all these clauses.
By construction, for every $\mathbf b$ there is a unique $(o,\mathbf b')$ such that $\prod_{C \in \mathcal C} C(\mathbf{b},o,\mathbf{b'})=1$,
and this $(o,\mathbf b')$ is given by the Boolean values on the wires for input $\mathbf b$, defined recursively over the circuit structure.
Clearly, $B(\mathbf b)=o$.
\end{proof}

We now convert the set $\mathcal C$ of clauses into a width 7 ABP in $2r$ many variables $y_i$ and $\overline{y}_i$.
We arrange the clauses in any order, and each clause $(l_1 \OR l_2 \OR l_3)$ is converted into an ABP, and all those programs are contatenated by identifying the sink of the left ABP with the source of the right ABP, which results in an ABP that computes their product.
Let $l_1 \in\{b_{i_1},\overline{b}_{i_1}\}$, $l_2 \in\{b_{i_2},\overline{b}_{i_2}\}$, $l_3 \in\{b_{i_3},\overline{b}_{i_3}\}$,
for some $1\leq i_1,i_2,i_3\leq r$.
We now convert the clause $(l_1 \OR l_2 \OR l_3)$.
For $j\in\{1,2,3\}$ let $\ell_j = y_{i_j}$ if $l_j = b_{i_j}$, and let $\ell_j = \overline{y}_{i_j}$ otherwise.
Analogously, let $\overline{\ell}_j = \overline{y}_{i_j}$ if $l_j = b_{i_j}$, and let $\overline{\ell}_j = y_{i_j}$ otherwise.
The clause $(l_1 \OR l_2 \OR l_3)$ is converted into the ABP
\[
\ell_1 \cdot \ell_2 \cdot \ell_3
+
\ell_1 \cdot \overline{\ell_2} \cdot \ell_3
+
\ell_1 \cdot \ell_2 \cdot \overline{\ell_3}
+
\ell_1 \cdot \overline{\ell_2} \cdot \overline{\ell_3}
+
\overline{\ell_1} \cdot \ell_2 \cdot \ell_3
+
\overline{\ell_1} \cdot \overline{\ell_2} \cdot \ell_3
+
\overline{\ell_1} \cdot \ell_2 \cdot \overline{\ell_3}.
\]
The key property is that for every satisfying assignment to the clause there is exactly one $s$-$t$-path in this program, and vice versa.
Clauses with 2 variables are handled analogously with 3 summands instead of 7, and degree 2 instead of 3.
Concatenating these ABPs we obtain an ABP $G'$.
We call an $s$-$t$-path in $G$ \emph{consistent} if it does not use both $y_i$ and $\overline{y}_i$ for any $i$.
We see that there is are bijections between the following sets:
\begin{equation}\label{eq:bijectionsbcp}
\begin{minipage}{14cm}
\begin{itemize}
\item The set of vectors $\mathbf b\in\{0,1\}^m$
\item The set of satisfying assignments $(\mathbf b,o,\mathbf b')$ of the clause conjunction $\mathcal C$
\item The set of consistent $s$-$t$-paths in $G'$
\end{itemize}
\end{minipage}
\end{equation}

Since a satisfying assignment $(\mathbf b,o,\mathbf b')$ to $\mathcal C$ depends on $\mathbf b$, we write $p_{\mathbf b}$ instead of $p_{\mathbf b,o,\mathbf b'}$.
And in the other direction, we denote by $(\mathbf b,o,\mathbf b')_p$ the satisfying assignment of $\mathcal C$ to a consistent path $p$,
and we write $\mathbf b_p$ for the first $m$ bits of $(\mathbf b,o,\mathbf b')_p$.

We modify the ABP for $G'$ by replacing every edge label $y_i$ or $\overline{y}_i$ in layer $k$ by $y_i^{(m+1+k)}$ and $\overline{y}_i^{(m+1+k)}$, respectively.
We call the resulting ABP $G$.

Since every computation gate in $B$ resulted in at most 3 clauses, the degree of $G$ is at most $3|B|$.
We let $d := \deg(q)+m+1$, where $G$ computes $q$,
and observe $d \leq 3|B|+m+1$.
We let $r$ be the length of $(\mathbf b,o,\mathbf b')$, and we observe $r \leq m+|B|$.

Let $Z = \{z_{i,(k)},\overline{z}_{i,(k)} \mid 1 \leq k \leq d, \, 1 \leq i \leq r\}$.
Note that $|Z| = n := 2dr \leq 2(3|B|+m+1)(m+|B|) \leq (3|B|+m)^2$.
A subset $S\subseteq Z$ is called consistent if it does not contain any two variables $z_{i,(k)}$
and $\overline{z}_{i,(k')}$.
Let
\[
h_{\textup{poly}} \ := \ \sum_{\substack{S \subseteq Z\\|S|=d\\S \textup{ consistent}}} \prod_{z\in S} z \ \in \ \sS_d.
\]
\begin{lemma}\label{lem:widthh}
$\w(h_{\textup{poly}})\leq 2dr$.
\end{lemma}
\begin{proof}

We sort $Z$ as follows:
\begin{align*}
\mathbf z = (&z_{1,(1)},\ldots,z_{1,(d)},
\overline{z}_{1,(1)},\ldots,\overline{z}_{1,(d)},
\\
& z_{2,(1)},\ldots,z_{2,(d)},
\overline{z}_{2,(1)},\ldots,\overline{z}_{2,(d)},
\\
&\mbox{~}\mbox{~}\vdots
\\
&z_{r,(1)},\ldots,z_{r,(d)},
\overline{z}_{r,(1)},\ldots,\overline{z}_{r,(d)}),
\end{align*}
which also sorts the triples $(\iota,i,k)$, $\iota\in\{\texttt{barred},\texttt{unbarred}\}$,
$1\leq i \leq r$, $1 \leq k \leq d$.
We write $\textup{pred}(\iota,i,k)$ for the predecessor of $(\iota,i,k)$,
which is the successor when traversing $\mathbf z$ row-wise from right to left, bottom to top.
The property $\textup{jump}(\iota,i,k)$ is defined to be true if and only if $\textup{pred}(\iota,i,k)=(\texttt{barred},i-1,d)$.
This captures a jump in the row when traversing $\mathbf z$ row-wise from right to left, bottom to top.
\[
h_{\iota,i,k} = \sum_{\substack{S\subseteq Z_{\leq(\iota,i,k)}\\|S|=d\\S\textup{ consistent}}} \prod_{z\in S} z
\qquad\qquad\qquad\qquad
h'_{\iota,i,k} = \sum_{\substack{S\subseteq Z_{\leq(\iota,i,k)}\\|S|=d\\S\textup{ consistent}\\\forall j: z_{i,(j)}\notin S}} \prod_{z\in S} z
\]
We have $h=h_{\texttt{unbarred},r,d}$.
We now observe
\begin{eqnarray*}
h_{\iota,i,k} &=& \begin{cases}
h_{\textup{pred}(\iota,i,k),d} + z_{\iota,i,k}\cdot h'_{\textup{pred}(\iota,i,k),d-1} & \textup{ if $\iota=\texttt{barred}$}
\\
h_{\textup{pred}(\iota,i,k),d} + z_{\iota,i,k}\cdot h_{\textup{pred}(\iota,i,k),d-1} & \textup{ if $\iota=\texttt{unbarred}$}
\end{cases}
\\
h'_{\iota,i,k} &=&
\begin{cases}
h'_{\textup{pred}(\iota,i,k),d} + z_{\iota,i,k}\cdot h'_{\textup{pred}(\iota,i,k),d-1} & \textup{ if $\iota=\texttt{barred}$} \textup{ and not }\textup{jump}(\iota,i,k)
\\
h'_{\textup{pred}(\iota,i,k),d}  & \textup{ if $\iota=\texttt{unbarred}$} \textup{ and not }\textup{jump}(\iota,i,k)
\\
h_{\textup{pred}(\iota,i,k),d} + z_{\iota,i,k}\cdot h_{\textup{pred}(\iota,i,k),d-1} & \textup{ if $\iota=\texttt{barred}$} \textup{ and }\textup{jump}(\iota,i,k)
\\
h_{\textup{pred}(\iota,i,k),d}  & \textup{ if $\iota=\texttt{unbarred}$} \textup{ and }\textup{jump}(\iota,i,k)
\end{cases}
\end{eqnarray*}
This gives an ABP of the desired width.
A simpler but slightly larger variant is
\[
h_{r,d} \ = \ 
\sum_{k=0}^d\big(e_{d,k}(z_{r,(1)},\ldots,z_{r,(d)})
+e_{d,k}(\overline{z}_{r,(1)},\ldots,\overline{z}_{r,(d)})
\big)\cdot h_{r-1,d-k}.\qedhere
\]
\end{proof}
Let $h = \pol(h_{\textup{poly}})$.
Since $h_{\textup{poly}}\leq \Gamma_{2dr,d}$, it follows $h \leq \pol(\Gamma_{2dr,d}) = H_{2dr,d}$, see \eqref{eq:defnonsfHnd}.

We define $\mathcal T\in \cE$ via mapping every variable to zero with the following exceptions:
\[
\mathcal T(z_{i,(k)}\boxtimes y_i^{(k)}) = \epsilon^{(k)},
\mathcal T(\overline{z}_{i,(k)}\boxtimes \overline{y}_i^{(k)}) = \epsilon^{(k)},
\mathcal T(z_{i,(i)}\boxtimes\qu{x}_{i}) = x_i,
\mathcal T(\overline{z}_{i,(i)}\boxtimes\qu{x}_{i}) = \overline{x}_i,
\mathcal T(z_{m+1,(m+1)}\boxtimes y_{\textup{o}}) = \epsilon^{(m+1)},
\]
where the $\qu{x}_{i}\in\Var$.

\begin{lemma}
\label{lem:inconcon}
Fix $B$ and the arithmetization $G$.
If an $s$-$t$-path $p$ in $G$ is inconsistent, then 
$\mathcal T(h \boxtimes (\qu{x}_1\cdots\qu{x}_m \, y_{\textup{o}} \, p)) = 0$.
If $p$ in $G$ is consistent, then $\mathcal T(h \boxtimes (\qu{x}_1\cdots\qu{x}_m \, y_{\textup{o}} \, p)) = B(\mathbf b_p) (x,\overline{x})^{\mathbf b_p,\overline{\mathbf b}_p} \prod_{i=m+1}^{d} \epsilon^{(i)}$.
\end{lemma}
\begin{proof}
If $p$ is inconsistent, then there is $i$ such that both $y_{i}^{(k)}$ and $\overline{y}_{i}^{(\widetilde k)}$ exist in $p$ for some $k,\widetilde k$.
Hence, to have nonzeroness, a monomial in $h_\textup{poly}$ must have both $z_{i,(k)}$ and $\overline{z}_{i,(\widetilde k)}$ as variables, but this is impossible by construction.
This proves the first part.

Let $p$ be consistent with corresponding Boolean assignment $(\mathbf b,o,\mathbf b')$ satisfying $\mathcal C(\mathbf b,o,\mathbf b')=1$.
If $B(\mathbf b)=1$, then it is straighforward by pairing up variables to find a monomial $c_\textup{poly}$ in $h_{\textup{poly}}$ that satisfies $\mathcal T(c\boxtimes(\qu{x}_1\cdots\qu{x}_m \, y_{\textup{o}} \, p)) = (x,\overline x)^{\mathbf b,\overline{\mathbf b}} \prod_{i=m+1}^d \epsilon^{(i)}$,
where $c=\pol(c_\textup{poly})$.

It remains to prove that for every monomial $c_{\textup{poly}}$ in $h_{\textup{poly}}$, $c:=\pol(c_{\textup{poly}})$, we have that
$\mathcal T(c  \boxtimes (\qu{x}_1\cdots\qu{x}_m \, y_{\textup{o}} \, p)) \neq 0$ implies both 
$B(\mathbf b_p)=1$ and $\mathcal T(c  \boxtimes (\qu{x}_1\cdots\qu{x}_m \, y_{\textup{o}} \, p)) = (x,\overline x)^{\mathbf b,\overline{\mathbf b}} \prod_{i=m+1}^d \epsilon^{(i)}$.
By construction of $\mathcal T$, the nonvanishing $\mathcal T(c \boxtimes(\qu{x}_1\cdots\qu{x}_m \, y_{\textup{o}} \, p)) \neq 0$ fixes $r-m-1$ many variables in $c$.
For the remaining variables, by the consistency of $h$,
$c$ contains $z_{i,(i)}$ if $b_i=1$ (because $\overline{z}_{i,(i)}$ is unavailable due to the consistency of $h$), and $c$ contains $\overline{z}_{i,(i)}$ otherwise (by the same argument).
The nonvanishing also implies that $c$ contains $z_{m+1,(m+1)}$ and by consistency of $h$ another $z_{m+1,(k)}$ for $k>m+1$.
This implies $o=1$, which by Lemma~\ref{lem:tseytin} implies $B(\mathbf b)=1$.
Now it is readily checked that $\mathcal T((\qu{x}_1\cdots\qu{x}_m \, y_{\textup{o}} \, p) \boxtimes c) = (x,\overline x)^{\mathbf b,\overline{\mathbf b}} \prod_{i=m+1}^d \epsilon^{(i)}$.
\end{proof}

\begin{proof}[Proof of Theorem~\ref{thm:booleancktsascoeffs}]
Given $B$ and the arithmetization ABP $G$ computing $q$. We have
\begin{eqnarray*}
\mathcal T(h \boxtimes(\qu{x}_1\cdots\qu{x}_m \, y_{\textup{o}} \, q))
&\stackrel{\textup{Lem.}~\ref{lem:inconcon}}{=}&
\sum_{\substack{p \textup{ consistent in } q \\ B(\mathbf b_p)=1}}
(x,\overline{x})^{\mathbf b_p,\overline{\mathbf b}_p} \prod_{i=m+1}^{d} \epsilon^{(i)}
\\
&\stackrel{\eqref{eq:bijectionsbcp}}{=}&
\sum_{\substack{\mathbf b,o,\mathbf b'\\\mathcal C(\mathbf b,o,\mathbf b')=1\\B(\mathbf b)=1}}
(x,\overline{x})^{\mathbf b,\overline{\mathbf b}} \prod_{i=m+1}^{d} \epsilon^{(i)}
\ \stackrel{\textup{Lem.}~\ref{lem:tseytin}}{=} \ 
\sum_{B(\mathbf b)=1}
(x,\overline{x})^{\mathbf b,\overline{\mathbf b}} \prod_{i=m+1}^{d} \epsilon^{(i)}.
\end{eqnarray*}

\vspace{-1.3cm}

\end{proof}

\section{P $\neq$ NP and related conjectures}
\label{sec:PNP}

In this section set set $R=\IC$ and $A=\sS$.
Figure~\ref{fig:implicationsconjectures} shows formal implications between several conjectures.
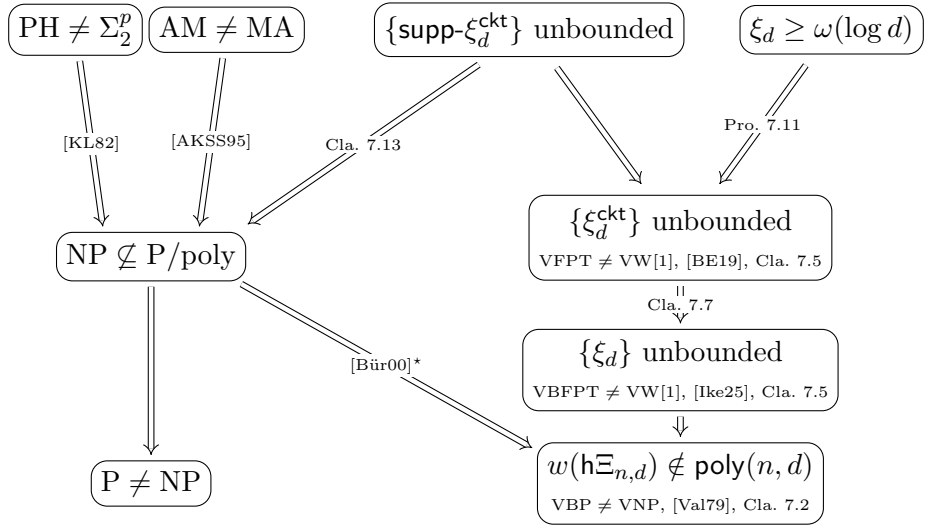
\begin{figure}[ht]
\centering
\begin{tikzpicture}[yscale=-1]
\node[draw,align=center,rounded corners=6pt] (sigma) at (2,-3) {$\textup{PH} \neq \Sigma^p_2$};
\node[draw,align=center,rounded corners=6pt] (AMMA) at (4,-3) {$\textup{AM} \neq \textup{MA}$};
\node[draw,align=center,rounded corners=6pt] (ppoly) at (3,0) {$\textup{NP}\not\subseteq\textup{P}/\textup{poly}$};
\node[draw,align=center,rounded corners=6pt] (PNP) at (3,3) {$\textup{P} \neq\textup{NP}$};
\node[draw,align=center,rounded corners=6pt] (wHC) at (10,3) {$\w(\hXi_{n,d})\notin\poly(n,d)$\\{\tiny $\VBP\neq\VNP$, \cite{Val79}, Cla.~\ref{cla:Valiant}}};
\node[draw,align=center,rounded corners=6pt] (unbdwexp) at (10,1.5) {$\{\xi_d\}$ unbounded\\{\tiny $\VBFPT\neq\VW[1]$, \cite{Ike25}, Cla.~\ref{cla:cktexp}}};
\node[draw,align=center,rounded corners=6pt] (unbdcktexp) at (10,-0.25) {$\{\xi^{\ckt}_d\}$ unbounded\\{\tiny $\VFPT\neq\VW[1]$, \cite{BE19}, Cla.~\ref{cla:cktexp}}};
\node[draw,align=center,rounded corners=6pt] (logwexp) at (12,-3) {$\xi_d\geq\omega(\log d)$};
\node[draw,align=center,rounded corners=6pt] (gammaIlog) at (8,-3) {$\{\suppdash\xi^\ckt_d\}$ unbounded};
\draw[double equal sign distance, -implies, shorten >=2pt, shorten <=2pt] (sigma) -- ($(ppoly.north)!0.5!(ppoly.north west)$) node [pos=0.5,fill=white,inner sep=1pt] {\tiny \cite{KL82}};
\draw[double equal sign distance, -implies, shorten >=2pt, shorten <=2pt] (AMMA) -- ($(ppoly.north)!0.5!(ppoly.north east)$) node [pos=0.5,fill=white,inner sep=1pt] {\tiny \cite{AKSS95}};
\draw[double equal sign distance, -implies, shorten >=2pt, shorten <=2pt] (ppoly) -- (PNP);
\draw[double equal sign distance, -implies, shorten >=2pt, shorten <=2pt] ($(ppoly.south east)!0.1!(ppoly.south)$) -- ($(wHC.north west)!0.2!(wHC.west)$) node [midway,fill=white,inner sep=1pt] {\tiny \cite{Bur00CvV}$^\star$};
\draw[double equal sign distance, -implies, shorten >=2pt, shorten <=2pt] (unbdwexp) -- (wHC);
\draw[double equal sign distance, -implies, shorten >=5pt, shorten <=5pt] (logwexp) -- (unbdcktexp) node[midway,fill=white,inner sep=1pt] {\tiny Pro.~\ref{pro:impliesBE}};
\draw[double equal sign distance, -implies, shorten >=2pt, shorten <=2pt] (unbdcktexp) -- (unbdwexp) node[midway,fill=white,inner sep=1pt] {\tiny Cla.~\ref{cla:impliesunbd}};
\draw[double equal sign distance, -implies, shorten >=5pt, shorten <=5pt] (gammaIlog) -- (unbdcktexp);
\draw[double equal sign distance, -implies, shorten >=5pt, shorten <=5pt] (gammaIlog) -- ($(ppoly.north east)!0.1!(ppoly.north)$) node[midway,fill=white,inner sep=1pt] {\tiny Cla.~\ref{cla:impliesPPoly}};
\end{tikzpicture}
\caption{Formal implications between conjectures. The setting is $R=\IC$, $A=\sS$. ${}^\star$The \cite{Bur00CvV} implication is only known to hold under the assumption of the Generalized Riemann Hypothesis.}
\label{fig:implicationsconjectures}
\end{figure}
We are very far from proving any of the conjectures in Figure~\ref{fig:implicationsconjectures},
as any of them either implies $\textup{P}\neq\textup{NP}$ or $\VBP\neq\VNP$.
No proof technique so far has even started scratching the surface:
At the moment, our best lower bounds are just $\xi_d\geq 2$,
proved via a simple argument \eqref{eq:atleasttwo}.

For an explanation of the conjectures $\textup{PH}\neq\Sigma^\textup{p}_2$ and $\textup{AM}\neq\textup{MA}$
we refer the reader to \cite{KL82} and \cite{AKSS95}, respectively.
We will discuss the other conjectures and their relations in the rest of this section,
as indicated by the edge labels in Figure~\ref{fig:implicationsconjectures}.
Implications without edge labels follow directly from the definitions.

An arithmetic circuit $C$ is a directed acyclic graph with a single outdegree 0 vertex,
whose indegree 0 vertices (called \emph{inputs}) are labeled with variables or constants,
and whose other vertices (called \emph{computation gates}) have indegree 2 and are each labeled with an operation $+,\times$.
An arithmetic circuit computes a polynomial by induction over the circuit structure.
The size of $C$ is defined as the number of the computation gates of~$C$.
We define $\csize(f)$ to be the size of a smallest arithmetic circuit computing $f$.
We define
\[
\gamma_d^\ckt \ := \ \inf\{k \mid \csize(\hGamma\!_{n,d}) \in O(n^k) \}.
\]

A \emph{weft 1} arithmetic circuit is allowed to have one so-called \emph{large gate} on each path from a leaf to the output,
where a large gate is a summation or multiplication gate with arbitrary indegree.

\subsection{Valiant's conjecture phrased in $n$ and $d$}
A sequence of multilinear homogeneous polynomials $(f_d)_{d\in\IN}$ is called \emph{Boolean-definable}
if there exists a sequence of Boolean circuits $(B_d)_{d\in\IN}$
with $|B_d|\in\poly(d)$ such that
\[
\forall d: \ f_d = \sum_{B_d(\mathbf b)=1} x^{\mathbf b}.
\]
Valiant's original conjecture (in slightly different phrasing) states that
\begin{equation}\label{eq:Valiantdefinable}
\textup{there exists a Boolean-definable sequence $(f_d)_{d\in\IN}$ with $\asize(f_d)\notin\poly(d)$.}
\end{equation}

\begin{claim}\label{cla:Valiant}
\eqref{eq:Valiantdefinable} holds
if and only if $\w(\hGamma\!_{n,d})\notin\poly(n,d)$.
\end{claim}
\begin{proof}
The proof goes by contraposition.

If $\w(\hGamma\!_{n,d})\in\poly(n,d)$, then by Theorem~\ref{thm:booleancktsascoeffs}
and by replacing every $\epsilon^{(i)}$ with 1,
we obtain $\asize(f_d)\in\poly(d)$ for every Boolean-definable sequence $(f_d)_{d\in\IN}$,
therefore \eqref{eq:Valiantdefinable} is false.

For the other proof direction, assume that \eqref{eq:Valiantdefinable} is false.
Note that $\hGamma\!_{d,d}$ is Boolean-definable.
Hence, $\asize(\hGamma\!_{d,d})\in\poly(d)$.
Therefore, $\asize(\hGamma\!_{n+d,n+d})\in\poly(n,d)$.
By setting some variables in $\hGamma\!_{n+d,n+d}$ to $1$ and $0$,
we get
$\asize(\hGamma\!_{n,d})\in\poly(n,d)$.
Using Proposition~\ref{pro:homogenization}, we obtain $\w(\hGamma\!_{n,d})\in \poly(n,d)$.
\end{proof}

\subsection{Algebraic fixed-parameter tractability}
\label{subsec:algfpt}
We define
\[
\Gamma\!_{n\geq d}^{\,\boxtimes m} := \begin{cases}
\Gamma\!_{n,d}^{\,\boxtimes m} & \textup{ if } n\geq d,
\\
0 & \textup{ otherwise.}
\end{cases}
\]
Note that the symbol $\geq$ in $\Gamma\!_{n\geq d}^{\,\boxtimes m}$
is purely syntactic, exactly as the comma in $\Gamma\!_{n,d}^{\,\boxtimes m}$.
We set $\hGamma\!_{n\geq d} := \Gamma\!_{n\geq d}^{\,\boxtimes 2}$.

A doubly indexed family of (not necessarily homogeneous) polynomials $f_{n,k}$ is called a \emph{parameterized p-family} (see \cite[Def~4.1]{BE19}) if
the number of variables is $\poly(n)$ and the degree is $\poly(n,k)$.
The number of variables of $\Gamma\!_{n\geq d}^{\,\boxtimes m}$ is less than $n^{3m}$,
and the degree of $\Gamma\!_{n\geq d}^{\,\boxtimes m}$ is at most $d$, hence
for every fixed~$m$, $\Gamma\!_{n\geq d}^{\,\boxtimes m}$ is a parameterized p-family.
The class $\VFPT$ consists of those parameterized p-families $f_{n,k}$ for which
there exists a function $\phi:\IN\to\IN$ and a polynomial $p\in\poly(n)$
with $\forall n,k: \, \csize(f_{n,k})\leq \phi(k)\cdot p(n)$.
Analogously,
The class $\VBFPT$ consists of those parameterized p-families $f_{n,k}$ for which
there exists a function $\phi:\IN\to\IN$ and a polynomial $p\in\poly(n)$
with $\forall n,k: \, \asize(f_{n,k})\leq \phi(k)\cdot p(n)$, see \cite{Ike25}.
A parameterized p-family $f_{n,k}$ is called \emph{weft-$1$-definable} if there exists a constant $\delta\in\IN$ such that
$f_{n,k}$ can be written in the form (up to renaming variables)
\[
f_{n,k} \ = \ \sum_{\substack{\mathbf b\\|\mathbf b|=k}} h_n(x_1,\ldots,x_{p(n)},b_1,\ldots,b_{q(n)})
\]
for $p(n)\in\poly(n)$ and $q(n)\in\poly(n)$ and $h_n$ can be computed by a weft $1$ depth $\delta$ arithmetic circuit of size $\poly(n)$. 
The class $\VW[1]$ consists of the weft-$1$-definable parameterized $p$-families.
The $\VFPT\neq\VW[1]$ conjecture states that there exists a weft-$1$-definable parameterized p-family that is not in $\VFPT$.
This conjecture is equivalent to ${}_k\Clique_{n}\notin\VFPT$, which is the main result of \cite{BE19} (for a slightly different variant of ${}_k\Clique_{n}$).
The following lemma states that weft $1$ is sufficient for any $m$-th power $\Gamma\!_{n\geq d}^{\,\boxtimes m}$.
\begin{lemma}
\label{lem:hGammaweftdef}
For every fixed $m$: $\Gamma\!_{n\geq d}^{\,\boxtimes m}$ is weft-$1$-definable. 
\end{lemma}
\begin{proof}
We prove that $\hGamma\!_{n\geq d}$ is weft-$1$-definable.
The proof for $m=2$ or higher orders $m$ is completely analogous.
Let $E$ denote the set of variables of $\hGamma\!_{n\geq d}$, and note $E=\emptyset$ of $|E|=(2n+(d-2)n^2)^2\leq n^6$.
We say that the pair $(e_{i,1}\boxtimes e_{i,2},e_{j,1}\boxtimes e_{j,2})$ is \emph{$1$-extendable} if $e_{i,1}$ and $e_{j,1}$ can be extended to a direct source-sink path by adding $d-2$ more edges.
Analogously, define $(e_{i,1}\boxtimes e_{i,2},e_{j,1}\boxtimes e_{j,2})$ to be \emph{$2$-extendable} if $e_{i,2}$ and $e_{j,2}$ can be extended to a path.
The key property is that a size $d$ set $S\subseteq E$ forms a pair of direct source-sink paths (i.e., a monomial in $\hGamma\!_{n\geq d}$)
if and only if all elements of $S$ are pairwise both $1$-extendable and $2$-extendable:
This is clearly necessary, and it is also sufficient, because being a pair of direct source-sink paths in level $1$ is a local property given by the $1$-extendability, and analogously being a pair of direct source-sink paths in level $2$.
Note that $\mathbf x^{\mathbf b} = \prod_{j\in E}(x_j b_j +1 - b_j)$, which is a product of depth 3 circuits.
\[
\hGamma\!_{n\geq d} \ = \ \sum_{\substack{\mathbf b\\|\mathbf b|=d}} \mathbf x^{\mathbf b} \left(\prod_{\substack{i,j\in E\\(i,j) \textup{ not 1-extendable}}} (1-b_i b_j)\right) \left(\prod_{\substack{i,j\in E\\(i,j) \textup{ not 2-extendable}}} (1-b_i b_j)\right).
\]
This is a sum of one large product of depth 3 circuits, i.e., a weft 1 depth 4 circuit.
The circuit size of the inner circuit is $O(|E|^2)\in\poly(n)$.
Hence, $\hGamma\!_{n\geq d}$ is weft-$1$-definable.
\end{proof}

\begin{proposition}
\label{pro:hGammaVFPT}
$\VFPT\neq\VW[1]$ if and only if $\hGamma\!_{n\geq d}\notin\VFPT$.
Moreover, $\VBFPT\neq\VW[1]$ if and only if $\hGamma\!_{n\geq d}\notin\VBFPT$.
\end{proposition}
\begin{proof}
Using Lemma~\ref{lem:hGammaweftdef}, if $\hGamma_{n\geq d}\notin\VFPT$, then $\hGamma_{n\geq d}$ is weft-$1$-definable but not in $\VFPT$, which means $\VFPT\neq\VW[1]$ by definition.
Analogously, $\hGamma_{n\geq d}\notin\VBFPT$ implies $\VBFPT\neq\VW[1]$.
For the other direction, let $\hGamma_{n\geq d}\in\VFPT$
with $\phi:\IN\to\IN$ and $\csize(\hGamma_{n\geq d})\leq \phi(d)\cdot p(n)$, where $p(n)\in\poly(n)$.
Proposition~\ref{pro:wclique} implies that for
$d=\binom{k+1}{2}$
we have
$\csize({}_k\Clique_{n})\leq \phi(\binom{k+1}{2})\cdot p(n)$,
which can be written as
$\csize({}_k\Clique_{n})\leq \phi'(k)\cdot p(n)$ for $p(n)\in\poly(n)$.
Therefore, ${_k}\Clique_{n}\in\VFPT$, which implies $\VFPT=\VW[1]$ by \cite{BE19}.
For $\VBFPT$ the proof is analagous by noting that Lemma~4.9 in \cite{BE19} works also for~$\VBFPT$.
Alternatively, one can prove both cases without going via cliques in a similar way as Theorem~\ref{thm:booleancktsascoeffs}.
\end{proof}

\begin{claim}
\label{cla:cktexp}
$\VFPT\neq\VW[1]$ if and only if $\{\xi^{\ckt}_d\mid d \in \IN\}$ is unbounded.
Moreover,
$\VBFPT\neq\VW[1]$ if and only if $\{\xi_d\mid d \in \IN\}$ is unbounded.
\end{claim}
\begin{proof}
We use Proposition~\ref{pro:hGammaVFPT}.
The proof goes by contrapositions.

Let $\hGamma\!_{n\geq d} \in \VFPT$, i.e.,
there exists $\phi$ and $\alpha$ with 
$\forall d,n: \csize(\hGamma\!_{n\geq d}) \leq \phi(d) \cdot n^\alpha$.
Hence, $\{\gamma^{\ckt}_d\}$ is bounded by $\alpha$ and thus $\{\xi^{\ckt}_d\}$ is bounded by $\alpha+4$, see \eqref{eq:gammaxifourplusgamma}.

Let $\{\xi^{\ckt}_d\}$ be bounded by $\alpha\in\IN$,
hence $\{\gamma^{\ckt}_d\}$ is also bounded by $\alpha$.
Then $\forall d: \csize(\hGamma\!_{n\geq d}) \in O(n^\alpha) \leq \phi(d)\cdot n^{\alpha}$ for some $\phi(d)$.
Hence, $\hGamma\!_{n\geq d}\in \VFPT$.

The proof for $\VBFPT$ is completely analogous.
\end{proof}

\begin{lemma}
\label{lem:cktfromprogram}
Then $\csize(f)\leq 2d \big(\asize(f)\big)^2$.
\end{lemma}
\begin{proof}
The evaluation of an ABP of width $n$ and degree $d$ can naively be implemented as an arithmetic circuit.
Every edge that is not incident to the source gives rise to a product gate, and every vertex besides the source gives rise to $n-1$ summation gates.
This gives at most $(d-2)n^2+n$ many product gates  
and $(n-1)((d-2)n+1)$ many summation gates.         
In total, these are at most $2 n^2 d$ many gates.
\end{proof}
While an increase of $1$ in $\asize$ requires several arithmetic operations to realize in a circuit,
not all algorithmic tricks in the circuit model have known equivalences in ABPs,
for example an algebraic variant of the technique in \cite[\S2]{NP85} can be applied to get better upper bounds on $\csize(\hXi_{n,d})$,
but it is unclear how to apply it to $\asize$.

\begin{claim}
\label{cla:impliesunbd}
$\VBFPT\subseteq\VFPT$.
\end{claim}
\begin{proof}
For $f_{n,d}\in\VBFPT$ we have $\asize(f_{n,d})\leq n^\alpha\phi(d)$ for some $\phi:\IN\to\IN$ and some $\alpha\in\IN$.
Hence, by Lemma~\ref{lem:cktfromprogram}, $\csize(f_{n,d}) \leq 2d \big(n^\alpha\phi(d)\big)^2 \leq n^{2\alpha}\phi'(d)$ for some $\phi':\IN\to\IN$.
Therefore $f_{n,d}\in\VFPT$.
\end{proof}

\begin{lemma}
\label{lem:depthreduction}
If $f$ of degree $d\geq 2$ (not necessarily homogeneous) is computed by an arithmetic circuit of size $s\geq 2$, then
$f$ is computed by an arithmetic circuit of depth at most $12 \log(s)\log(d)$.
\end{lemma}
\begin{proof}
This is a weaker result than \cite{VSBR83}, but it suffices here and has a much shorter proof, as pointed out in \cite{Sap15}.
The degree $\deg(v)$ of a vertex $v$ in an arithmetic circuit is the degree of the (not necessarily homogeneous) polynomial computed at $v$.
Given a circuit $C$ computing $f$, let
\[
D=\{v\in C\mid \tfrac d 3 < \deg(v) \leq \tfrac{2d}{3}\}.
\]
We convert $C$ into $C'$ by replacing every subcircuit $C_v$ with $v\in D$ by a new input variable $y_v$,
and call the computed polynomial $f'$.
Collecting $y$-variables, we write
$
f' = \sum_{i,j} A_{i,j} y_i y_j \ + \ \sum_{i} B_{i} y_i \ + \ p
$
with each $A_{i,j},B_i,p$ not involving any $y$-variables, and
each $\deg(A_{i,j}),\deg(B_i)\leq 2d/3$, $\deg(p)\leq \frac d 3$.
For a set $S\subseteq\IN$, 
we write $f'_S$ for the partial evaluation in which we set $y_i$ to 1 if $i\in S$,
and $y_i$ to 0 otherwise.
Note that $B_i = f'_{\{i\}}$, and $A_{i,j}=f'_{\{i,j\}}-f'_{\{i\}}-f'_{\{j\}}$.
\begin{equation}
\label{eq:fpassum}
f \ = \ \sum_{i,j} (f'_{\{i,j\}}-f'_{\{i\}}-f'_{\{j\}}) C_i C_j \ + \ \sum_{i} f'_{\{i\}} C_i \ + \ p.
\end{equation}
The $s^2$ many summation gates in \eqref{eq:fpassum} can be implemented in depth $\lceil2\log(s)\rceil+3$.
Hence, \eqref{eq:fpassum} is a depth $\lceil 2\log(s)\rceil+4$
circuit stacked on top of polynomials of degree at most $2d/3$ that have size at most~$s$.
We apply the construction recursively to these polynomials.
The resulting depth satisfies
$\textup{Depth}(d) = \textup{Depth}(2d/3)+2\log(s)+5$,
hence $\textup{Depth}(d) \leq \log_{3/2}(d) (2\log(s)+5) \leq 12 \log(s) \log(d)$.
\end{proof}

An arithmetic circuit whose graph structure is a tree is called an arithmetic formula.
We define $\fsize(f)$ to be the size of a smallest arithmetic formula computing~$f$.
\begin{lemma}
\label{lem:ABPfromformula}
For every (not necessarily homogeneous) polynomial $f$, we have $\asize(f)\leq\fsize(f)$.
\end{lemma}
\begin{proof}
The proof is by induction.
For linear polynomials $f$, $\fsize(f) = 0 = \asize(f)$.
For a sum $f+g$ of minimal formula size,
we have $\fsize(f+g)= \fsize(f)+\fsize(g)+1$,
but $\asize(f+g)\leq\asize(f)+\asize(g)$ by
taking minimal size ABPs for $f$ and $g$ and identifying their sources
and identifying their sinks.
For a product $fg$ of minimal formula size,
we have $\fsize(fg)= \fsize(f)+\fsize(g)+1$
and $\asize(fg)\leq\asize(f)+\asize(g)+1$
by taking minimal size ABPs for $f$ and $g$ and identifying the sink for $f$
with the source for $g$ (note that $\asize$ does not count sources or sinks, and this creates a new vertex that is not a source or sink).
\end{proof}

\begin{proposition}
\label{pro:impliesBE}
If $\xi_d\geq\omega(\log d)$, then $\VFPT\neq\VW[1]$.
\end{proposition}
\begin{proof}
The proof goes by contraposition.
Let $\VFPT=\VW[1]$.
Then $\hGamma\!_{n\geq d} \in \VFPT$.
Hence, there exists $\alpha\in\IN$ and a function $\phi$ such that $\forall n,d: \, \csize(\hGamma\!_{n\geq d}) \leq \phi(d)\cdot n^\alpha$.
We apply Lemma~\ref{lem:depthreduction} to obtain a circuit for $\hGamma\!_{n\geq d}$
of depth $12 \log(d)\log(\phi(d)n^\alpha)$.
We expand this arithmetic circuit into an arithmetic formula of size $2^{12 \log(d)\log(\phi(d)n^\alpha)} = (\phi(d)n^\alpha)^{12 \log d}$.
Using Lemma~\ref{lem:ABPfromformula}, we obtain $\asize(\hGamma\!_{n\geq d})\leq (\phi(d)n^\alpha)^{12 \log d}$.
We apply Proposition~\ref{pro:homogenization} to obtain
$\w(\hGamma\!_{n\geq d})\leq (\phi(d)n^\alpha)^{12 \log d}$.
For a fixed $d$, this is in $O(n^{12 \alpha \log d})$.
Hence, $\gamma_d\leq O(\log d)$,
and thus also $\xi_d\leq O(\log d)$.
\end{proof}

The following proposition is
analogous to Proposition~\ref{pro:collapse},
and it
shows that it is sufficient to study second-order computation only and not also higher orders.
\begin{proposition}\label{pro:collapseparam}
There exists $\phi:\IN\to\IN$, $\phi(d)\geq d$, with $\w(\hXi_{n,d})\leq \phi(d)\poly(n)$ if and only if for every $m$
there exists $\varphi:\IN\to\IN$, $\varphi(d)\geq d$, with
$\w(\textup{\textsf{H}}_{n,d}^{\boxtimes m}\boxtimes \Xi_{n,d})\leq \varphi(d)\poly(n)$.
\end{proposition}
\begin{proof}
One direction is clear.
For the other direction, fix any $m$.
Let $\varphi:\IN\to\IN$ and $k\in\IN$ such that $\forall d,n: \w(\hXi_{n,d})\leq n^k \varphi(d)$,
in other words $\hXi_{n,d}\leq \Xi_{n^k \varphi(d),d}$.
\begin{eqnarray*}
\textsf{H}_{n,d}^{\boxtimes m}\boxtimes \Xi_{n,d}
&=&
\textsf{H}_{n,d}^{\boxtimes m-1}\boxtimes (\textsf{H}_{n,d}\boxtimes \Xi_{n,d})
\ \leq \ 
\textsf{H}_{n,d}^{\boxtimes m-1}\boxtimes \Xi_{n^k \varphi(d),d}
\\
& = & 
\textsf{H}_{n,d}^{\boxtimes m-2}\boxtimes (\textsf{H}_{n,d}\boxtimes \Xi_{n^k \varphi(d),d})
\ \leq \ 
\textsf{H}_{n,d}^{\boxtimes m-2}\boxtimes (\textsf{H}_{n^k \varphi(d),d}\boxtimes \Xi_{n^k \varphi(d),d})
\\
& \leq &
\textsf{H}_{n,d}^{\boxtimes m-2}\boxtimes \Xi_{n^{k^2} \varphi(d)^{k+1},d}
\ \leq \ 
\textsf{H}_{n,d}^{\boxtimes m-3}\boxtimes \Xi_{n^{k^3} \varphi(d)^{k^2+k+1},d}
\\
&\leq& \cdots \ \ \leq \ \Xi_{n^{k^m} \varphi(d)^{k^{m-1}+k^{m-2}+\cdots+1},d} \ .
\end{eqnarray*}
Since $k$ and $m$ are constant, this means $\w(\textsf{H}_{n,d}^{\boxtimes m}\boxtimes \Xi_{n,d})\leq \poly(n)\phi(d)
$ with $\phi(d)=\varphi(d)^{k^{m-1}+k^{m-2}+\cdots+1}$.
\end{proof}

\subsection{The connection to Boolean circuits}
We define $\exists\hGamma\!_{n\geq d}: \{0,1\}^{n^4 d^2}\to\{0,1\}$ via
$\exists\hGamma\!_{n \geq d}(\mathbf b)=\min\{1,\hGamma\!_{n \geq d}(\mathbf b)\}$.
We phrase the conjectures $\textup{NP}\not\subseteq\textup{P/poly}$
and $\textup{P}\neq\textup{NP}$
via their negations as follows.
We have
$\textup{NP}\subseteq\textup{P/poly}$
if and only if there exists a sequence of Boolean circuits of size $\poly(d)$
that compute $\exists\hGamma\!_{d\geq d}$.
We have $\textup{P}=\textup{NP}$
if and only if there exists a sequence of Boolean circuits $C_d$ of size $\poly(d)$
that compute $\exists\hGamma\!_{d\geq d}$
with the additional requirement that there exists an algorithm $A$
that on input $d$ runs in time $\poly(d)$ and outputs the adjacency matrix and gate labels of $C_d$.
For $h\in \sS_d$ we define
\[
\suppdash\csize(h) := \min\{ \csize(f) \mid f\in \sS_d, \ f(\mathbf b)=0 \textup{ iff } h(\mathbf b)=0 \}.
\]
\[
\suppdash\gamma^\ckt_d := \inf\{ k \mid \suppdash\csize(\hGamma\!_{n\geq d}) \in O(n^k) \} = \limsup_{n\to\infty}\big(\log_n(\suppdash\csize(\hGamma\!_{n\geq d}))\big).
\]

\begin{claim}
\label{cla:impliesPPoly}
If $\suppdash\gamma^\ckt_d$ is unbounded, then $\textup{NP}\not\subseteq\textup{P/poly}$.
\end{claim}
\begin{proof}
The proof goes by contraposition.
For $f:\{0,1\}^N\to\{0,1\}$,
let $\bsize(f)$ denote the size of a smallest Boolean circuit computing~$f$.
Let $\textup{NP}\subseteq\textup{P/poly}$.
Since $\exists\hGamma\!_{d\geq d} \in \textup{NP}$, it follows that there exists $\alpha\in\IN$ with $\forall d: \bsize(\exists\hGamma\!_{d\geq d}) \leq d^\alpha$.
Hence, $\forall d,n: \bsize(\exists\hGamma\!_{n\geq d}) \leq n^\alpha$.
We arithmetize the Boolean circuit by replacing $\NOT$ with $1-x$, $\AND$ with $x\cdot y$,
and $\OR$ with $xy-x-y$.
Evaluated on the Boolean hypercube, this circuit has the same support as $\exists\hGamma\!_{n\geq d}$.
The size of the arithmetic circuit is at most 4 times larger than the size of the Boolean circuit.
Hence, $\suppdash\gamma^\ckt_d$ is bounded by~$\alpha$.
\end{proof}

\section{Stabilizers and other properties}\label{sec:geomrepr}
In this section we have $R=\IC$ and $A=\sS$ (unless stated otherwise).
\subsection{Different settings and their properties}
\label{subsec:differentsettings}
Given a group $G$ acting linearly on a finite dimensional complex vector space $V$,
the \emph{stabilizer} of $f\in V$ is defined as $\stab_G(f)=\{g\in G\mid gf=f\}$.
For any subgroup $H\leq G$, the invariant space $V^H$ is defined as $\{v\in V\mid \forall h\in H: hv=v\}$.
We say that $f$ is \emph{characterized by its symmetries} if $\dim V^{\stab_G(f)} = 1$.

Recall that $\Gamma\!_{n,d}$ is the bottom-right entry of $\Xi_{n,d}$.
Let $\trXi_{n,d} := \tr(\Xi_{n,d})$ be the trace.
In this section we analyze the objects $\trXi_{n,d}$,
$\Gamma\!_{n,d}$, 
and $\Xi_{n,d}$ by the following properties:
\begin{enumerate}
\item the object's endomorphism orbit is topologically closed for $A=\sT$, by which we mean
$\{T(\Xi_{n,d})T\in \Mat_n(\IC)\times\End\Mat_n(\IC)\}$ is closed, and analogous statements for the other objects,
\item the object has a reductive stabilizer,
\item the object is characterized by its stabilizer,
\item the object has a connected stabilizer,
\item on plugging in identity matrices the object is invertible in every characteristic
\end{enumerate}
\begin{itemize}
\item
It is clear that $\trXi_{n,d}$ lacks property~5, because $\trXi_{n,d}(\id)=n$.
Moreover, $\trXi_{n,d}$ also lacks property~1, see \cite{BIMPS20}.
The fact that $\trXi_{n,d}$ also lacks property~4 was proved in \cite{Ges16}.

\item
$\Gamma\!_{n,d}$ has property~1, which was proved in \cite{Nis91}, see also \cite{For16}.
Notably, $\Gamma\!_{n,d}$ has property~5.
However,
$\Gamma\!_{n,d}$ lacks property~2, which can be seen by calculating for $d=3$, $n=2$, the Lie algebra that is the kernel of the Lie algebra action of $\gl_{2n+(n-2)^2 d}$ on $\Gamma\!_{n,d}$, and then verifying non-reductivity, for example with GAP.

\item
For $\Xi_{n,d}$ we have property~1, see Theorem~\ref{thm:Nisanmatrix}.
Clearly, $\Xi_{n,d}$ has property~5.
Properties 2, 3, and 4 are concerned with the stabilizer, hence they depend on the group action.
Let $\mu_d\subset\GL_{n^2 d} = \{\diag(\alpha,\ldots,\alpha)\mid \alpha^d=1\}$, and let
$G_{n,d} = \GL_{n^2 d}/\mu_d$, which is a connected reductive linear algebraic Lie group with Lie algebra $\gl_{n^2 d}$.
Note that $\GL_{n,d}$ acts faithfully on $\Sym^d\IC^{n^2 d}$.
We distinguish 3 different group actions:
\begin{enumerate}
\item
The action of $G_{n,d}$ on $\Mat_n(\Sym^d \IC^{n^2 d})$
given by $g (e_1 \otimes f \otimes e_j^\ast) = e_1 \otimes g(f) \otimes e_j^\ast$,
\item
The action of $\GL_{n}\times G_{n,d}$ on $\Mat_n(\Sym^d \IC^{n^2 d})$
given by $(g_\textup{outer},g_{\textup{inner}}) (g_{\textup{outer}} e_1 \otimes f \otimes e_j^\ast) = (g_{\textup{outer}}\,e_1) \otimes g_{\textup{inner}}(f) \otimes (e_j^\ast \, g_{\textup{outer}}^{-1})$,
\item 
The action of $\GL_{n}\times G_{n,d}\times \GL_{n}$ on $\Mat_n(\Sym^d \IC^{n^2 d})$
given by $(g_\textup{left},g_{\textup{inner}},g_\textup{right}) (e_1 \otimes f \otimes e_j^\ast) =
 (g_{\textup{left}} e_1) \otimes g_{\textup{inner}}(f) \otimes (e_j^\ast g_{\textup{right}}^\transpose)$.
\end{enumerate}
Under the first action, $\Xi_{n,d}$ lacks property~3, because if $g(\Xi_{n,d})=\Xi_{n,d}$, then also
$g(\trXi_{n,d})=\trXi_{n,d}$, but $\Xi_{n,d}\neq \trXi_{n,d}$.
This does not contradict the fact that $\trXi_{n,d}$ has property~3, because the ambient space under consideration there is $\Sym^d\IC^{n^2 d}$, which is smaller than $\Mat_n(\Sym^d\IC^{n^2 d})$.

Under the third action, note that the actions of the three factors of $\mathcal G := \GL_{n}\times G_{n,d}\times \GL_{n}$ commute.
Moreover, for every $g_{\textup{left}}\in\GL_n$ there exists $h_{g_{\textup{left}}}\in G_{n,d}$
such that $(g_{\textup{left}},\id,\id)\Xi_{n,d} = (\id,h_{g_{\textup{left}}},\id)\Xi_{n,d}$,
which is proved in the same way as in Lemma~\ref{lem:Xirestrict}.
Analogously, for every $g_{\textup{right}}\in\GL_n$ there exists $h'_{g_{\textup{right}}}\in G_{n,d}$
such that $(\id,\id,g_{\textup{right}})\Xi_{n,d} = (\id,h'_{g_{\textup{right}}},\id)\Xi_{n,d}$.
In other words, for everh $g\in\GL_n\times \GL_n$ there exists $h_g\in G_{n,d}$ with
$g\Xi_{n,d} = h_g \Xi_{n,d}$.
By the construction in Lemma~\ref{lem:Xirestrict} we also see that $\GL_n\times\GL_n\to G_{n,d}$, $g\mapsto h_g$, is a group homomorphism,
in particular $h_{g^{-1}}=(h_g)^{-1}$.
\begin{eqnarray}
\stab_{\mathcal G}(\Xi_{n,d}) &=&
\nonumber
\{(g,h)\mid gh\Xi_{n,d}=\Xi_{n,d}\}
\\
\nonumber
&=& \{(g,h)\mid h\Xi_{n,d}=g^{-1}\Xi_{n,d}\}
\\
\nonumber
&=& \{(g,h)\mid h\Xi_{n,d}=h_{g^{-1}}\Xi_{n,d}\}
\\
\nonumber
&=& \{(g,h)\mid h_{g^{-1}}^{-1} h\Xi_{n,d}=\Xi_{n,d}\}
\\
\nonumber
&=& \{(g,h)\mid h_g h \Xi_{n,d}=\Xi_{n,d}\}
\\
\nonumber
&=&
\{(g,h)\mid h_g h \in \stab_{G_{n,d}}(\Xi_{n,d})\}
\\
\label{eq:outerinnergroupsstab}
&=&
\{(g,h_{g^{-1}} h)\mid h \in \stab_{G_{n,d}}(\Xi_{n,d})\}
\end{eqnarray}
Thus, in order to determine $\stab_{\mathcal G}(\Xi_{n,d})$,
it suffices to determine $\stab_{G_{n,d}}(\Xi_{n,d})$.

Every $h$ stabilizing the matrix $\Xi_{n,d}$ also stabilizes its trace, therefore
$\stab_{G_{n,d}}(\Xi_{n,d})\subseteq \stab_{G_{n,d}}(\trXi_{n,d})$,
which has been determined in \cite{Ges16} as the group that maps
\[
(X^{(1)},\ldots,X^{(d)})
\]
to
\begin{equation}\label{eq:matrixinnerbasechange}
(X^{(1)} h_1, h_1^{-1}X^{(2)}h_2, \ldots,h_{d-1} X^{(d)})
\end{equation}
takes in semidirect product with the dihedral group that is generated by the maps that send
\[
(X^{(1)},\ldots,X^{(d)})
\]
to
\begin{equation}\label{eq:matrixtransposesinstab}
({X^{(d)}}^\transpose,\ldots,{X^{(1)}}^\transpose)
\end{equation}
and 
\[
(X^{(1)},\ldots,X^{(d)})
\]
to
\begin{equation}\label{eq:matrixcyclicshiftinstab}
(X^{(2)},\ldots,X^{(d)},X^{(1)}).
\end{equation}
While \eqref{eq:matrixinnerbasechange} stabilizes $\Xi_{n,d}$,
none of the $2d$ many elements generated by \eqref{eq:matrixtransposesinstab} and \eqref{eq:matrixcyclicshiftinstab}
stabilize $\Xi_{n,d}$, which can be seen by mapping
everything but the topleft $2\times 2$ matrices to zero, and observing that
\[
\begin{pmatrix}
1&1\\
0&1
\end{pmatrix}^{i-1}
\cdot
\begin{pmatrix}
1&0\\
0&2
\end{pmatrix}
\cdot
\begin{pmatrix}
1&1\\
0&1
\end{pmatrix}^{d-i}
=
\begin{pmatrix}
1&d+i-2\\
0&2
\end{pmatrix}
\]
and
\[
\begin{pmatrix}
1&0\\
1&1
\end{pmatrix}^{i-1}
\cdot
\begin{pmatrix}
1&0\\
0&2
\end{pmatrix}
\cdot
\begin{pmatrix}
1&0\\
1&1
\end{pmatrix}^{d-i}
=
\begin{pmatrix}
1&0\\
2d-i-1&2
\end{pmatrix}
\]
which are $2d$ many different matrices.
Hence, from \eqref{eq:outerinnergroupsstab} we conclude that $\stab_{\mathcal G}(\hXi_{n,d})$ consists of those $g\in\mathcal G$ that send
\[
(X^{(1)},\ldots,X^{(d)})
\]
to
\[
(h_0^{-1} X^{(1)} h_1, h_1^{-1}X^{(2)}h_2, \ldots, h_{d-1} X^{(d)} h_d),
\]
while simultaneously multiplying the matrix with $h_0$ from the left and $h_d^{-1}$ from the right.
This is isomorphic to $\GL_n^{d+1}$, which is reductive and connected, hence satisfying properties 2 and~4.
To see that $\Xi_{n,d}$ is characterized by its stabilizer, we have to determine $(\IC^n\otimes \Sym^d\IC^{n^2 d}\otimes \IC^{n\ast})^{\stab \Xi_{n,d}}$, but we perform a more general analysis by determining the representation theoretic multiplicities in $\IC[\mathcal G\Xi_{n,d}] \simeq \IC[\mathcal G]^{\stab_{\mathcal G}\Xi_{n,d}}$, see \cite[\S4]{BLMW11}, \cite[A.1.9]{Hut17}.
Let $\{\la\}_n$ denote the irreducible $\GL_n$-representation of type $\la$. We write $\{\la\}$ is $n$ is clear from the context.
Let $N=n^2 d$. The irreducible polynomial representations of $G_{n,d}$ are the same as for $\GL_{n^2}d$ with the only difference that $d$ must divide $\la_1+\cdots+\la_{n^2 d}$. Let $H:=\stab_{\mathcal G}(\Xi_{n,d})$.
We write $\IC^{n^2 d} = (V_0\otimes V_1)\oplus\cdots\oplus (V_{d-1}\otimes V_d)$, where $V_j\simeq \IC^n$.
We first restrict
\[
\{\mu\}\!\!\downarrow^{\GL_N}_{\GL(V_0\otimes V_1^*)\times\cdots\times\GL(V_{d-1}\otimes V_d^*)}
=
\bigoplus_{\rho^{(1)},\ldots,\rho^{(d)}\vdash_{n^2}} c^\mu_{\rho^{(1)},\ldots,\rho^{(d)}}\{\rho^{(1)}\}\otimes \cdots \otimes \{\rho^{(d)}\},
\]
where $c^\mu_{\rho^{(1)},\ldots,\rho^{(d)}}$ is the multi-Littlewood-Richardson coefficient.
We restrict further:
\begin{equation}\label{eq:fullyrestricted}
\{\mu\}\!\!\downarrow^{\GL_N}_{\GL(V_0)\times\GL(V_1^*)\times\GL(V_1)\times\cdots\times\GL(V_{d-1}^*)\times\GL(V_{d-1})\times\GL(V_d^*)}
\hspace{6cm}
\end{equation}
\[
=
\bigoplus_{\substack{\rho^{(1)},\ldots,\rho^{(d)}\\\iota^{(0)},\ldots,\iota^{(d-1)}\\\kappa^{(1)},\ldots,\kappa^{(d)}}}
c^\mu_{\rho^{(1)},\ldots,\rho^{(d)}}k(\iota^{(0)},\rho^{(1)},\kappa^{(1)})\cdots k(\iota^{(d-1)},\rho^{(d)},\kappa^{(d)})  \{\iota^{(0)}\}\otimes\{\kappa^{(1)}\} \otimes \cdots \otimes \{\iota^{(d-1)}\}\otimes \{\kappa^{(d)}\},
\]
where $k(\la,\mu,\nu)$ is the Kronecker coefficient.
Rename: $\la=\kappa^{(0)}$ and $\nu=\iota^{(d)}$, and take $H$-invariants, which forces $\forall i: \kappa^{(i)}=\iota^{(i)}$:
\[
\big(\{\kappa^{(0)}\}\otimes\{\mu\}\otimes\{\kappa^{(d)}\}\big)^H \ = \ 
\bigoplus_{\substack{\rho^{(1)},\ldots,\rho^{(d)}\\\kappa^{(1)},\ldots,\kappa^{(d-1)}}}
c^\mu_{\rho^{(1)},\ldots,\rho^{(d)}} k(\kappa^{(0)},\rho^{(1)},\kappa^{(1)})\cdots k(\kappa^{(d-1)},\rho^{(d)},\kappa^{(d)}) \, \IC
\]
Hence, all $\kappa^{(i)}$ have the same number of boxes, which is also the same as all $\rho^{(j)}$, which is $|\mu|/d$.
Writing this in multi-Kronecker notation, we have
\begin{equation}\label{eq:orbitmultikron}
\big(\{\kappa^{(0)}\}\otimes\{\mu\}\otimes\{\kappa^{(d)}\}\big)^H \ = \ 
\bigoplus_{\substack{\rho^{(1)},\ldots,\rho^{(d)}}}
c^\mu_{\rho^{(1)},\ldots,\rho^{(d)}} k(\kappa^{(0)},\rho^{(1)},\rho^{(2)},\ldots,\rho^{(d-1)},\rho^{(d)},\kappa^{(d)}) \, \IC
\end{equation}
To see that $\Xi_{n,d}$ is char by stab, we determine
\[
(\IC^n \otimes \Sym^d \IC^{n^2 d} \otimes \IC^n)^H = (\{(1)\}\otimes\{(d)\}\otimes\{(1)\})^H
\]
for which the above calculation has only $1$ summand, and $|\mu|/d=1$, the multi-Littlewood-Richardson coefficient is 1,
and the multi-Kronecker coefficient is 1, thus
$\dim(\IC^n \otimes \Sym^d \IC^{n^2 d} \otimes \IC^n)^H = 1$.
In conclusion, $\Xi_{n,d}$ with the third action has all 5 properties.
We now have a look at the second action, i.e., the conjugation action.
In this case, the stabilizer $H$ is smaller and when taking $H$-invariants in \eqref{eq:fullyrestricted},
instead of getting a $1$-dimensional invariant space,
we have to take $\GL_n$-invariants of $\IC^n\otimes(\IC^{n})^*\otimes\IC^n\otimes(\IC^{n})^*$,
which is 2-dimensional, because $\IC^n\otimes\IC^n = \Sym^2\IC^n \oplus \Wedge^2\IC^n$.
Hence, for this action $\Xi_{n,d}$ does not have property~3.
\end{itemize}

In conclusion, only $\Xi_{n,d}$, and only with the third action has all 5 properties.
In fact, properties 1, 2, and 3 are sufficient for ruling out all other settings.
We will continue working in this setting throughout \S\ref{sec:geomrepr}.

\subsection{The ambient space}
Since we are working with matrices of polynomials, it is worth recording the representation theory of
the space $\Sym^\delta(\IC^a \otimes \Sym^d \IC^N \otimes \IC^b)$
with respect to the action of the group $\GL_a \times \GL_N \times \GL_b$.
\begin{eqnarray*}
\Sym^\delta(\IC^a \otimes \Sym^d \IC^N \otimes \IC^b)
&=& 
\big(\tensor^\delta(\IC^a \otimes \Sym^d \IC^N \otimes \IC^b)\big)^{\aS_\delta}
\\
&=& 
\big(\tensor^\delta\IC^a \otimes \tensor^\delta\Sym^d \IC^N \otimes \tensor^\delta\IC^b\big)^{\aS_\delta}
\\
&=& 
\big((\bigoplus_{\la\vdash_a \delta}\{\la\}_a\otimes[\la]) \otimes \tensor^\delta\Sym^d \IC^N \otimes (\bigoplus_{\nu\vdash_b \delta}\{\nu\}_b\otimes[\nu])\big)^{\aS_\delta}
\\
&\stackrel{\textup{Pieri}}{=}& 
\big((\bigoplus_{\la\vdash_a \delta}\{\la\}_a\otimes[\la]) \otimes (\bigoplus_{\mu\vdash_N d\delta}\{\mu\}_N\otimes K_{\mu,d^\delta}) \otimes (\bigoplus_{\nu\vdash_b \delta}\{\nu\}_b\otimes[\nu])\big)^{\aS_\delta}
\\
&=& 
\big((\bigoplus_{\la\vdash_a \delta}\{\la\}_a\otimes[\la]) \otimes (\bigoplus_{\mu\vdash_N d\delta}\{\mu\}_N\otimes ( \bigoplus_{\kappa\vdash\delta}[\kappa]^{\oplus a_\mu(\kappa[d])} )) \otimes (\bigoplus_{\nu\vdash_b \delta}\{\nu\}_b\otimes[\nu])\big)^{\aS_\delta}
\\
&=& 
\bigoplus_{\substack{
\la\vdash_a \delta\\
\mu\vdash_N d\delta\\
\kappa\vdash\delta\\
\nu\vdash_b \delta}}
\{\la\}_a\otimes\{\mu\}_N\otimes\{\nu\}_b
\otimes
\IC^{a_\mu(\kappa[d])}
\otimes
\big(
[\la]\otimes[\kappa]\otimes[\nu]
\big)^{\aS_\delta}
\\
&=& 
\bigoplus_{\substack{
\la\vdash_a \delta\\
\mu\vdash_N d\delta\\
\kappa\vdash\delta\\
\nu\vdash_b \delta}}
\big(\{\la\}_a\otimes\{\mu\}_N\otimes\{\nu\}_b\big)^{a_\mu(\kappa[d]) k(\la,\kappa,\nu)}
\\
&=& 
\bigoplus_{\substack{
\la\vdash_a \delta\\
\mu\vdash_N d\delta\\
\nu\vdash_b \delta}}
\big(\{\la\}_a\otimes\{\mu\}_N\otimes\{\nu\}_b\big)^{\sum_{\kappa\vdash\delta}a_\mu(\kappa[d]) k(\la,\kappa,\nu)}
\end{eqnarray*}
One can check with a computer that for $\delta\leq 6$ these multiplicities
never exceed the multiplicities in \eqref{eq:orbitmultikron}.

\subsection{$\hXi_{n,d}$ versus $\hXi_{n,d,\pi}$}
Let $H = \stab_{\GL_n\times G_{n,d}\times\GL_n}\Xi_{n,d}$.
By construction, we have $H\times H \subseteq \stab(\hXi_{n,d})$.
However, this does not characterize $\hXi_{n,d}$, since each $\hXi_{n,d,\pi}$ is also stabilized by this subgroup.
Indeed,
\begin{eqnarray*}
&&(\{\bar\la\}_{n^2} \otimes \{\bar\mu\}_{n^2 d \cdot n^2 d} \otimes \{\bar\nu\}_{n^2})^{H\times H}
\\
&=&
\bigoplus_{\la,\la',\mu,\mu',\nu,\nu'} k(\bar\la,\la,\la') k(\bar\mu,\mu,\mu') k(\bar\nu,\nu,\nu')
(\{\la\}_{n} \otimes \{\la'\}_{n}  \otimes \{\mu\}_{n^2 d}\otimes \{\mu'\}_{n^2 d} \otimes \{\nu\}_{n}\otimes \{\nu'\}_{n})^{H\times H}
\\
&=&
\bigoplus_{\la,\la',\mu,\mu',\nu,\nu'} k(\bar\la,\la,\la') k(\bar\mu,\mu,\mu') k(\bar\nu,\nu,\nu')
(\{\la\}_{n} \otimes  \{\mu\}_{n^2 d}\otimes \{\nu\}_{n})^H \otimes (\{\la'\}_{n} \otimes\{\mu'\}_{n^2 d} \otimes \{\nu'\}_{n})^{H}
\end{eqnarray*}
Now, for $\bar\la=\bar\nu=(1)$ and $\bar\mu=(d)$,
the nonzero Kronecker coefficients are all $1$ and correspond to $\la=\la'=\mu=\mu'=(1)$ and $\nu=\nu'$.
Therefore,
\begin{eqnarray*}
&&(\{\bar\la\}_{n^2} \otimes \{\bar\mu\}_{n^2 d \cdot n^2 d} \otimes \{\bar\nu\}_{n^2})^{H\times H}
\\
&=&\bigoplus_{\mu\vdash d} (\{(1)\}_{n} \otimes  \{\mu\}_{n^2 d}\otimes \{(1)\}_{n})^H \otimes (\{(1)\}_{n} \otimes\{\mu\}_{n^2 d} \otimes \{(1)\}_{n})^{H}
\\
&\stackrel{c^\mu_{((1),\ldots,(1))} = \dim[\mu]}{=}& \sum_{\mu\vdash d} \dim[\mu] \cdot \dim[\mu]
\\
&=&\dim\IC[\aS_d] = d!
\end{eqnarray*}
Thus, the $\hXi_{n,d,\pi}$ are a basis of the $H\times H$-invariant space.
Each $\hXi_{n,d,\pi}$ is of course invariant under rescaling any of the $n^4(d^2-d)$ many variables that it does not involve, and we call the corresponding torus $T_\pi$. The action of $H\times H$ commutes with the action of every $T_\pi$.
Therefore,
\[
\dim(\IC^{n^2}\otimes\Sym^d\IC^{n^4 d^2}\otimes \IC^{n^2})^{H\times H\times T_\pi} = 1,
\]
in particular $\hXi_{n,d,\pi}$ is characterized by its stabilizer.

\subsection{Polystability}
We follow \cite[Prop.~2.8]{BI17}.
Define $\det_{n,1,n}(A,B,C):=\det(A^n)\det(B)\det(C^n)$.
\begin{proposition}
\label{pro:polystability}
Let $G=\{(g_{\textup{left}},g_{\textup{middle}},g_{\textup{right}})\in
\GL_n\times\GL_{n^2 d}\times \GL_n
\mid 
\det_{n,1,n}(g_{\textup{left}},g_{\textup{middle}},g_{\textup{right}})= 1 \}$.
The orbit $G\Xi_{n,d}$ is closed.
\end{proposition}
\begin{proof}
For the sake of contradiction, assume that $G\Xi_{n,d}$ is not closed.
Let $Y$ be a $G$-orbit of $\overline{G\Xi_{n,d}}\setminus G\Xi_{n,d}$ of minimal dimension.
Then $Y$ must be closed.
Since $G$ is reductive,
the Hilbert-Mumford criterion says that there exist a one-parameter subgroup $\sigma:\IC^\times \to G$ with $\lim_{t\to 0}\sigma(t)\Xi_{n,d}\in Y$.
For applying \cite[Cor.~1]{Lun75}, \cite[Cor~4.5]{Kem78}
we need to pick a reductive $R \subseteq G\cap\stab(\Xi_{n,d})$,
which in \cite{BI17} is always a torus.
For a variable $x_{i,j}^{(k)}$ let $h_{x_{i,j}^{(k)}}(\alpha) \in \GL_{n^2 d}$
be the identity matrix with a single $\alpha$ instead of a $1$ at position $(x_{i,j}^{(k)},x_{i,j}^{(k)})$.
For a basis vector $e_i$, $1\leq i \leq n$, let $h_{e_i}(\alpha)\in\GL_n$ be defined analogously,
and for a basis vector $e_j^*$, $1\leq j \leq n$, let $h_{e^*_j}(\alpha)\in\GL_n$ also be defined analogously.
All these elements of $\GL_n\times\GL_{n^2 d}\times\GL_n$ commute with each other.

Let $D=(V,E)$ be the digraph with vertex set $v_{i}^{(k)}$, $1\leq i\leq n$, $0\leq k \leq d$, where $\{v_{i}^{(k)}\mid 1\leq i \leq n\}$ for the $k$-th vertex layer,
and we have the complete bipartite edge set from vertex layer $k$ to vertex layer $k+1$, labeled with $x_{i,j}^{(k)}$.
Moreover, each vertex $v_{i}^{(0)}$ has a half-edge pointing at it, labeled with $e_i$,
and each vertex $v_{j}^{(d)}$ has a half-edge pointing away from it, labeled with $e_j^*$.
For every vertex $v\in V$, let $\textup{in}(v)$ denote the set of incoming (half-)edges,
and $\textup{out}(v)$ denote the set of outgoing (half-)edges.
For $v\in V$ and $\alpha\in\IC$ let
\[
h_{v,\alpha} := 
\prod_{e \in \textup{in}(v)} h_e(\alpha) \cdot \prod_{e \in \textup{out}(v)} h_e(\alpha^{-1})
\ \in G \cap \stab_G(\Xi_{n,d}).
\]
Let $R$ be generated by $\{h_{v,\alpha}\mid v\in V, \alpha\in\IC\}$.
Since $R$ is the image of a torus, it is a torus itself, and hence reductive.
We now determine
the centralizer of $R$ in $G$.
Let $g\in\GL_n\times\GL_{n^2 d}\times\GL_n$ commute with $h_{v,\alpha}$ for all $\alpha$, i.e., $g h_{v,\alpha} = h_{v,\alpha} g$.
We think of $\GL_n\times\GL_{n^2 d}\times\GL_n$ as embedded into $\GL_{n+n^2d+n}$ as block matrix.
Let $\textup{neu}(v) := E\setminus (\textup{in}(v) \cup\textup{out}(v))$.
The identity $\forall \alpha:g h_{v,\alpha} = h_{v,\alpha} g$
forces that $g_{e',e}$ can only by nonzero if
$\{e',e\}\subseteq\textup{in}(v)$
or
$\{e',e\}\subseteq\textup{out}(v)$
or
$\{e',e\}\subseteq\textup{neu}(v)$.
For every edge $e=(v,w)$
we have that $\textup{out}(v)\cap \textup{in}(w)$
contains only $e$,
thus $g$ has at most $1$ nonzero entry in row~$e$
and at most $1$ nonzero entry in column~$e$.
For a half-edge $e$ at a source~$s$
we have $\textup{in}(s)=\{e\}$
and hence $g$ has at most $1$ nonzero entry in row~$e$
and at most $1$ nonzero entry in column~$e$.
Analogously for half-edges at sinks.
Hence, the centralizer of $R$ contains only diagonal matrices.
The Luna-Kempf theorem says that we can assume that
\begin{equation}\label{eq:defsigmat}
\sigma(t) = \big(\diag(t^{\la_1},\ldots,t^{\la_n}), \ \diag(t^{\mu_1},\ldots,t^{\mu_N}), \ \diag(t^{\nu_1},\ldots,t^{\nu_n})\big).
\end{equation}
Interpret $\Xi_{n,d} \in \IC^n\otimes\Sym^d \IC^{n^2 d}\otimes \IC^{n\ast}$
and define $\supp(\Xi_{n,d})$ to be the set of basis tensors
of the form $e_i \otimes m \otimes e_j^*$, where $m$ is a monomial,
that have nonzero coefficient in $\Xi_{n,d}$.
Let $\textup{expvec}(e^{\alpha_{\textup{left}}} \otimes e^{\alpha_{\textup{middle}}} \otimes (e^{\alpha_{\textup{right}}})^*)=(\alpha_{\textup{left}},\alpha_{\textup{middle}},\alpha_{\textup{right}})$.
Since $\lim_{t\to 0} \sigma(t) \Xi_{n,d}$ converges, we have
\begin{equation}\label{eq:alphalamununonneg}
\forall\alpha\in\textup{expvec}(\supp(\Xi_{n,d})): \ 
\langle\alpha,(\la,\mu,\nu)\rangle\geq 0.
\end{equation}
Note that $\#\textup{expvec}(\supp(\Xi_{n,d})) = n^{d+1}$.
For each $\alpha\in\textup{expvec}(\supp(\Xi_{n,d}))$ let $c_\alpha = \frac{n^2}{n^{d+1}} = n^{-d-1}$.
We have
\begin{equation}\label{eq:sumcalphaa}
\sum_{\alpha\in\textup{expvec}(\supp(\Xi_{n,d}))} c_\alpha \alpha =
(n,\ldots,n,1,\ldots,1,n,\ldots,n)
\end{equation}
Since $\det_{n,1,n}(\sigma(t))=1$, by \eqref{eq:defsigmat} we have
\begin{equation}\label{eq:nonenlamunu}
\langle(n,\ldots,n,1,\ldots,1,n,\ldots,n),(\la,\mu,\nu)\rangle = 0.
\end{equation}
We now put the pieces together as follows.
\[
0 \stackrel{\eqref{eq:sumcalphaa},\eqref{eq:nonenlamunu}}{=} \langle\sum_{\alpha\in\textup{expvec}(\supp(\Xi_{n,d}))} c_\alpha \alpha,(\la,\mu,\nu)\rangle
  = \sum_{\alpha\in\textup{expvec}(\supp(\Xi_{n,d}))} c_{\alpha} \langle\alpha,(\la,\mu,\nu)\rangle\]
Since $c_\alpha>0$ for all $\alpha\in\textup{expvec}(\supp(\Xi_{n,d}))$,
by \eqref{eq:alphalamununonneg}
we have $\langle\alpha,(\la,\mu,\nu)\rangle=0$ for all $\alpha\in\textup{expvec}(\supp(\Xi_{n,d}))$.
This implies that $\forall t: \sigma(t)\Xi_{n,d} = \Xi_{n,d}$.
Therefore, $\lim_{t\to 0} \sigma(t)\Xi_{n,d} = \Xi_{n,d}$, which is a contradiction.
\end{proof}

\begin{proposition}
\label{pro:polystabilityhXindpi}
Let $G=\{(g_{\textup{left}},g_{\textup{middle}},g_{\textup{right}})\in
\GL_{n^2}\times\GL_{n^4 d}\times \GL_{n^2}
\mid 
\det_{n^2,1,n^2}(g_{\textup{left}},g_{\textup{middle}},g_{\textup{right}})= 1 \}$.
For all $\pi\in\aS_d$ the orbit $G\hXi_{n,d,\pi}$ is closed.
\end{proposition}
\begin{proof}
We proceed analogously to Proposition~\ref{pro:polystability}.
Fix $\pi\in\aS_d$.
For a variable $x_{i',j'}^{(\pi(k)}\boxtimes x_{i,j}^{(k)}$ let
$h_{x_{i',j'}^{\pi(k)}\boxtimes x_{i,j}^{(k)}}(\alpha)\in\GL_{n^4 d}$
be the identity matrix with a single $\alpha$ instead of a $1$
at position $(x_{i',j'}^{\pi(k)}\boxtimes x_{i,j}^{(k)},x_{i',j'}^{\pi(k)}\boxtimes x_{i,j}^{(k)})$, and analogously for $h_{e_{i'}\boxtimes e_i}(\alpha)$
and $h_{e_{j'}^*\boxtimes e_j^*}(\alpha)$.
For fixed $\alpha$ we extend $h_{.,.}(\alpha)$ bilinearly.
Let $D^2=D_1\dot\cup D_2$ denote the digraph that consists of two disjoint copies of $D$ from the proof of
Proposition~\ref{pro:polystability}.
An edge that is not a half-edge is called a \emph{full-edge}.
The vertex set of the first copy is denoted by $V_1$,
and its set of full-edges is denoted by $E_1$,
its set of left half-edges by $H^{\lefthalf}_1$, and the set of right half-edges by $H^{\righthalf}_1$.
Analogously for $V_2$, $E_2$, $H^{\lefthalf}_2$, and $H^{\righthalf}_2$.
The subset of $E_1$ of edges in edge layer $k$ is denoted by $E_1^{(k)}$,
and analogously for $E_2^{(k)}$.
For $e\in E_1^{(\pi(k))}$ let $\overline{e}=\sum_{e\in E_2^{(k)}} e$,
and for $e\in E_2^{(k)}$ let $\overline{e}=\sum_{e\in E_1^{(\pi(k))}} e$.
For $e\in H^{\lefthalf}_1$ let $\overline{e}=\sum_{e\in H^{\lefthalf}_2} e$,
and for $e\in H^{\lefthalf}_2$ let $\overline{e}=\sum_{e\in H^{\lefthalf}_1} e$.
For $e\in H^{\righthalf}_1$ let $\overline{e}=\sum_{e\in H^{\righthalf}_2} e$,
and for $e\in H^{\righthalf}_2$ let $\overline{e}=\sum_{e\in H^{\righthalf}_1} e$.
For $v\in V_1\cup V_2$ and $\alpha\in\IC$ let
\[
h_{v,\alpha} := 
\prod_{e \in \textup{in}(v)} h_{e\boxtimes \overline{e}}(\alpha) \cdot \prod_{e \in \textup{out}(v)} h_{e\boxtimes\overline{e}}(\alpha^{-1})
\ \in G \cap \stab_G(\hXi_{n,d,\pi}).
\]
Let $R$ be generated by $\{h_{v,\alpha}\mid v\in V_1\cup V_2, \alpha\in\IC\}$.
We now determine the centralizer of $R$ in~$G$.
The identity $\forall \alpha:g h_{v,\alpha} = h_{v,\alpha} g$ forces that
$g_{e_1'\boxtimes e'_2,e_1\boxtimes e_2}$
can only be nonzero if
$\{e'_1,e_1\}\subseteq\textup{in}(v)$
or
$\{e'_1,e_1\}\subseteq\textup{out}(v)$
or
$\{e'_1,e_1\}\subseteq\textup{neu}(v)$.
Analogoulsy to the proof of Proposition~\ref{pro:polystability},
this means that $g$ can only be nonzero at entries $(e_1\boxtimes e_2',e_1\boxtimes e_2)$, and analogously for half-edges.
Together with the same argument for $(e'_2,e_2)$ this shows that $g$ can only be nonzero at diagonal entries.
We proceed as in the proof of Proposition~\ref{pro:polystability},
but note that $\#\textup{expvec}(\supp(\hXi_{n,d,\pi})) = n^{2d+2}$,
so we set $c_\alpha = \frac{n^4}{n^{2d+2}} = n^{-(2d-2)}$
and obtain
\[
\sum_{\alpha\in\textup{expvec}(\supp(\hXi_{n,d,\pi}))} c_\alpha \alpha =
(n^2,\ldots,n^2,1,\ldots,1,n^2,\ldots,n^2).
\]
We have
$
\langle(n^2,\ldots,n^2,1,\ldots,1,n^2,\ldots,n^2),(\la,\mu,\nu)\rangle = 0,
$
and hence
\[
0 = \sum_{\alpha\in\textup{expvec}(\supp(\hXi_{n,d,\pi}))} c_{\alpha} \langle\alpha,(\la,\mu,\nu)\rangle.
\]
Since $c_\alpha>0$ for all $\alpha\in\textup{expvec}(\supp(\hXi_{n,d,\pi}))$,
we have $\langle\alpha,(\la,\mu,\nu)\rangle=0$ for all $\alpha\in\textup{expvec}(\supp(\hXi_{n,d,\pi}))$.
This implies that $\forall t: \sigma(t)\hXi_{n,d,\pi} = \hXi_{n,d,\pi}$.
Therefore, $\lim_{t\to 0} \sigma(t)\hXi_{n,d,\pi} = \hXi_{n,d,\pi}$, which is a contradiction.
\end{proof}

\begin{proposition}
Let $G=\{(g_{\textup{left}},g_{\textup{middle}},g_{\textup{right}})\in
\GL_{n^2}\times\GL_{n^4 d^2}\times \GL_{n^2}
\mid 
\det_{n^2 d,1,n^2 d}(g_{\textup{left}},g_{\textup{middle}},g_{\textup{right}})= 1 \}$.
The orbit $G\hXi_{n,d}$ is closed.
\end{proposition}
\begin{proof}
We proceed analogously to Proposition~\ref{pro:polystabilityhXindpi}.
For a variable $x_{i,j}^{(k)}\boxtimes x_{i',j'}^{(k')}$ let
$h_{x_{i,j}^{(k)}\boxtimes x_{i',j'}^{(k')}}(\alpha)\in\GL_{n^4 d^2}$
be the identity matrix with a single $\alpha$ instead of a $1$
at position $(x_{i,j}^{(k)}\boxtimes x_{i',j'}^{(k')},x_{i,j}^{(k)}\boxtimes x_{i',j'}^{(k')})$, and analogously for $h_{e_{i'}\boxtimes e_i}(\alpha)$
and $h_{e_{j'}^*\boxtimes e_j^*}(\alpha)$.
For fixed $\alpha$ we extend $h_{.,.}(\alpha)$ bilinearly.
Let $D^2=D_1\dot\cup D_2$ denote the digraph that consists of two disjoint copies of $D$ from the proof of
Proposition~\ref{pro:polystability}.
An edge that is not a half-edge is called a \emph{full-edge}.
The vertex set of the first copy is denoted by $V_1$,
and its set of full-edges is denoted by $E_1$,
its set of left half-edges by $H^{\lefthalf}_1$, and the set of right half-edges by $H^{\righthalf}_1$.
Analogously for $V_2$, $E_2$, $H^{\lefthalf}_2$, and $H^{\righthalf}_2$.
For $e\in E_1$ let $\overline{e}=\sum_{e\in E_2} e$,
and for $e\in E_2$ let $\overline{e}=\sum_{e\in E_1} e$.
For $e\in H^{\lefthalf}_1$ let $\overline{e}=\sum_{e\in H^{\lefthalf}_2} e$,
and for $e\in H^{\lefthalf}_2$ let $\overline{e}=\sum_{e\in H^{\lefthalf}_1} e$.
For $e\in H^{\righthalf}_1$ let $\overline{e}=\sum_{e\in H^{\righthalf}_2} e$,
and for $e\in H^{\righthalf}_2$ let $\overline{e}=\sum_{e\in H^{\righthalf}_1} e$.
For $v\in V_1\cup V_2$ and $\alpha\in\IC$ let
\[
h_{v,\alpha} := 
\prod_{e \in \textup{in}(v)} h_{e\boxtimes \overline{e}}(\alpha) \cdot \prod_{e \in \textup{out}(v)} h_{e\boxtimes\overline{e}}(\alpha^{-1})
\ \in G \cap \stab_G(\hXi_{n,d}).
\]
Let $R$ be generated by $\{h_{v,\alpha}\mid v\in V_1\cup V_2, \alpha\in\IC\}$.
We now determine the centralizer of $R$ in~$G$.
The identity $\forall \alpha:g h_{v,\alpha} = h_{v,\alpha} g$ forces that
$g_{e_1'\boxtimes e'_2,e_1\boxtimes e_2}$
can only be nonzero if
$\{e'_1,e_1\}\subseteq\textup{in}(v)$
or
$\{e'_1,e_1\}\subseteq\textup{out}(v)$
or
$\{e'_1,e_1\}\subseteq\textup{neu}(v)$.
Analogoulsy to the proof of Proposition~\ref{pro:polystability},
this means that $g$ can only be nonzero at entries $(e_1\boxtimes e_2',e_1\boxtimes e_2)$, and analogously for half-edges.
Together with the same argument for $(e'_2,e_2)$ this shows that $g$ can only be nonzero at diagonal entries.
We proceed as in the proof of Proposition~\ref{pro:polystability},
but note that $\#\textup{expvec}(\supp(\hXi_{n,d})) = d!\,n^{2d+2}$,
so we set $c_\alpha = \frac{n^4 d}{d!\,n^{2d+2}} = \frac{1}{(d-1)!\,n^{2d-2}}$
and obtain
\[
\sum_{\alpha\in\textup{expvec}(\supp(\hXi_{n,d}))} c_\alpha \alpha =
(n^2 d,\ldots,n^2 d,1,\ldots,1,n^2 d,\ldots,n^2 d).
\]
We have
$
\langle(n^2 d,\ldots,n^2 d,1,\ldots,1,n^2 d,\ldots,n^2 d),(\la,\mu,\nu)\rangle = 0,
$
and hence
\[
0 = \sum_{\alpha\in\textup{expvec}(\supp(\hXi_{n,d}))} c_{\alpha} \langle\alpha,(\la,\mu,\nu)\rangle.
\]
Since $c_\alpha>0$ for all $\alpha\in\textup{expvec}(\supp(\hXi_{n,d}))$,
we have $\langle\alpha,(\la,\mu,\nu)\rangle=0$ for all $\alpha\in\textup{expvec}(\supp(\hXi_{n,d}))$.
This implies that $\forall t: \sigma(t)\hXi_{n,d} = \hXi_{n,d}$.
Therefore, $\lim_{t\to 0} \sigma(t)\hXi_{n,d} = \hXi_{n,d}$, which is a contradiction.
\end{proof}

\subsection*{Acknowledgments}
The author greatly benefitted from discussions with Maxim van den Berg, Fulvio Gesmundo, Darij Grinberg, Gorav Jindal, JM Landsberg, Vladimir Lysikov, Jakob Moosbauer, Harold Nieuwboer, and Dimitrios Tsintsilidas.
The author highlighted the idea of width as the main computational resource on 2022-Sep-19 at the ``Algebraic Geometry with Applications to TEnsors and Secants (AGATES) Kickoff workshop'' in Warsaw. The results of this paper were first presented on 2026-May-26 at the Shonan Meeting 246 on ``Tensor Isomorphism: Algorithms, Geometry, and Applications''. The author thanks both venues for their hospitality.
Generative AI was used for literature search.

{\footnotesize
\bibliographystyle{alpha}
\newcommand{\etalchar}[1]{$^{#1}$}

}

\end{document}